\documentclass[11pt, a4paper]{article}

\usepackage{setspace}
\usepackage{bm}
\usepackage{upgreek}
\usepackage{cite}
\usepackage{rotating}
\usepackage{graphicx}
\usepackage{verbatim}
\usepackage{wrapfig}
\usepackage[margin=2.5cm]{geometry}
\usepackage{amsmath}
\usepackage{placeins}
\usepackage{url}

\usepackage{fancyhdr}
\usepackage{color}
\usepackage{amsthm}
\usepackage{amsmath, amssymb}
\usepackage{mathabx}
\usepackage{longtable}
\usepackage{lscape}
\usepackage{multicol}
\usepackage{pdfpages}
\usepackage{siunitx}
\usepackage{graphicx}
\usepackage{enumerate}

\usepackage[affil-it]{authblk}

\usepackage{caption}
\usepackage{subcaption}
\usepackage{bbm}
\usepackage{nameref}

\usepackage{tikz}
\usetikzlibrary{positioning}
\tikzset{>=stealth}
\usetikzlibrary{decorations.markings}
\usetikzlibrary{patterns,angles,quotes}
\usetikzlibrary{math}
\usetikzlibrary{calc}
\usetikzlibrary{shapes,backgrounds}
\usetikzlibrary{arrows.meta,intersections}

\usepackage{hyperref}

\usepackage{chngcntr}
\counterwithout{table}{section}

\newtheorem{all}{Theorem}[section]

\theoremstyle{plain}
\newtheorem{lemma}[all]{Lemma}
\newtheorem{thm}[all]{Theorem}

\theoremstyle{definition}

\newtheorem{rem}[all]{Remark}
\newtheorem{asptn}[all]{Assumption}

\newcommand{\BI}{{\mathbb{I}}}

\newcommand{\BP}{{\mathbb{P}}}
\newcommand{\BQ}{{\mathbb{Q}}}
\newcommand{\BR}{{\mathbb{R}}}
\newcommand{\BS}{{\mathbb{S}}}
\newcommand{\BT}{{\mathbb{T}}}

\newcommand{\BZ}{{\mathbb{Z}}}

\newcommand{\FT}{{\mathcal{F}}}

\newcommand{\dd}{{\mathrm{d}}}
\newcommand{\p}{\partial}
\newcommand{\sgn}{\mathrm{sgn}}

\newcommand{\eps}{\epsilon}

\newcommand{\com}[1]{}

\DeclareMathOperator\artanh{artanh}

\numberwithin{equation}{section}
\begin{document}
\title{Enhanced Superconductivity at a Corner for the Linear BCS Equation
}
\author[1]{Barbara Roos\thanks{barbara.roos@ist.ac.at}}
\author[1]{Robert Seiringer\thanks{robert.seiringer@ist.ac.at}}
\affil[1]{Institute of Science and Technology Austria, Am Campus 1, 3400 Klosterneuburg, Austria}

\date{\today}

\maketitle

\begin{abstract}
We consider the critical temperature for superconductivity, defined via the linear BCS equation.
We prove that at weak coupling the critical temperature for a sample confined to a quadrant in two dimensions is strictly larger than the one for a half-space, which in turn is strictly larger than the one for $\mathbb{R}^2$.
Furthermore, we prove that the relative difference of the critical temperatures vanishes in the weak coupling limit.
\end{abstract}

\paragraph*{MSC Class:} 81Q10 (Primary) 82D55 (Secondary)

\section{Introduction}
Recent work \cite{samoilenka_boundary_2020,samoilenka_microscopic_2021,barkman_elevated_2022-1,benfenati_boundary_2021,talkachov_microscopic_2022,shanenko_size-dependent_2006} predicts the occurrence of boundary superconductivity in the BCS model.
Close to edges superconductivity sets in at higher temperatures than in the bulk, and at corners the critical temperature appears to be even higher than at edges.
A first rigorous justification was provided in \cite{hainzl_boundary_2023,roos_bcs_2023}, where it was proved that the system on half-spaces in dimensions $d\in \{1,2,3\}$ can have higher critical temperatures than on $\BR^d$.
Here, we consider $d=2$ and the goal is to show that a quadrant has a higher critical temperature than a half-space.
Since comparing the critical temperatures for the non-linear BCS model is very difficult, we work with the critical temperature defined via the linear BCS equation and show that a quadrant has a higher critical temperature than a half-space, at least at weak coupling in the same spirit as in \cite{frank_bcs_2019,frank_bcs_2018}.
This may serve as a starting point for future investigations of the non-linear model.

Superconductivity is more stable close to boundaries also when a magnetic field is applied.
This phenomenon has been widely studied using Ginzburg-Landau theory, see e.g.~\cite{fournais_spectral_2010, correggi_ginzburg-landau_2024} and the references therein.
Ginzburg-Landau theory can be rigorously derived from BCS theory for domains without boundaries \cite{frank_microscopic_2012,deuchert_microscopic_2023-1,deuchert_microscopic_2023}, while for domains with boundaries this is an open problem.

We consider the full plane, and the half- and quarter-spaces $\Omega_k=(0,\infty)^k \times \BR^{2-k}$ for $k\in\{0,1,2\}$.
We define the critical temperature as in \cite{hainzl_boundary_2023,roos_bcs_2023} using the operator
\begin{equation}
H_T^{\Omega}=\frac{-\Delta_x-\Delta_y-2\mu}{\tanh\left(\frac{-\Delta_x-\mu}{2T}\right)+\tanh\left(\frac{-\Delta_y-\mu}{2T}\right)}-\lambda V(x-y)
\end{equation}
acting in $L_{\rm sym}^2(\Omega\times \Omega)=\{\psi \in L^2(\Omega\times \Omega) \vert \psi(x,y)=\psi(y,x)\ \mathrm{for\ all}\ x,y\in \Omega\}$, where $-\Delta$ denotes the Dirichlet or Neumann Laplacian and the subscript indicates on which variable it acts, $T$ is the temperature, $\mu$ is the chemical potential, $V$ is the interaction, and $\lambda$ is the coupling constant.
The first term is defined through functional calculus.
For $V\in L^t(\BR^2)$ with $t>1$, the $H_T^{\Omega_k}$ are self-adjoint operators defined via the KLMN theorem \cite[Remark 2.2]{roos_bcs_2023}.

The critical temperatures are defined as
\begin{equation}\label{tcj}
T_c^k(\lambda):= \inf\{T\in (0,\infty) \vert \inf \sigma(H_T^{\Omega_k})\geq 0\}.
\end{equation}
The operator $H_T^{\Omega_k}$ is the Hessian of the BCS functional at the normal state \cite{frank_bcs_2019}, and the linear BCS equation reads $H_T^{\Omega_k} \alpha=0$.

In particular, the system is superconducting for $T<T_c^k(\lambda)$, when the normal state is not a minimizer of the full, non-linear BCS functional.
A priori, superconductivity may also occur at temperatures $T>T_c^k(\lambda)$, either when the ground state energy of the Hessian is not monotone in the temperature, or when the normal state is a local minimum of the BCS functional, but not a global one.
For translation invariant systems with suitable interactions $V$, in particular for $\Omega_0=\BR^2$, this is not the case and the system is in the normal state if $T>T_c^0(\lambda)$.
This was proved in \cite{hainzl_bcs_2008,hainzl_bardeencooperschrieffer_2016} without the restriction to symmetric Cooper pair wave functions and is adapted for symmetric Cooper pair wave functions in \cite{roos_linear_2024}.
Hence $T_c^0$ separates the normal and the superconducting phase.
However, it remains an open question whether the same is true for $T_c^1$ and $T_c^2$.

We prove that for small enough $\lambda$, the critical temperatures defined through the linear criterion \eqref{tcj} satisfy $T_c^2(\lambda)>T_c^1(\lambda)$.
Together with the result from \cite{roos_bcs_2023}, we get the strictly decreasing sequence $T_c^2(\lambda)>T_c^1(\lambda)>T_c^0(\lambda)$ of critical temperatures at weak coupling.

Similarly to \cite[Lemma 2.3]{roos_bcs_2023}, where it was shown that $T_c^1(\lambda)\geq T_c^0(\lambda)$ for all $\lambda$, the following Lemma is relatively easy to prove.
\begin{lemma}\label{t2geqt1}
Let $\lambda,T>0$ and $V\in L^{t}(\BR^2)$ for some $t>1$.
Then $\inf \sigma(H_T^{\Omega_2})\leq \inf \sigma(H_T^{\Omega_1})$.
\end{lemma}
Its proof can be found in Section~\ref{sec:basic_properties}.
In particular it follows that for all $\lambda>0$, we have $T_c^2(\lambda)\geq T_c^1(\lambda)$.
The difficulty lies in proving a strict inequality, which the rest of the paper will be devoted to.
In order to prove $T_c^2(\lambda)> T_c^1(\lambda)$, we shall give a precise analysis of the asymptotic behavior of $H_{T_c^1(\lambda)}^{\Omega_1}$ as $\lambda\to0$.

For $\mu>0$ let $\FT:L^1(\BR^2)\to L^2(\BS^{1})$ act as the restriction of the Fourier transform to a sphere of radius $\sqrt\mu$, i.e., $\FT \psi(\omega)=\widehat \psi (\sqrt{\mu} \omega)$ and for $V\geq 0$ define $O_{\mu}=V^{1/2} \FT^\dagger \FT  V^{1/2}$ on $L^2(\BR^2)$.
The operator $O_\mu$ is compact.
For the desired asymptotic behavior of $H_{T_c^1(\lambda)}^{\Omega_1}$ we need 
that $O_\mu$ has a non-degenerate eigenvalue $e_\mu = \sup \sigma(O_\mu)>0$ at the top of its spectrum \cite{hainzl_bardeencooperschrieffer_2016,henheik_universality_2023}.

We require the following assumptions for our main result.
\begin{asptn}\label{asptn1}
Let $\mu>0$.
Assume that
\begin{enumerate}[(i)]
\item $V\in L^1(\BR^2) \cap L^{t}(\BR^2)$ for some $t>1$, \label{aspt.1}
\item $V$ is radial, $V\not \equiv 0$,\label{aspt.2}
\item $\vert \cdot \vert V \in L^1(\BR^2)$,\label{aspt.3}
\item $V\geq 0$,\label{aspt.4}
\item $e_\mu=\sup \sigma(O_\mu)$ is a non-degenerate eigenvalue. \label{aspt.5}
\end{enumerate}
\end{asptn}
\begin{rem}
Similarly to the three dimensional case discussed in \cite[Section III.B.1]{hainzl_bardeencooperschrieffer_2016}, because of rotation invariance the eigenfunctions of $O_\mu$ are given, in radial coordinates $r \equiv (|r|,\varphi)$, by $V^{1/2}(r)e^{im\varphi} J_m (\sqrt{\mu}\vert r \vert)$, where $J_m$ denote the Bessel functions. 
The corresponding eigenvalues are
\begin{equation}
e_\mu^{(m)}=\frac{1}{2\pi} \int_{\BR^2} V(r) \vert J_m (\sqrt{\mu}\vert r \vert)\vert^2 \dd r
\end{equation}
and in particular $e_\mu^{(m)}=e_\mu^{(-m)}$.
Assumption~\eqref{aspt.5} therefore means that $e_\mu=e_\mu^{(0)}$ and that all other eigenvalues $e_\mu^{(m)}$ are strictly smaller.
Hence, the eigenstate corresponding to $e_\mu$ has zero angular momentum.
Analogously to the three dimensional case, a sufficient condition for~\eqref{aspt.5} to hold is that  $\widehat{V}\geq0$.
\end{rem}
Our first main result is:
\begin{thm}\label{thm1}
Let $\mu>0$ and let $V$ satisfy Assumption~\ref{asptn1}.
Assume the same boundary conditions, either Dirichlet or Neumann, on $\Omega_1$ and $\Omega_2$.
Then there is a $\lambda_1>0$, such that for all $0<\lambda<\lambda_1$, $T_c^2(\lambda)>T_c^1(\lambda)$.
\end{thm}
\begin{rem}
The critical temperature $T_c^1(\lambda)$ is the smallest temperature $T$ satisfying $\inf \sigma(H_{T}^{\Omega_1})=0$.
Other solutions to this equation would define larger critical temperatures.
Upon inspection, the proof of Theorem~\ref{thm1} shows that for any temperature $T$ satisfying $\inf \sigma(H_{T}^{\Omega_1})=0$ the system on the quadrant is superconducting for temperatures in an interval around $T$.
\end{rem}
The second main result is that the relative difference in critical temperatures vanishes in the weak coupling limit.
\begin{thm}\label{thm2}
Let $\mu>0$ and let $V$ satisfy Assumption~\ref{asptn1}.
Assume either Dirichlet or Neumann boundary conditions on $\Omega_2$.
Then
\begin{equation}\label{eq:treldiff}
\lim_{\lambda\to0} \frac{T_c^2(\lambda)-T_c^0(\lambda)}{T_c^0(\lambda)}=0.
\end{equation}
\end{thm}
Since $T_c^2(\lambda)\geq T_c^1(\lambda)\geq T_c^0(\lambda)$, this implies $\lim_{\lambda\to0} \frac{T_c^2(\lambda)-T_c^1(\lambda)}{T_c^1(\lambda)}=0$ and  $\lim_{\lambda\to0} \frac{T_c^1(\lambda)-T_c^0(\lambda)}{T_c^0(\lambda)}=0$.
The latter was already shown in \cite{roos_bcs_2023} and we closely follow \cite{roos_bcs_2023} to prove Theorem~\ref{thm2}.

The paper is structured as follows.
In Section~\ref{sec:strategy} we explain the proof strategy for Theorem~\ref{thm1}.
Section~\ref{sec:basic_properties} contains the proofs of some basic properties of $H_T^\Omega$.
Section~\ref{sec:reg_conv} discusses the regularity and asymptotic behavior of the ground state of $H_T^{\Omega_1}$.
In Section~\ref{sec:epsto0} we prove Lemma~\ref{lea:eps_to_0}, the first key step in the proof of Theorem~\ref{thm1}.
The second key step, Lemma~\ref{lea:lambdato0} is proved in Section~\ref{sec:lambdato0}.
In Section~\ref{sec:pfthm2} we prove Theorem~\ref{thm2}.
Section~\ref{sec:aux} contains the proofs of auxiliary Lemmas.

\subsection{Proof strategy for Theorem~\ref{thm1}}\label{sec:strategy}
The proof of Theorem~\ref{thm1} is based on the variational principle.
The idea is to construct a trial state for $H_{T_c^1(\lambda)}^{\Omega_2}$ involving the ground state of $H_{T_c^1(\lambda)}^{\Omega_1}$.
However, the latter operator is translation invariant in the second component of the center of mass variable and therefore has purely essential spectrum.
To work with an operator that has eigenvalues, we fix the momentum in the translation invariant direction, and choose it in order to minimize the energy.

Let $U:L^2(\BR^2\times \BR^2)\to L^2(\BR^2\times \BR^2)$ be the unitary operator switching to relative and center of mass coordinates $r=x-y$ and $z=x+y$, i.e.~$U\psi (r,z)= \frac{1}{2}\psi((r+z)/2,(z-r)/2)$.
We shall apply $U$ to functions defined on a subset of $\Omega\subset \BR^2\times \BR^2$, by identifying $L^2(\Omega)$ with the set of functions in $L^2(\BR^2\times \BR^2)$ supported in $\Omega$.
The operator $U H_{T}^{\Omega_1} U^\dagger$, which is $H_T^{\Omega_1}$ transformed to relative and center of mass coordinates, acts on functions on $\tilde \Omega_1 \times \BR$, where $\tilde \Omega_1=\{(r,z_1)\in \BR^{3} \vert \vert r_1\vert< z_1\}$, and is translation invariant in $z_2$.
For every $q_2\in \BR$ let $H_T^1(q_2)$ be the operator obtained from $U H_{T}^{\Omega_1} U^\dagger$ by restricting to momentum $q_2$ in the $z_2$ direction.
The operator $H_T^1(q_2)$ acts in $L_{\rm s}^2(\tilde \Omega_1) =\{\psi \in L^2(\tilde \Omega_1) \vert \psi(r,z_1)=\psi(-r,z_1)\}$ and we have $\inf \sigma(H_{T_c^1(\lambda)}^{\Omega_1})=\inf_{q_2\in \BR} \inf \sigma(H_{T_c^1(\lambda)}^1(q_2))$.
We want to choose $q_2$ to be optimal. That this can be done is a consequence of the following Lemma, whose proof will be given in Section~\ref{sec:pfetaex}.
\begin{lemma}\label{lea:etaex}
Let $T,\lambda, \mu>0$ and $V\in L^t(\BR^2)$ for some $t>1$.
The function $q_2\mapsto \inf \sigma(H_T^1(q_2))$ is continuous, even and diverges to $+\infty$ as $|q_2|\to\infty$.
\end{lemma}
Therefore, the infimum is attained and we can define $\eta(\lambda)$ to be the minimal number in $[0,\infty)$ such that $\inf \sigma(H_{T_c^1(\lambda)}^1(\eta(\lambda)))=\inf \sigma(H_{T_c^1(\lambda)}^{\Omega_1})$.

Next, we shall argue that $H_{T_c^1(\lambda)}^1(\eta(\lambda))$ indeed has a ground state, at least for small enough coupling, which allows us to construct the desired trial state.
By \cite[Remark 2.5]{roos_bcs_2023}, there is a $\lambda_1>0$ such that $\inf \sigma(H^{\Omega_0}_{T_c^0(\lambda)})$ is attained at zero total momentum for $\lambda<\lambda_1$.
Let $H_T^0$ denote the operator $H^{\Omega_0}_T$ restricted to zero total momentum.
For $\lambda<\lambda_1$ the critical temperature $T_c^0(\lambda)$ is the unique temperature satisfying $\inf \sigma(H_T^0)=0$.
In the weak coupling limit both $T_c^0(\lambda)$ and $T_c^1(\lambda)$ vanish \cite{henheik_universality_2023}, \cite[Theorem 1.7]{roos_bcs_2023}.
Furthermore, at weak enough coupling $T_c^1(\lambda)>T_c^0(\lambda)$ \cite[Theorem 1.3]{roos_bcs_2023}.
In particular, there is a $\lambda_0>0$ such that for $\lambda\leq \lambda_0$ the critical temperatures satisfy $T_c^0(\lambda)<T_c^1(\lambda)<T_c^0(\lambda_1)$.

\begin{lemma}\label{lea:H1ev}
Let $\mu>0$,  let $V$ satisfy Assumption~\ref{asptn1} and let $0<\lambda\leq \lambda_0$.
Then $H_{T_c^1(\lambda)}^1(\eta(\lambda))$ has an eigenvalue at the bottom of its spectrum.
\end{lemma}
The proof of Lemma~\ref{lea:H1ev} can be found in Section~\ref{sec:pfH1ev}.
For $\lambda\leq \lambda_0$ let $\tilde \Phi_\lambda$ be the ground state of $H_{T_c^1(\lambda)}^1(\eta(\lambda))$.
In the case $\eta(\lambda)=0$, the operator $H_{T_c^1(\lambda)}^1(\eta(\lambda))$ commutes with reflections $r_2\to-r_2$ and we may assume that $\tilde \Phi_\lambda$ is even or odd under this reflection.
Irrespective of the value of $\eta(\lambda)$, we extend the function $\tilde \Phi_\lambda$ (anti)symmetrically from $\tilde \Omega_1$ to $\BR^3$, such that the extended function $\Phi_\lambda$ satisfies $\Phi_\lambda((-r_1,r_2),-z_1)=\Phi_\lambda(r,z_1)$ and $\mp \Phi_\lambda((z_1,r_2),r_1)=\Phi_\lambda(r,z_1)$, where $-/+$ corresponds to Dirichlet/Neumann boundary conditions (see Figure~\ref{fig:extension} for an illustration).
The function $\Phi_\lambda$ is the key ingredient for our trial state.
Let $\chi_{\tilde\Omega_1}$ denote multiplication by the characteristic function of $\tilde \Omega_1$; then $\tilde \Phi_\lambda=\chi_{\tilde\Omega_1} \Phi_\lambda$ .
We choose the normalization such that $\lVert V^{1/2}\chi_{\tilde\Omega_1} \Phi_\lambda\rVert_2=1$, where $V^{1/2}\psi(r,z)=V^{1/2}(r)\psi(r,z)$.
(Since $V\in L^t(\BR^2)$ for some $t>1$ and $\Phi_\lambda\in H^1(\BR^3)$, it follows by the H\"older and Sobolev inequalities that $V^{1/2} \Phi_\lambda$ is an $L^2$ function \cite{lieb_analysis_2001}.)

\begin{figure}
\centering
\begin{tikzpicture}[scale=0.3]
\tikzmath{\l=10;\x=7;\y=3;}

\draw (0,\l) node[left]{$y_1$};
\draw (\l,0) node[right]{$x_1$};
\draw (\l,\l) node[below right]{$z_1$};
\draw (\l,-\l) node[above right]{$r_1$};

\coordinate (P1) at (\x,\y);
\coordinate (P2) at (-\x,\y);
\coordinate (P3) at (\x,-\y);
\coordinate (P4) at (-\x,-\y);

\draw (P1) node[right]{$\psi(r,z)$};
\draw (P2) node[left]{$\mp \psi(r,z)$};
\draw (P3) node[right]{$\mp \psi(r,z)$};
\draw (P4) node[left]{$\psi(r,z)$};

\foreach \point in {P1,P2,P3,P4}
  \fill (\point) circle (5pt);

 \draw[black, thick, decoration={markings, mark=at position 1 with {\arrow[scale=2,>=stealth]{>}}},
        postaction={decorate}]
 (-\l,0)--(\l,0);
 \draw[black, thick, decoration={markings, mark=at position 1 with {\arrow[scale=2,>=stealth]{>}}},
        postaction={decorate}]
 (0,-\l)--(0,\l);

 \draw[black, thick, decoration={markings, mark=at position 1 with {\arrow[scale=2,>=stealth]{>}}},
        postaction={decorate}]
 (-\l,0)--(\l,0);
 \draw[black, dashed, decoration={markings, mark=at position 1 with {\arrow[scale=2,>=stealth]{>}}},
        postaction={decorate}]
 (-\l,-\l)--(\l,\l);
 \draw[black, dashed, decoration={markings, mark=at position 1 with {\arrow[scale=2,>=stealth]{>}}},
        postaction={decorate}]
 (-\l,\l)--(\l,-\l);

\draw[black, dotted]  (P1)--(P2)--(P4)--(P3)--(P1);

\end{tikzpicture}

\caption{Sketch of the (anti)symmetric extension of a function $\psi$ defined on the upper right quadrant in the $(r_1,z_1)$-coordinates. The extension is defined by mirroring along the $x_1$ and $y_1$-axes and multiplying by $- 1$ for Dirichlet boundary conditions.}
\label{fig:extension}

\end{figure}
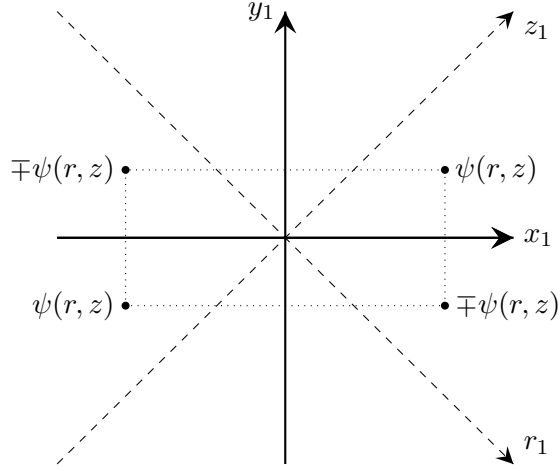

Our choice of trial state is
\begin{multline}\label{trial_state}
\psi_\lambda^\eps(r_1,r_2,z_1,z_2)=(\Phi_\lambda(r_1,r_2,z_1) e^{i\eta(\lambda) z_2}+\Phi_\lambda(r_1,-r_2,z_1) e^{-i\eta(\lambda) z_2}) e^{-\eps \vert z_2 \vert}\\
\mp(\Phi_\lambda(r_1,z_2,z_1) e^{i\eta(\lambda) r_2}+\Phi_\lambda(r_1,-z_2,z_1) e^{-i\eta(\lambda) r_2}) e^{-\eps \vert r_2 \vert}
\end{multline}
for some $\eps>0$.
Here and throughout the paper we use the convention that upper signs correspond to Dirichlet and lower signs to Neumann boundary conditions, unless stated otherwise.
The function \eqref{trial_state} is the natural generalization of the trial state for a half-space used in \cite{roos_bcs_2023}.
Note that $\psi_\lambda^\eps$ is the (anti)symmetrization of $\Phi_\lambda(r,z_1)e^{i \eta(\lambda)z_2-\eps |z_2|}$ and satisfies the boundary conditions.
The trial state vanishes if $\eta=0$ and $\Phi_\lambda$ is odd under $r_2\to-r_2$; our proof will implicitly show that at weak coupling $\Phi_\lambda$ must be even if $\eta=0$.
We shall prove the following two Lemmas in Sections~\ref{sec:epsto0} and~\ref{sec:lambdato0}, respectively.

\begin{lemma}\label{lea:eps_to_0}
Let $\mu>0$,  let $V$ satisfy Assumption~\ref{asptn1} and let $0<\lambda\leq \lambda_0$.
Then
\begin{equation}\label{eq:eps_to_0}
\lim_{\eps \to 0} \langle \psi_\lambda^\eps, UH^{\Omega_2}_{T_c^1(\lambda)}U^\dagger  \psi_\lambda^\eps \rangle =\lambda( L_1+L_2)
\end{equation}
with
\begin{multline}\label{eq:eps_to_0.2}
L_1=  \int_{\tilde \Omega_1\times \BR}  \chi_{\vert z_2\vert <\vert r_2\vert}V(r)\Bigg(\vert \Phi_\lambda(r_1,r_2,z_1)\vert^2+\vert \Phi_\lambda(r_1,z_2,z_1)\vert^2 \\
+\overline{\Phi_\lambda(r_1,r_2,z_1)}\Phi_\lambda(r_1,-r_2,z_1)e^{-2i\eta(\lambda) z_2}
+\overline{\Phi_\lambda(r_1,z_2,z_1)}\Phi_\lambda(r_1,-z_2,z_1)e^{-2i\eta(\lambda) r_2}\\
\mp \overline{\Phi_\lambda(r_1,r_2,z_1)}\Phi_\lambda(r_1,z_2,z_1)e^{i\eta(\lambda)(r_2- z_2)}
\mp \overline{\Phi_\lambda(r_1,z_2,z_1)}\Phi_\lambda(r_1,r_2,z_1)e^{-i\eta(\lambda) (r_2-z_2)}\\
\mp \overline{\Phi_\lambda(r_1,r_2,z_1)}\Phi_\lambda(r_1,-z_2,z_1)e^{-i\eta(\lambda)(r_2+ z_2)}
\mp \overline{\Phi_\lambda(r_1,z_2,z_1)}\Phi_\lambda(r_1,-r_2,z_1)e^{i\eta(\lambda) (-r_2+z_2)}
\Bigg)\dd r \dd z
\end{multline}
and
\begin{multline}\label{eq:eps_to_0.3}
L_2= - \int_{\tilde \Omega_1\times \BR}  V(r)\Bigg(\vert \Phi_\lambda(r_1,z_2,z_1)\vert^2
+\overline{\Phi_\lambda(r_1,z_2,z_1)}\Phi_\lambda(r_1,-z_2,z_1)e^{-2i\eta(\lambda) r_2}\Bigg) \dd r \dd z\\
\mp 2\pi \int_{\BR^2}\Bigg(\overline{\widehat{\Phi_\lambda}(p_1, \eta (\lambda),q_1)}
\rlap{$\phantom{V} \chi_{\tilde\Omega_1}$}\widehat {V\phantom{\chi_{\Omega}}\Phi_\lambda}
(p_1,\eta(\lambda),q_1)
+\overline{\widehat{\Phi_\lambda}(p_1, -\eta(\lambda) ,q_1)}
\rlap{$\phantom{V} \chi_{\tilde\Omega_1}$}\widehat {V\phantom{\chi_{\Omega}}\Phi_\lambda}
(p_1,-\eta(\lambda),q_1)  \Bigg)\dd p_1 \dd q_1,
\end{multline}
where $\widehat \psi(p,q_1)=\int_{\BR^3} \frac{e^{- i p \cdot r -i q_1 z_1}}{(2\pi)^{3/2}}\psi(r,z_1) \dd r \dd z_1 $ denotes the Fourier transform and $\chi_{\tilde\Omega_1}$ denotes multiplication by the characteristic function of $\tilde \Omega_1$.
\end{lemma}

\begin{lemma}\label{lea:lambdato0}
Let $\mu>0$ and let $V$ satisfy  Assumption~\ref{asptn1}.
As $\lambda\to0$ we have $L_1=O(1)$ and $L_2\leq - \frac{C}{\lambda}$ for some constant $C>0$.
\end{lemma}
In particular, there is a $\lambda_2>0$ such that for all $0<\lambda\leq \lambda_2$, $\lim_{\eps \to 0} \langle \psi_\lambda^\eps, UH^{\Omega_2}_{T_c^1(\lambda)}U^\dagger  \psi_\lambda^\eps \rangle <0$ and hence also $\inf \sigma(H^{\Omega_2}_{T_c^1(\lambda)})<0$.
The final ingredient is the continuity of $\inf \sigma(H_T^{\Omega_2}) $ in $T$, which can be proved analogously to \cite[Lemma 4.1]{roos_bcs_2023}.
For $\lambda\leq \lambda_2$ we have for $T< T_c^1(\lambda)$ by Lemma~\ref{t2geqt1} and the definition of $T_c^1$ that $\inf \sigma(H^{\Omega_2}_{T})\leq \inf \sigma(H^{\Omega_1}_{T})< 0$.
We saw that $\inf \sigma(H^{\Omega_2}_{T_c^1(\lambda)})<0$ and thus by continuity there is an $\eps>0$ such that for all $T\in(0,T_c^1(\lambda)+\eps]$ we have $\inf \sigma(H^{\Omega_2}_{T})<0$.
In particular, $T_c^2(\lambda)>T_c^1(\lambda)$.
This concludes the proof of Theorem~\ref{thm1}.

\begin{rem}
Compared to the proof of $T_c^1(\lambda)>T_c^0(\lambda)$ in \cite{roos_bcs_2023} there are two main differences and additional difficulties here.
The first difference is that $\Phi_\lambda$ here depends on $r$ and $z_1$, and not just $r$.
In particular, we need to understand the dependence and regularity of $\Phi_\lambda$ in $z_1$.
The second difference is that for the full space minimizer it was possible to prove that the optimal momentum in the translation invariant center of mass direction is zero, whereas here we have to work with the momentum $\eta(\lambda)$, which potentially is non-zero, and we need knowledge about its asymptotics for $\lambda\to0$.
As a consequence, we may have that $\Phi_\lambda(r_1,r_2,z_1) e^{i\eta(\lambda) z_2}\neq \Phi_\lambda(r_1,-r_2,z_1) e^{-i\eta(\lambda) z_2}$, which is why the expressions in Lemma~\ref{lea:eps_to_0} are twice as long as in the analogous ones in \cite[Lemma 4.3]{roos_bcs_2023}.
\end{rem}

\begin{rem}
The Assumptions~\ref{asptn1} are almost identical to the assumptions for proving $T_c^1(\lambda)>T_c^0(\lambda)$ in dimension two in \cite{roos_bcs_2023}.
Our method to compute the asymptotics of $\Phi_\lambda$ additionally requires the assumption $V\geq 0$, however.
In particular, in the proof of Lemma~\ref{etato0} we require the Birman-Schwinger operators corresponding to $H_T^{\Omega_k}$  to be self-adjoint for technical reasons. 
No such assumption is needed to determine the asymptotics of the ground state in the translation invariant case, hence we expect this assumption not to be necessary here either.
\end{rem}

\begin{rem}
We expect that our method of proof can also be applied in the three-dimensional case.
For a quarter space in $d=3$, we conjecture that similarly to the case of a half-space \cite{roos_bcs_2023}, the three-dimensional analogues of $L_1$ and $L_2$ in Lemma~\ref{lea:eps_to_0} are of equal order and converge to some finite numbers as $\lambda\to0$.
The limits of $L_1$ and $L_2$ then need to be computed to determine whether $\lim_{\lambda\to0}(L_1+L_2) <0$.
This makes the computation in three dimensions much more tedious than in two dimensions, which is why  we do not work out the details of the three-dimensional case here. Instead, we describe the intuition and the expected outcome.
In \cite{roos_bcs_2023}, the ground state on the full space  
could effectively be replaced by $\Phi_0=(\int_{\BR^3} V(r) j_3(r)^2\dd r)^{-1} j_3$, with $j_3(r)=(2\pi)^{-3/2}\int_{\BS^2} e^{i \sqrt{\mu} w \cdot r } \dd \omega$, in the limit $\lambda\to0$.
Motivated by the asymptotics of the half-space minimizer $\Phi_\lambda$ in two dimensions proved in Lemma~\ref{etato0}, we expect that as $\lambda \to 0$, $\eta(\lambda)\to 0$ and the function $\Phi_\lambda$ behaves like $\Phi_0$ in the $r$-variable, and concentrates at zero momentum in the $z_1$ direction.
A combination of the methods used in \cite{roos_bcs_2023} and the methods developed in this paper should then allow to compute the limit, and the expected  result is
\begin{equation}
\lim_{\lambda\to0} L_1= 2 \int_{\BR^4} \chi_{|z_2|<|r_2|} V(r) \vert \Phi_0(r)\mp \Phi_0(r_1,z_2,r_3)|^2 \dd r \dd z_2
\end{equation}
and
\begin{equation}
\lim_{\lambda\to0} L_2= -2 \int_{\BR^4}V(r)| \Phi_0(r_1,z_2,r_3)|^2 \dd r \dd z_2 \mp \frac{2\pi}{\mu^{1/2}}\int_{\BR^3}V(r)|\Phi_0(r)|^2 \dd r.
\end{equation}
We therefore expect $T_c^2(\lambda)>T_c^1(\lambda)$ at weak enough coupling if $V$ satisfies $\lim_{\lambda\to0} (L_1+L_2)<0$, which due to radiality of $V$ and $\Phi_0$ is the same condition as for $T_c^1(\lambda)>T_c^0(\lambda)$ in \cite[Theorem 1.3]{roos_bcs_2023}.
In \cite[Theorem 1.4 and Remark 1.5]{roos_bcs_2023} this condition on $V$ is further analyzed.
\end{rem}

\section{Basic properties of $H_T^{\Omega_1}$ and $H_T^{\Omega_2}$}\label{sec:basic_properties}
In this section we shall introduce some notation that will be useful later on, and prove Lemmas~\ref{t2geqt1}, \ref{lea:etaex} and \ref{lea:H1ev}.
The following functions will be important to describe the kinetic part of $H_T^\Omega$:
\begin{equation}\label{BT}
K_T(p,q)=\frac{p^2+q^2-2\mu}{\tanh\left(\frac{p^2-\mu}{2T}\right)+\tanh\left(\frac{q^2-\mu}{2T}\right)},\quad \text{and}\quad
B_T(p,q)=\frac{1}{K_T(p+q,p-q)}.
\end{equation}
We may write $B_{T,\mu}$ when the dependence on $\mu$ matters.
The function $K_T$ satisfies the following bounds \cite[Lemma 2.1]{hainzl_boundary_2023}.
\begin{lemma}\label{KT-Laplace}
For every $T>0$ there are constants $C_1(T,\mu),C_2(T,\mu)>0$ such that $C_1(1+p^2+q^2)\leq K_{T}(p,q)\leq C_2(1+p^2+q^2)$.
\end{lemma}
We will frequently use the following estimates for $B_T$ \cite[Eq.~(2.3)]{roos_bcs_2023}:
\begin{equation}\label{BT_bound}
B_T(p,q)\leq \frac{1}{\max\{|p^2+q^2-\mu|,2T\}} \quad {\rm and} \quad B_T(p,q)\chi_{p^2+q^2>2\mu>0} \leq \frac{C(\mu)}{1+p^2+q^2},
\end{equation}
where $C(\mu)$ depends only on $\mu$.

We use the notation $H^1_0(\Omega)$ for the Sobolev space of functions vanishing at the boundary of $\Omega$.
In the case of Dirichlet boundary conditions, the form domain corresponding to $H_T^{\Omega_k}$ is $D_k^D:=\{\psi \in H^1_0(\Omega_k\times \Omega_k) \vert \psi(x,y)=\psi(y,x)\}$.
For Neumann boundary conditions, one needs to replace the Sobolev space $H^1_0$ by $H^1$ to obtain $D_k^N$.
Let $K_{T}^\Omega$ be the kinetic term in $H_T^\Omega$.
The corresponding quadratic form acts as
\begin{equation}
\langle \psi, K_{T}^\Omega \psi \rangle = \int_{\BR^{4}}K_{T}(p,q) \left|\int_{\Omega^2}T_\Omega(x,p) T_\Omega(y,q) \psi(x,y) \dd x \dd y\right|^2\dd p\dd q,
\end{equation}
with
\begin{equation}
T_{\Omega_1}(x,p)=\frac{(e^{-i p_1 x_1}\mp e^{i p_1 x_1})  e^{-i p_2 x_2}}{2^{1/2}2\pi}, \quad {\rm and} \quad
T_{\Omega_2}(x,p)=\frac{(e^{-i p_1 x_1}\mp e^{i p_1 x_1}) ( e^{-i p_2 x_2}\mp e^{i p_2 x_2})}{4\pi}.
\end{equation}
As already mentioned in the Introduction, we shall use the convention that upper signs correspond to Dirichlet and lower signs to Neumann boundary conditions, unless stated otherwise.
We now switch to relative and center of mass coordinates $r=x-y$, $z=x+y$, $p'=(p-q)/2$ and $q'=(p+q)/2$.
Note that
\begin{equation}
T_{\Omega_1}(x,p)T_{\Omega_1}(y,p) = \frac{1}{(2\pi)^2}t(p'_1,q'_1,r_1,z_1)e^{-i ( p_2'  r_2+q_2' z_2)},
\end{equation}
where
\begin{equation}\label{t_halfspace}
t(p_1,q_1,r_1,z_1)=\frac{1}{2}\left(e^{-i(p_1 r_1+q_1 z_1)}+e^{i(p_1 r_1+q_1 z_1)}\mp e^{-i(p_1 z_1+q_1 r_1)} \mp e^{i(p_1 z_1+q_1 r_1)}\right).
\end{equation}
Therefore, conjugating the kinetic term $K_T^{\Omega_1}$ with $U$, which is the operator switching to relative and center of mass coordinates, gives
\begin{equation}\label{UK1U}
\langle \psi, U K_{T}^{\Omega_1}U^\dagger \psi \rangle = \int_{\BR^{4}}B_{T}(p',q')^{-1} \left|\int_{\tilde\Omega_1\times \BR}\frac{1}{(2\pi)^2}t(p'_1,q'_1,r_1,z_1)e^{-i ( p_2'  r_2+q_2' z_2)} \psi(r,z) \dd r \dd z\right|^2\dd p'\dd q'.
\end{equation}
The operators $H_T^1(q_2)$ defined by restricting $U H_T^{\Omega_1} U^\dagger$ to momentum $q_2$ in $z_2$-direction can thus be expressed as
\begin{equation}\label{HT1q}
\langle \psi,H_T^1(q_2) \psi \rangle =\langle \psi, K_T^1(q_2) \psi \rangle-\lambda \int_{\tilde \Omega_1}V(r)\vert \psi(r,z_1)\vert^2 \dd r\dd z_1
\end{equation}
where $\tilde \Omega_1=\{(r,z_1)\in \BR^{3} \vert \vert r_1\vert< z_1\}$ and the kinetic term $K_T^1(q_2)$  on $L_{\rm s}^2(\tilde \Omega_1)$  is given by
\begin{equation}\label{KT1q}
\langle \psi, K_T^1(q_2) \psi \rangle
=\int_{\BR^{3}}B_{T}(p,(q_1,q_2))^{-1} \left|\int_{\tilde\Omega_1}\frac{1}{(2\pi)^{3/2}}t(p_1,q_1,r_1,z_1)e^{-i p_2 r_2} \psi(r,z_1) \dd r \dd z_1\right|^2\dd p\dd q_1.
\end{equation}

It is convenient to introduce the Birman-Schwinger operators $A_T^0$ and $A_T^1$ corresponding to $H_T^{\Omega_0}$ and $H_T^{\Omega_1}$, respectively.
Let $A_T^0$ be the operator with domain $L^2(\BR^2\times \BR^2)$ restricted to functions satisfying $\psi(r,z)=\psi(-r,z)$ and given by
\begin{equation}\label{eq:defAT0}
\langle \psi,A_{T}^0 \psi \rangle=\int_{\BR^{4}}  B_{T}(p,q) |\widehat{ V^{1/2}\psi}(p,q)|^2\dd p \dd q.
\end{equation}
Define the operator $A_T^1$ on $\psi \in L_{\rm s}^2(\tilde \Omega_1 \times \BR)=\{\psi \in L^2(\tilde \Omega_1 \times \BR) \vert \psi(r,z)=\psi(-r,z)\}$ via
\begin{equation}
\langle \psi, A_{T}^1 \psi \rangle
=\int_{\BR^{4}}B_{T}(p,q) \left|\int_{\tilde\Omega_1\times \BR}\frac{1}{(2\pi)^2}t(p_1,q_1,r_1,z_1)e^{-i (p_2 r_2+q_2 z_2)} V^{1/2}(r)\psi(r,z) \dd r \dd z\right|^2\dd p\dd q.
\end{equation}
For $j\in \{0,1\}$, the operator $A_T^j$ is the Birman-Schwinger operator corresponding to $H_T^{\Omega_j}$ in relative and center of mass variables \cite[Section 6]{roos_bcs_2023}.
The Birman-Schwinger principle implies that
\[\sgn \inf \sigma(H_T^{\Omega_j}) = \sgn (1/\lambda-\sup \sigma(A_T^{j})),\]
where we use the convention that $\sgn\, 0  =0 $.

Due to translation invariance in $z_2$, for fixed momentum $q_2$ in this direction, we obtain the operators $A_T^1(q_2)$ on $\psi \in L_{\rm s}^2(\tilde \Omega_1)$ given by
\begin{equation}\label{a1q}
\langle \psi, A_{T}^1(q_2) \psi \rangle
=\int_{\BR^{3}}B_{T}(p,(q_1,q_2)) \left|\int_{\tilde\Omega_1}\frac{1}{(2\pi)^{3/2}}t(p_1,q_1,r_1,z_1) e^{-i p_2 r_2} V^{1/2}(r)\psi(r,z_1) \dd r \dd z_1\right|^2\dd p\dd q_1.
\end{equation}
The operator $A_T^1(q_2)$ is the Birman-Schwinger version of $H_T^1(q_2)$.
In particular, $H_{T_c^1(\lambda)}^1(\eta(\lambda))$ has the eigenvalue zero at the bottom of its spectrum if and only if $1/\lambda$ is the largest eigenvalue of $A_{T_c^1(\lambda)}^1(\eta(\lambda))$.

Let $\iota: L^2(\tilde \Omega_1)\to  L^2(\BR^{3})$ be the isometry
\begin{equation}\label{iota}
\iota \psi(r_1,r_2,z_1)
= \frac{1}{\sqrt{2}} (\psi(r_1,r_2,z_1) \chi_{\tilde \Omega_1}(r,z_1) +\psi(-r_1, r_2,-z_1) \chi_{\tilde \Omega_1}(-r_1,r_2,-z_1)).
\end{equation}
Using the definition of $t$ in \eqref{t_halfspace} and evenness of $V$ in $r_2$ one can rewrite \eqref{a1q} as
\begin{equation}\label{at1r}
\langle \psi, A_{T}^1(q_2) \psi \rangle=\int_{\BR^3} B_T(p,q)\left \vert \frac{1}{\sqrt{2}}( \widehat{V^{1/2}\iota \psi}(p,q_1)\mp  \widehat{ V^{1/2}\iota\psi}((q_1,p_2),p_1)) \right \vert^2 \dd p \dd q_1
\end{equation}
Let $F_2$ denote the Fourier transform in the second variable
$F_2 \psi (r, q_1)=\frac{1}{\sqrt{2\pi}} \int_{\BR} e^{-iq_1 z_1} \psi(r,z_1) \dd z_1$
and $F_1$ the Fourier transform in the first variable
$F_1 \psi (p,q)=\frac{1}{2\pi} \int_{\BR^2} e^{-ip\cdot r} \psi(r,q) \dd r.$
Define the operators $G_T(q_2)$ on $L^2(\BR^{3})$ through
\begin{equation}\label{gt}
\langle \psi, G_{T}(q_2) \psi \rangle=\int_{\BR^{3}} \overline{F_1 V^{1/2}\psi((q_1,p_2),p_1)} B_{T}(p,q) F_1 V^{1/2}\psi(p,q_1) \dd p \dd q_1.
\end{equation}
Let $A_{T}^0(q_2) $ acting on $L_{\text{s}}^2(\BR^2\times \BR)$ be given by $\langle \psi,A_{T}^0(q_2) \psi \rangle=\int_{\BR^{3}}  B_{T}(p,q) |\widehat{ V^{1/2}\psi}(p,q_1)|^2\dd p \dd q_1$.
It follows from \eqref{at1r} and $B_T(p,q)=B_T((q_1,p_2),(p_1,q_2))$ that
\begin{equation}\label{AT1_decomposition}
A_T^1(q_2)= \iota^\dagger (A_T^0(q_2)\mp F_2^\dagger   G_T(q_2) F_2 ) \iota.
\end{equation}

\subsection{Proof of Lemma~\ref{t2geqt1}}\label{sec:pft2geqt1}
\begin{proof}[Proof of Lemma~\ref{t2geqt1}]
The goal is to show that $\inf \sigma(H_T^{\Omega_2})\leq \inf \sigma(H_T^{\Omega_1})$.
We proceed analogously to the proof of \cite[Lemma 2.3]{roos_bcs_2023}.
Let $S_l$ be the shift by $l$ in the second component, i.e.~$S_l \psi (x,y)= \psi((x_1, x_2-l), (y_1,y_2-l))$.
Let $\psi$ be a function in $D_1^{D/N}$ with bounded support, for the case of Dirichlet/Neumann boundary conditions, respectively.
For $l$ big enough, $S_l \psi$ is supported on $\Omega_2\times \Omega_2$ and satisfies the boundary conditions.
The goal is to prove that $\lim_{l\to \infty} \langle S_l \psi, H^{\Omega_2}_T S_l \psi \rangle = \langle \psi, H^{\Omega_1}_T \psi \rangle$.
Then, since functions with bounded support are dense in $D_1^{D/N}$ (with respect to the Sobolev norm), the claim follows.

Note that $\langle S_l \psi,V S_l \psi \rangle = \langle \psi,V \psi \rangle$.
Let $\tilde \psi$ be the (anti-)symmetric continuation of $\psi$ from $\Omega_1\times \Omega_1$ to $\BR^2\times \BR^2$ as in Figure~\ref{fig:extension}, giving $\tilde \psi \in H^1(\BR^4)$.
Furthermore,
using symmetry of $K_T$ in $p_2$ and $q_2$ one obtains
\begin{multline}\label{2.6}
\langle S_l \psi,K_T^{\Omega_2} S_l \psi \rangle
=\frac{1}{4}\int_{\BR^{4}}\overline{\widehat{ \tilde \psi}(p,q)}K_T(p,q)\Big[\widehat{ \tilde \psi}(p,q)\mp\widehat{ \tilde \psi}((p_1,-p_2),q) e^{i 2l p_2}\mp \widehat{ \tilde \psi}(p,(q_1,-q_2))e^{i 2l q_2}\\
+\widehat{ \tilde \psi}((p_1,- p_2),(q_1,-q_2))e^{i 2l (p_2+q_2)}\Big] \dd p\dd q
\end{multline}
for $l$ big enough such that $S_l \psi$ is supported on $\Omega_2\times \Omega_2$.
The first term is exactly $\langle \psi,K_T^{\Omega_1} \psi \rangle $.
Note that by the Schwarz inequality and since $K_T(p,q)\leq C(1+p^2+q^2)$ according to Lemma~\ref{KT-Laplace}, the function
\begin{equation}
(p,q)\mapsto \overline{\widehat{ \tilde \psi}(p,q)}K_T(p,q)\widehat{ \tilde \psi}((p_1,-p_2),q)
\end{equation}
is in $L^1(\BR^{2d})$ since $\tilde \psi \in H^1(\BR^{4})$.
By the Riemann-Lebesgue Lemma, the second term in \eqref{2.6} vanishes for $l\to \infty$.
By the same argument, also the remaining terms vanish in the limit.
\end{proof}

\subsection{Proof of Lemma~\ref{lea:etaex}}\label{sec:pfetaex}
\begin{proof}[Proof of Lemma~\ref{lea:etaex}]
To prove continuity of the function $q_2\mapsto \inf \sigma(H_T^1(q_2))$, it suffices to show that for all $T>0$ and $\mu,Q_0,Q_1\in \BR$ there is a constant $C(T,\mu,Q_0,Q_1)$ such that for all $Q_0<q_2,q_2'<Q_1$ we have
\[
|B_T(p,q)^{-1}-B_T(p,(q_1,q_2'))^{-1}| \leq C(T,\mu,Q_0,Q_1) |q_2-q_2'|(1+p^2+q_1^2).
\]
The claim then follows analogously to the proof of \cite[Lemma 4.1]{roos_bcs_2023}.

We write
\[
B_T(p,q)^{-1}-B_T(p,(q_1,q_2'))^{-1}=(q_2'-q_2)f(p,q,q_2'-q_2) B_T^{-1}(p,(q_1,q_2')) B_T^{-1}(p,q),
\]
 where $f$ is defined as in the following Lemma.
\begin{lemma}\label{bdiff}
Let $T,\mu,Q_1>0$ and define the function $f:\BR^2\times \BR^2 \times \BR\to \BR$ through
\begin{equation}
f(p,q,x)=\frac{1}{x}\left(B_T(p,(q_1,q_2+x))-B_T(p,q)\right)
\end{equation}
for $x\neq 0$ and $f(p,q,0)=\p_{q_2} B_T(p,q)$.
Then $f$ is continuous and for $|q_2|<Q_1$ there is a constant $C$ depending only on $T, \mu$ and $Q_1$ such that
\begin{equation}
|f(p,q,x)|\leq \frac{C}{1+p_1^2+p_2^2+q_1^2}.
\end{equation}
\end{lemma}
The proof is provided in Section~\ref{sec:pfbdiff}.
Together with $B_T^{-1}(p,q)\leq C(1+p^2+q^2)$ (c.f~Lemma~\ref{KT-Laplace}) the desired bound on $|B_T(p,q)^{-1}-B_T(p,(q_1,q_2'))^{-1}|$ follows.

The function $q_2\to \inf \sigma(H_T^1(q_2))$ is even since $\langle \psi, H_T^1(-q_2) \psi \rangle = \langle \tilde \psi, H_T^1(q_2) \tilde \psi \rangle$, where $\tilde \psi(r,z_1)=\psi ((r_1,-r_2),z_1)$, which follows directly from the definitions of $H_T^1(q_2)$ and $K_T(q_2)$ in \eqref{HT1q} and \eqref{KT1q} using radiality of $V$ and substituting $(p_2,r_2)\to -(p_2,r_2)$.
The divergence of $\inf \sigma(H_T^1(q_2))$ as $|q_2|\to \infty$ follows since the function $B_T(p,q)^{-1}$ in $K_T(q_2)$ is bounded below by $|p^2+q^2-\mu|$, see \eqref{BT_bound}.
\end{proof}

\subsection{Proof of Lemma~\ref{lea:H1ev}}\label{sec:pfH1ev}
\begin{proof}[Proof of Lemma~\ref{lea:H1ev}]
The half-space Birman-Schwinger operator $A_{T}^1(q_2)$ for $q_2\in \BR$ can be decomposed into a term involving $A_T^0(q_2)$ and a perturbation involving $G_T(q_2)$ according to \eqref{AT1_decomposition}.
The operator $A_T^0(q_2)$ has purely essential spectrum and let $a_T^0:=\sup \sigma(A_T^0)$.

Below we shall prove that $G_T(q_2)$ is compact.
The part of the spectrum of $A_T^1$ that lies above $a_T^0$  hence consists of eigenvalues.

We first argue that $A_{T_c^1(\lambda)}^1$ has spectrum above $a_{T_c^1(\lambda)}^0$.
The Birman-Schwinger principle implies
\[
\sup \sigma(A_{T_c^1(\lambda)}^1(\eta(\lambda)))=\lambda^{-1}=a_{T_c^0(\lambda)}^0.
\]
We need to show that $a_{T_c^0(\lambda)}^0>a_{T_c^1(\lambda)}^0$.
The idea is to use that $a_T^0$ is strictly decreasing in $T$ when the supremum of $\sigma(A_T^0)$ is attained at zero total momentum and that $T_c^1(\lambda)>T_c^0(\lambda)$ at weak coupling.
At weak coupling $\lambda<\lambda_1$, $\inf \sigma(H_{T_c^0(\lambda)}^{\Omega_0})$ is attained at zero total momentum and $T_c^0$ is uniquely determined by $\inf \sigma(H^0_{T_c^0(\lambda)})=0$.
The Birman-Schwinger principle implies that the supremum of $\sigma(A_T^0)$ is attained at zero total momentum, i.e.~$a_T^0=\sup \sigma(A_T^0(0))$ for $T< T_c^0(\lambda_1)$. 
At weak enough coupling $\lambda\leq \lambda_0$ we have $T_c^0(\lambda_1)>T_c^1(\lambda)>T_c^0(\lambda)$.
Using the strict monotonicity of $a_T^0$
\[
\sup \sigma(A_{T_c^1(\lambda)}^1(\eta(\lambda)))=a_{T_c^0(\lambda)}^0> a_{T_c^1(\lambda)}^0.
\]
Hence $\lambda^{-1}$ is an eigenvalue of $A_{T_c^1(\lambda)}^1(\eta(\lambda))$ and by the Birman-Schwinger principle $H_{T_c^1(\lambda)}^1(\eta(\lambda))$ has an eigenvalue at the bottom of the spectrum.

To prove compactness of $G_T(q_2)$ defined in \eqref{gt}, we prove that its Hilbert-Schmidt norm is finite.
Writing out the Hilbert-Schmidt norm in terms of the integral kernel of $G_T(q_2)$ and carrying out the integrations over relative and center of mass coordinates, one obtains
\begin{equation}
\lVert G_T(q_2)\rVert_{\rm HS}^2 = \int_{\BR^4} \vert \widehat{V}(0,p_2-p_2')\vert^2 B_T(p,q)B_T((p_1,p_2'),q) \dd p_1 \dd q_1 \dd p_2 \dd p_2'.
\end{equation}
Using $B_T(p,q)\leq C(T,\mu)/(1+p^2+q^2)$ (c.f.~\eqref{BT_bound}) and Young's inequality, this is bounded above by
\begin{equation}\label{GT_HSbound}
C(T,\mu)^2 \left(\int_\BR \vert \widehat{V}(0,\vert p_2 \vert)\vert^{2r}\dd p_2 \right)^{1/r} \int_{\BR}\left(\int_{\BR} \left(\frac{1}{1+p_1^2+q_1^2+p_2^2}\right)^s \dd p_2 \right)^{2/s}\dd p_1 \dd q_1
\end{equation}
where $2=1/r+2/s$.
By assumption $V\in L^1\cap L^t$ for some $t>1$.
Note that $\widehat{V}$ is continuous by Riemann-Lebesgue and $\widehat{V}\in L^{t'} \cap L^\infty$ for some $t'<\infty$ by the Hausdorff-Young inequality.
In particular, due to the radiality of $V$, we can bound $\left(\int_\BR \vert \widehat{V}(0,\vert p_2 \vert) \vert^{2r}\right)^{1/r}\leq \lVert V\rVert_\infty^2+\frac{1}{2\pi}\lVert \widehat{V}\rVert_{2r}^2$, which is finite for the choice $r=t'/2$.
With this choice, we have $s>1$.
Note that $ \left(\int_{\BR}\left(\frac{1}{1+p_1^2+q_1^2+p_2^2}\right)^s \dd p_2 \right)^{2/s} =\frac{C}{(1+p_1^2+q_1^2)^{2-1/s}}$ for some constant $C$.
Hence the integral over $p_1,q_1$ in \eqref{GT_HSbound} is finite for $s>1$.
\end{proof}

\section{Regularity and asymptotic behavior of the half-space ground state}\label{sec:reg_conv}
In this section we collect regularity and convergence results for $\Phi_\lambda$ (defined in Section~\ref{sec:strategy}), which we shall use later to prove Lemmas~\ref{lea:eps_to_0} and \ref{lea:lambdato0}.
The asymptotics of $T_c^0(\lambda)$ and $T_c^1(\lambda)$ for $\lambda \to 0$ are known:
\begin{rem}\label{rem:Tetato0}
At weak enough coupling, $\inf \sigma(H_{T_c^0(\lambda)}^{\Omega_0})$ is attained at zero total momentum \cite[Remark 2.5]{roos_bcs_2023}.
In the case of zero total momentum, the asymptotics of $T_c^0(\lambda)$ were computed in \cite[Theorem 2.5]{henheik_universality_2023} to be $\vert \lambda^{-1}- e_\mu \ln \frac{\mu}{T_c^0(\lambda)}\vert=O(1)$ for $\lambda\to 0$.
Furthermore, \cite[Theorem 1.7]{roos_bcs_2023} implies that $\ln \frac{\mu}{T_c^0(\lambda)}-\ln \frac{\mu}{T_c^1(\lambda)}=o(1)$ for $\lambda \to 0$.
Therefore, $\vert \lambda^{-1}- e_\mu \ln \frac{\mu}{T_c^1(\lambda)}\vert=O(1)$ as well.
In particular, both $T_c^0(\lambda)$ and $T_c^1(\lambda)\to0$ as $ \lambda\to0$ exponentially fast.
\end{rem}
Let $\Psi_\lambda(r,z_1):=\frac{1}{\sqrt{2}} V^{1/2}(r) \Phi_\lambda(r,z_1) \chi_{|r_1|<|z_1|}$ as function on $\BR^3$.
Note that $\lVert \Psi_\lambda \rVert_2=1$ due to the symmetry under $(r_1,z_1)\to -(r_1,z_1)$ and the normalization $\lVert V^{1/2}\chi_{\tilde\Omega_1} \Phi_\lambda\rVert_2=1$.
The first convergence result describes the asymptotic behavior of $\eta(\lambda)$ and $\Psi_\lambda$ as $\lambda\to0$.
According to the Birman-Schwinger principle, $\chi_{\tilde\Omega_1}\Psi_\lambda$ is an eigenvector of $A_{T_c^1(\lambda)}(\eta(\lambda))$ corresponding to the largest eigenvalue.

Let
\begin{equation}\label{jd}
j_2(r):=\frac{1}{2\pi}\int_{\BS^{1}} e^{i \omega \cdot r  \sqrt{\mu}} \dd \omega.
\end{equation}
Due to assumptions~\ref{asptn1}\eqref{aspt.2} and \eqref{aspt.5}, the eigenvector corresponding to the largest eigenvalue $e_\mu$ of $O_\mu$ has angular momentum zero and is given by \cite{roos_bcs_2023}
\begin{equation}
\psi^0(r)=\frac{V^{1/2}(r)j_2(r)}{\Big(\int_{\BR^2}V(r')j_2(r')^2 \dd r'\Big)^{1/2}}.
\end{equation}
Let $\BP:L^2(\BR^3)\to L^2(\BR^3)$ denote the projection onto $\psi^0$ in the $r$-variable, i.e.
\[\BP \psi (r,q_1)= \psi^0(r) \int_{\BR^2} \overline{\psi^0(r')} \psi(r',q_1) \dd r'.\]
For $0\leq \beta<1$ let $\BQ_\beta$ denote the projection onto small momenta in $q_1$, i.e.
\[\BQ_\beta \psi (r,q_1)=\psi (r,q_1) \chi_{\frac{\vert q_1\vert}{\sqrt{\mu}}<\left(\frac{T_c^1(\lambda)}{\mu}\right)^\beta}.\]
Let $\BP^\perp=\BI-\BP$ and $\BQ_\beta^\perp=1-\BQ_\beta$.

Our first convergence result for the minimizer of $H_{T_c^1(\lambda)}^{\Omega_1}$ is that for $\lambda\to0$ the optimal momentum $\eta(\lambda)\to0$ and $\Psi_\lambda$ concentrates at momentum zero in $z_1$ direction and approaches $\psi^0$ in the $r$-variables.
This is made precise in the following Lemma, whose proof  can be found in Section~\ref{sec:etato0}.
\begin{lemma}\label{etato0}
Let $\mu>0$, $V$ satisfy Assumption~\ref{asptn1} and let $0\leq\beta< 1$.
For $\lambda\to0$ we have
\begin{enumerate}[(i)]
\item $\eta(\lambda)=O(T_c^1(\lambda))$ \label{i-eta}
\item $\lVert \BP^\perp F_2 \Psi_\lambda  \rVert_2^2 =O(\lambda)$ \label{i-Pphi}
\item $\lVert \BQ_\beta^\perp F_2 \Psi_\lambda  \rVert_2^2 =O(\lambda)$  \label{i-q-range}
\end{enumerate}
\end{lemma}

For a function $f$ depending on two variables we define the mixed Lebesgue norm $\lVert f \rVert_{L_i^p L_j^q}$ for $\{i,j\}=\{1,2\}$, as first taking the $L^q$-norm in the $j$-th variable and then taking the $L^p$-norm in the $i$-th variable.
The following estimate is analogous to \cite[Lemma 3.7]{roos_bcs_2023} and follows from the Cauchy-Schwarz inequality.
\begin{lemma}\label{lea:lp_prop}
Let $V\in L^1(\BR^2)$ and $\psi\in L^2(\BR^2\times \BR)$.
Then
\begin{multline}
\lVert \widehat{V^{1/2}\psi} \rVert_{L_1^\infty L_{2}^2}\leq \sup_p \left( \int_\BR \vert \widehat{V^{1/2}\psi}(p,q_1)\vert^2 \dd q_1\right)^{1/2}\\
\leq \lVert \widehat{V^{1/2}\psi} \rVert_{ L_2^2 L_1^\infty}=\left(\int_\BR\sup_p  \vert \widehat{V^{1/2}\psi}(p,q_1)\vert^2 \dd q_1\right)^{1/2}
\leq \frac{\lVert V\rVert_1^{1/2}}{2\pi} \lVert \psi\rVert_2.
\end{multline}
\end{lemma}

To simplify notation, we shall sometimes write $T_c^1,\eta$ instead of $T_c^1(\lambda),\eta(\lambda)$.
Recall the definition of $t(p_1,q_1,r_1,z_1)$ from \eqref{t_halfspace} and note that due to the (anti-)symmetry of $\Phi_\lambda$
\begin{equation}\label{pf_eps0_1}
\frac{1}{(2\pi)^{3/2}}\int_{\tilde \Omega_1} t(p_1,q_1,r_1,z_1)e^{-i p_2 r_2}\Phi_\lambda(r,z_1)\dd r \dd z_1= \frac{1}{2}\widehat{\Phi_\lambda}(p,q_1).
\end{equation}
Combining this with the eigenvalue equation  $\chi_{\tilde\Omega_1} \Phi_\lambda =\lambda (K_{T_c^1(\lambda)}^1(\eta(\lambda))^{-1} V \chi_{\tilde\Omega_1} \Phi_\lambda$ gives
\begin{equation}\label{eval_eq}
\widehat{\Phi_\lambda}(p,q_1)
=\frac {2\lambda} {(2\pi)^{3/2}} \int_{\tilde \Omega_1}B_{T_c^1(\lambda)}(p,(q_1,\eta(\lambda)))
\ t(p_1,q_1,r'_1,z_1')e^{-i p_2 r'_2} V(r') \Phi_\lambda(r',z_1') \dd r' \dd z_1'
\end{equation}
for $(p,q_1)\in \BR^3$.

To describe the asymptotics of $\Phi_\lambda$ for $\lambda\to0$, it is convenient to split the function into different summands with different asymptotic properties.
We use \eqref{eval_eq} together with $\Psi_\lambda=\frac{1}{\sqrt{2}} V^{1/2}\Phi_\lambda  \chi_{|r_1|<|z_1|}$ to split $\Phi_\lambda$ into the sum $\Phi_\lambda^{d}\mp \Phi_\lambda^{ex}$, where the first term uses the first two summands of $t(p_1,q_1,r'_1,z'_1)$
\begin{equation}\label{phi_lambda_j.1}
\Phi_\lambda^{d}(r,z_1)=\sqrt{2}\lambda\int_{\BR^{3}} \frac{e^{i (p\cdot r +q_1 z_1)}}{(2\pi)^{3/2}} B_{T_c^1}(p,(q_1,\eta)) \widehat{V^{1/2}\Psi_\lambda}(p, q_1)\dd p\dd q_1
\end{equation}
and the second term uses the last two summands of $t(p_1,q_1,r'_1,z'_1)$
\begin{equation}\label{phi_lambda_j.2}
\Phi_\lambda^{ex}(r,z_1)=\sqrt{2} \lambda\int_{\BR^{3}}\frac{e^{i (p\cdot r +q_1 z_1)}}{(2\pi)^{3/2}} B_{T_c^1}(p,(q_1,\eta)) \widehat{V^{1/2}\Psi_\lambda}(( q_1,p_2), p_1)\dd p\dd q_1.
\end{equation}
For $j\in \{d,ex\}$ we further split $\Phi^{j}_\lambda=\Phi^{{j},<}_\lambda+\Phi^{{j},>}_\lambda$, where $\Phi^{{j},\#}$ for $\#\in\{<,>\}$ has the characteristic function $\chi_{p^2+q_1^2\#2\mu}$ in the integrand.
Furthermore, let $\Phi^{\#}=\Phi^{d,\#}\mp \Phi^{ex,\#}$.

The following three Lemmas contain regularity properties for $\Phi_\lambda$, which are later used for dominated convergence arguments in the proof of Lemma~\ref{lea:eps_to_0}.
Furthermore, they also contain information about the weak coupling behavior of the different $\Phi_\lambda^{j,\#}$, which is important for the proof of Lemma~\ref{lea:lambdato0}.
The first Lemma is useful to prove that $L_1$ is of order $O(1)$.
\begin{lemma}\label{lea:phi_infty}
Let $\mu>0$, let $V$ satisfy Assumption~\ref{asptn1} and let $0<\lambda\leq \lambda_0$. Then
$
\lVert \Phi_\lambda \rVert_{L_1^\infty L_{2}^2} <\infty.
$
Furthermore, $\lVert \Phi_\lambda^d \rVert_{L_1^\infty L_{2}^2}=O(1)$ and $\lVert\Phi_{\lambda}^{ex,>}\rVert_{L_1^\infty L_{2}^2}=O(\lambda)$ as $\lambda\to 0$.
\end{lemma}
To understand the asymptotics of $L_2$ the following result comes in handy.
\begin{lemma}\label{lea:Vphi_reg.1}
Let $\mu>0$, let $V$ satisfy Assumption~\ref{asptn1} and let $0<\lambda\leq \lambda_0$.
The function $(r,z)\mapsto V^{1/2}(r) \vert \Phi_\lambda(r_1,z_2,z_1)\vert $ is in $L^2(\BR^4)$.
Furthermore, as $\lambda\to0$, the $L^2(\BR^4)$-norms of the functions $V^{1/2}(r) \vert \Phi_\lambda^{>}(r_1,z_2,z_1)\vert$, $V^{1/2}(r) \vert \Phi_\lambda^{d,<}(r_1,z_2,z_1)\vert $ and $V^{1/2}(r) \vert \Phi_\lambda^{ex,<}(r_1,z_2,z_1)\vert $ are of order $O(\lambda)$, $O(\lambda^{-1/2})$, and $O(\lambda^{1/2})$, respectively.
\end{lemma}
This suggests that the only possible origin for divergence in $L_2$ lies in contributions from $V^{1/2}(r) \vert \Phi_\lambda^{d,<}(r_1,z_2,z_1)|$.
In the proof of Lemma~\ref{lea:lambdato0} we shall show that the $L^2$ norm of this term indeed grows as $\lambda^{-1/2}$, resulting in the $1/\lambda$ divergence of $L_2$.
Furthermore, we need the following for the proof of Lemma~\ref{lea:eps_to_0}.
\begin{lemma}\label{lea:cont_bd}
Let $\mu>0$, let $V$ satisfy Assumption~\ref{asptn1} and let $0<\lambda\leq \lambda_0$.
Define the functions $g_0$, $g_+$ and $g_{-}$ on $\BR^2$ as
\begin{equation}\label{eq:defg0}
g_0(p_2,q_2):= \int_{\BR^2}\overline{\widehat{\Phi_\lambda}(p,q_1)}
\rlap{$\phantom{V} \chi_{\tilde\Omega_1}$}\widehat {V\phantom{\chi_{\Omega}}\Phi_\lambda}
(p_1,q_2,q_1)  \dd p_1 \dd q_1
\end{equation}
and
\begin{equation}
g_\pm (p_2,q_2):=\int_{\BR^2}\overline{\widehat \Phi_\lambda(p,q_1) } \Big[B_{T_c^1}^{-1}(p,q)
-B_{T_c^1}^{-1}(p,(q_1,\eta)) \Big]
\widehat \Phi_\lambda((p_1,\pm q_2),q_1) \dd p_1 \dd q_1.
\end{equation}
The functions $g_0$ and $g_\pm$ are continuous and bounded and $g_\pm (p_2,\eta)=0$ for all $p_2 \in \BR$.
\end{lemma}
The proofs of these three Lemmas are given in Sections~\ref{sec:phi_infty} -- \ref{sec:pf_cont_bd}, which may be skipped at first reading.

\subsection{Proof of Lemma~\ref{etato0}}\label{sec:etato0}
\begin{proof}[Proof of Lemma~\ref{etato0}]
Recall the operators $ A_T^0$, and $A_T^1$ from Section~\ref{sec:basic_properties} and let $a_T^j=\sup \sigma(A_T^j)$.
In the proof of \cite[Theorem 1.7]{roos_bcs_2023} it was shown that $a_T^0\leq a_T^1$ for all $T>0$.
Recall the decomposition of $A_T^1(q_2)$ into $A_T^0(q_2)$ and $G_T(q_2)$ in \eqref{AT1_decomposition}.
The operator norm of $G_T(q_2)$ is bounded uniformly in $T$ and $q_2$ according to \cite[Lemma 6.1]{roos_bcs_2023}.
Recall that  $\sqrt{2}\chi_{\tilde\Omega_1} \Psi_\lambda $ is a normalized eigenvector of $A^1_{T_c^1(\lambda)}(\eta(\lambda))$ and note that $\iota \sqrt{2}\chi_{\tilde\Omega_1} \Psi_\lambda =\Psi_\lambda $, where $\iota$ is the isometry extending a function defined on $\tilde\Omega_1$ to $\BR^3$ symmetrically under $(r_1,z_1)\to -(r_1,z_1)$, see \eqref{iota}.
With the asymptotics $T_c^1(\lambda)\to 0$ for $\lambda\to 0$ and $a_T^0=e_\mu \ln(\mu/T)+O(1)$ for $T\to 0 $ discussed in Remark~\ref{rem:Tetato0}, we have for $\lambda \to 0$
\begin{equation}\label{pf_etato0_1}
e_\mu \ln \mu/T_c^1(\lambda)+O(1)= a_{T_c^1(\lambda)}^0 \leq a_{T_c^1(\lambda)}^1 = \langle \Psi_\lambda,  A_{T_c^1(\lambda)}^0(\eta(\lambda )) \Psi_\lambda \rangle  +O(1)
\end{equation}
For $q\in \BR^2$ let $B_{T}(\cdot,q)$ denote the operator on $L^2(\BR^2)$ which acts as multiplication by $B_{T}(p,q)$ (defined in \eqref{BT}) in momentum space.
Note that
\begin{equation}
 \langle \Psi_\lambda,  A_{T_c^1(\lambda)}^0(\eta(\lambda )) \Psi_\lambda \rangle  = \int_\BR \langle F_2 \Psi_\lambda(\cdot , q_1) , V^{1/2}B_{T_c^1(\lambda)}(\cdot, (q_1,\eta(\lambda))) V^{1/2} F_2 \Psi_\lambda (\cdot, q_1) \rangle \dd q_1
\end{equation}
According to \cite[Lemma 6.8]{roos_bcs_2023}, there is a constant $C(\mu,V)$, such that for all $q\in \BR^2$ and $\psi\in L^2_{\rm s}(\BR^2)$ with $\lVert \psi \rVert_2=1$
\begin{equation}\label{supdeltaVBV}
\langle \psi, V^{1/2} B_{T}(\cdot, q) V^{1/2} \psi \rangle
\leq\langle \psi, O_\mu  \psi\rangle \ln \left(\min\left\{\frac{\sqrt{\mu}}{\vert q\vert},\frac{\mu}{T}\right\}\right)\chi_{2<\min\{\mu/T,\sqrt{\mu}/\vert q \vert\}}+ C(\mu,V).
\end{equation}
In combination, we have for $\lambda\to 0$
\begin{equation}\label{phit-aspt-est}
e_\mu \ln \mu/T_c^1(\lambda) \leq \int_{\vert q_1\vert<\sqrt{\mu}/2} \langle  F_2 \Psi_\lambda(\cdot , q_1), O_\mu  F_2 \Psi_\lambda(\cdot , q_1)\rangle  \ln \left(\min\left\{\frac{\sqrt{\mu}}{\sqrt{q_1^2+\eta(\lambda )^2}},\frac{\mu}{T_c^1(\lambda)}\right\}\right)\dd q_1 +O(1)
\end{equation}
We will use this to prove the three parts of the claim.

\eqref{i-eta}
We want to prove a bound on $\eta(\lambda)$.
Since $e_\mu=\sup \sigma(O_\mu)$, we can bound
\[
\langle   F_2 \Psi_\lambda(\cdot , q_1), O_\mu  F_2 \Psi_\lambda(\cdot , q_1)\rangle  \leq  e_\mu \lVert F_2 \Psi_\lambda(\cdot, q_1)\rVert_2^2.
\]
Moreover, clearly $\ln \left(\min\left\{\frac{\sqrt{\mu}}{\sqrt{q_1^2+\eta(\lambda )^2}},\frac{\mu}{T_c^1(\lambda)}\right\}\right) \leq\ln(\sqrt{\mu}/\eta(\lambda)).$
By \eqref{phit-aspt-est} and since $\lVert F_2\Psi_\lambda \rVert_2=1$, there is a constant $c$ such that $e_\mu \ln(\mu /T_c^1(\lambda))\leq e_\mu \ln(\sqrt{\mu}/\eta(\lambda )) +c$ for small $\lambda$.
In particular, $\vert\eta(\lambda)\vert \leq \frac{\exp(c/e_\mu)}{\sqrt{\mu}} T_c^1(\lambda)$, i.e. $\eta(\lambda )=O(T_c^1(\lambda))$.

\eqref{i-Pphi}
We want to bound $\lVert \BP^\perp F_2 \Psi_\lambda  \rVert$.
Denote the ratio of the second highest and the highest eigenvalue of $O_\mu$ by $\alpha$, where $\alpha<1$ by Assumption~\ref{asptn1}\eqref{aspt.5}.
Then
\begin{multline}
\int_\BR \langle   F_2 \Psi_\lambda(\cdot , q_1), O_\mu  F_2 \Psi_\lambda(\cdot , q_1)\rangle \dd q_1  \leq e_\mu \left( \lVert \BP  F_2 \Psi_\lambda\rVert^2 + \alpha \lVert \BP^\perp  F_2 \Psi_\lambda\rVert^2  \right)\\
=e_\mu \left(\lVert  F_2 \Psi_\lambda \rVert^2-(1-\alpha) \lVert \BP^\perp  F_2 \Psi_\lambda \rVert^2  \right)
\end{multline}
Therefore, by \eqref{phit-aspt-est}
\begin{equation}
\ln \mu/T_c^1(\lambda)\leq \left(1-(1-\alpha)\lVert \BP^\perp F_2 \Psi_\lambda\rVert^2 \right)\ln \mu/T_c^1(\lambda)  +O(1)
\end{equation}
for $\lambda \to 0$.
This means that $\lVert \BP^\perp F_2 \Psi_\lambda  \rVert^2 =O(1/\ln \mu/T_c^1(\lambda))$.
According to Remark~\ref{rem:Tetato0}, $\lim_{\lambda\to0} \lambda \ln \mu/T_c^1(\lambda)=e_\mu^{-1}$ and thus $\lVert \BP^\perp F_2 \Psi_\lambda  \rVert^2 =O(\lambda)$.

\eqref{i-q-range}
In this part, we bound $\lVert \BQ_\beta^\perp F_2 \Psi_\lambda  \rVert$.
Let
\[
\epsilon(\lambda)=\lVert \BQ_\beta^\perp F_2 \Psi_\lambda \rVert^2=\int_{\BR^3} \vert F_2 \Psi_\lambda(r,q_1)\vert^2\chi_{\vert q_1\vert>\sqrt{\mu}\left(\frac{T_c^1(\lambda)}{\mu}\right)^\beta} \dd r \dd q_1 .
\]
By \eqref{phit-aspt-est}, we have for small $\lambda$
\begin{equation}
e_\mu\ln \mu/T_c^1(\lambda) \leq  (1-\epsilon(\lambda ))e_\mu\ln \mu/T_c^1(\lambda)+\epsilon(\lambda) e_\mu \ln \frac{\mu^\beta}{T_c^1(\lambda)^\beta} +C
\end{equation}
for some constant $C$.
Hence
\begin{equation}
\epsilon(\lambda)\leq \frac{C}{(1-\beta)e_\mu \ln  \mu/T_c^1(\lambda) } = O(\lambda)
\end{equation}
where we used $\lim_{\lambda\to0} \lambda \ln \mu/T_c^1(\lambda)=e_\mu^{-1}$ (Remark~\ref{rem:Tetato0}) in the last step.
\end{proof}

\subsection{Proof of Lemma~\ref{lea:phi_infty}}\label{sec:phi_infty}
\begin{proof}[Proof of Lemma~\ref{lea:phi_infty}]
The goal is to prove $\lVert \Phi_\lambda \rVert_{L_1^\infty L_{2}^2} <\infty$, as well as $\lVert \Phi_\lambda^d \rVert_{L_1^\infty L_{2}^2}=O(1)$ and $\lVert\Phi_{\lambda}^{ex,>}\rVert_{L_1^\infty L_{2}^2}=O(\lambda)$ as $\lambda\to 0$.
If we show $\lVert \Phi_\lambda^{d} \rVert_{L_1^\infty(\BR^2) L_{2}^2(\BR)} <\infty$ and $\lVert \Phi_\lambda^{ex} \rVert_{L_1^\infty(\BR^2) L_{2}^2(\BR)} <\infty$, the Schwarz inequality implies $\lVert \Phi_\lambda \rVert_{L_1^\infty(\BR^2) L_{2}^2(\BR)} <\infty$ .

We shall first prove that $\lVert \Phi_\lambda^{d} \rVert_{L_1^\infty L_{2}^2} $ is finite and of order $O(1)$ for $\lambda\to 0$.
Using the definition of $\Phi_\lambda^d$ \eqref{phi_lambda_j.1} we have
\begin{multline}\label{l2.33}
\lVert \Phi_\lambda^d(r,\cdot)\rVert_2^2 \\
=2\lambda^2 \int_{\BR^5} \overline{\widehat{V^{1/2}\Psi_\lambda}(p',q_1)} B_{T_c^1}(p',(q_1,\eta)) \frac{e^{i (p-p')\cdot r}}{(2\pi)^2} B_{T_c^1}(p,(q_1,\eta))\widehat{V^{1/2}\Psi_\lambda}(p,q_1)\dd p \dd p' \dd q_1\\
\leq 2\lambda^2\sup_{q_1\in \BR}\sup_{\psi\in L^2(\BR^2),\lVert \psi \rVert_2=1} \int_{\BR^4} \overline{\widehat{V^{1/2}\psi}(p')} B_{T_c^1}(p',(q_1,\eta)) \frac{e^{i (p-p')\cdot r}}{(2\pi)^2} B_{T_c^1}(p,(q_1,\eta))\widehat{V^{1/2}\psi}(p)\dd p \dd p'
\end{multline}
For fixed $r$, the latter integral is the quadratic form corresponding to the projection onto the function $\phi_{q_1}(r')=\frac{1}{2\pi}F_1 B_{T_c^1}(r-r',(q_1,\eta))V^{1/2}(r')$.
Hence, taking the supremum over $\psi$, \eqref{l2.33} equals
\begin{equation}
2\lambda^2\sup_{q_1\in \BR} \lVert \phi_{q_1}\rVert_2^2=2\lambda^2\sup_{q_1\in \BR}\int_{\BR^4} \frac{e^{i (p-p')\cdot r}}{(2\pi)^3} B_{T_c^1}(p,(q_1,\eta))\widehat{V}(p-p')B_{T_c^1}(p',(q_1,\eta)) \dd p \dd p'.
\end{equation}
We split the integration into $p^2>2\mu, p^2<2\mu$ and $p'^2>2\mu, p'^2<2\mu$.
Using the upper bounds on $B_T$ stated in \eqref{BT_bound} leads to the bound
\begin{multline}\label{l2.34}
\lVert \Phi_\lambda^d(r,\cdot)\rVert_2^2\leq \frac{2 \lambda^2}{(2\pi)^3} \Bigg[\lVert \widehat{V}\rVert_\infty \sup_{q_1} \Bigg( \int_{\BR^2} B_{T_c^1}(p,(q_1,\eta)) \chi_{p^2<2\mu} \dd p \Bigg)^2\\
+ 2\sup_{q_1} \int_{\BR^4} B_{T_c^1}(p,(q_1,\eta)) \chi_{p^2<2\mu} |\widehat{V}(p-p')| \frac{C}{1+p'^2}\dd p \dd p'\\
+ \int_{\BR^4}\frac{C}{1+p^2} |\widehat{V}(p-p')| \frac{C}{1+p'^2}\dd p \dd p'\Bigg]
\end{multline}
for a constant $C$ independent of $\lambda$.
We start by considering the first term in the square bracket.
Note that $\lVert \widehat{V}\rVert_\infty<\frac{\lVert V\rVert_1}{2\pi}<\infty$.
For fixed $T>0$, the function $ B_{T}(p,q)$ is bounded, hence the term is finite for fixed $\lambda$.
For $T\to0$ we have
\begin{equation}\label{l2.22}
\sup_{q\in \BR^2}  \int_{\BR^2} B_{T}(p,q) \chi_{p^2<2\mu} \dd p =O(\ln \mu/T).
\end{equation}
To see this, we first apply the inequality \cite[(6.1)]{hainzl_boundary_2023}
\begin{equation}\label{BT-ineq}
B_T(p,q)\leq \frac{1}{2}(B_T(p+q,0)+B_T(p-q,0)).
\end{equation}
This gives the upper bound $\sup_{q\in \BR^2}  \int_{\BR^2} B_{T}(p,0) \chi_{(p-q)^2<2\mu} \dd p $.
The vector $q$ shifts the disk-shaped domain of integration, but does not change its size.
In particular, the contribution with $p^2<2\mu$ is bounded above by $ \int_{\BR^2} B_{T}(p,0) \chi_{p^2<2\mu} \dd p =O(\ln \mu/T) $ \cite[Proposition 3.1]{henheik_universality_2023} while the contribution with $p^2>2\mu$ is uniformly bounded in $T$ since $B_T(p,0) \chi_{p^2>2\mu}\leq C(\mu)/(1+p^2)$ by \eqref{BT_bound}.
Since for $\lambda\to 0$ we have $\ln \mu/T_c^1(\lambda)=O(1/\lambda)$ by Remark~\ref{rem:Tetato0}, the first term in the square bracket in \eqref{l2.34} is of order $1/\lambda^2$ as $\lambda\to0$.
For the second term in the square bracket we use H\"older's inequality in $p'$.
By assumption, $V$ is in $L^{t}(\BR^2)$ for some $t>0$, thus by the Hausdorff-Young inequality we have $\widehat{V}\in L^{t'}$ where $1=1/t'+1/t$.
Hence, the second term is bounded by
\begin{equation}
2\sup_{q_1} \int_{\BR^4} B_{T_c^1}(p,(q_1,\eta)) \chi_{p^2<2\mu}\dd p  \lVert\widehat{V}\rVert_{t'} \left \lVert \frac{C}{1+|\cdot|^2}\right \rVert_{L^{t}(\BR^2)},
\end{equation}
which is finite for fixed $\lambda$ and of order $O(1/\lambda)$ for $\lambda\to0$ by \eqref{l2.22}.
Using Young's inequality, one sees that the third term in the square bracket is bounded.
Taking into account the factor $\lambda^2$ in front of the square bracket, we conclude that $\lVert \Phi_\lambda^d(r,\cdot)\rVert_2^2=O(1)$ uniformly in $r$.

We shall now show that for fixed $\lambda$, $\lVert \Phi^{ex}_\lambda\rVert_{L_1^\infty L_{2}^2}<\infty$ and $\lVert \Phi^{ex,>}_\lambda\rVert_{L_1^\infty L_{2}^2} =O(\lambda)$ as $\lambda\to 0$.
We have
\begin{multline}\label{phi_lambda_j.3}
\lVert \Phi_\lambda^{ex}(r,\cdot)\rVert_2^2=2\lambda^2 \int_{\BR^{2d+1}}\overline{ \widehat{V^{1/2}\Psi_\lambda}(( q_1,p_2'), p_1')}B_{T_c^1}(p',(q_1,\eta)) \frac{e^{i (p-p')\cdot r }}{(2\pi)^{d}} B_{T_c^1}(p,(q_1,\eta))\\
\times \widehat{V^{1/2}\Psi_\lambda}(( q_1,p_2), p_1) \dd p \dd p' \dd q_1
\end{multline}
Similarly, we get an expression for $\lVert \Phi^{ex,>}_\lambda(r,\cdot)\rVert_2^2$ if we multiply the above integrand by the characteristic functions $\chi_{p^2+q_1^2>2\mu} \chi_{p'^2+q_1^2>2\mu}$.
Using the bounds for $B_T$ in \eqref{BT_bound},  we bound $\lVert \Phi_\lambda^{ex}\rVert_{L_1^\infty L_{2}^2}^2$ and $\lVert \Phi^{ex,>}_\lambda\rVert_{L_1^\infty L_{2}^2}^2$ above by
\begin{equation}
C \lambda^2 \int_{\BR^{2d+1}}\vert \overline{ \widehat{V^{1/2}\Psi_\lambda}(( q_1,p_2'), p_1')}\vert\frac{1}{1+p'^2+q_1^2}\frac{1}{1+p^2+q_1^2}\vert \widehat{V^{1/2}\Psi_\lambda}((q_1,p_2),p_1)\vert \dd p \dd p' \dd q_1
\end{equation}
where the constant $C$ depends on $\mu$ and $\lambda$ for the bound on $\lVert \Phi^{ex}_\lambda\rVert_{L_1^\infty L_{2}^2}^2$, but is independent of $\lambda$ for the bound on $\lVert \Phi^{ex,>}_\lambda\rVert_{L_1^\infty L_{2}^2}^2$.
Using the Schwarz inequality in $p_1$ and $p_1'$ and then the bound on the mixed Lebesgue norm in Lemma~\ref{lea:lp_prop} we get the upper bound
\begin{multline}\label{phi_lambda_j.4}
C\lambda^2 \lVert \widehat{V^{1/2}\Psi_\lambda} \rVert_{L_1^\infty L_{2}^2}^2 \int_{\BR^{2d+1}}\left(\int_\BR \frac{1}{(1+p'^2+q_1^2)^2}\dd p'_1\right)^{1/2}\left(\int_\BR \frac{1}{(1+p^2+q_1^2)^2}\dd p_1\right)^{1/2} \dd p_2 \dd p'_2 \dd q_1\\
\leq \tilde C \lambda^2 \lVert V\rVert_1 \lVert \Psi_\lambda \rVert_2^2
\end{multline}
Therefore, $\lVert \Phi^{ex}_\lambda\rVert_{L_1^\infty L_{2}^2}$ is finite and $\lVert \Phi^{ex,>}_\lambda\rVert_{L_1^\infty L_{2}^2}=O(\lambda)$.
\end{proof}

\subsection{Proof of Lemma~\ref{lea:Vphi_reg.1}}
\begin{proof}[Proof of Lemma~\ref{lea:Vphi_reg.1}]
The goal is to show that the function $(r,z)\mapsto V^{1/2}(r) \vert \Phi_\lambda(r_1,z_2,z_1)\vert $ is in $L^2(\BR^4)$ and that for $\lambda \to 0$ the $L^2(\BR^4)$-norms of the functions $V^{1/2}(r) \vert \Phi_\lambda^{>}(r_1,z_2,z_1)\vert$, $V^{1/2}(r) \vert \Phi_\lambda^{d,<}(r_1,z_2,z_1)\vert $ and $V^{1/2}(r) \vert \Phi_\lambda^{ex,<}(r_1,z_2,z_1)\vert $ are of order $O(\lambda)$, $O(\lambda^{-1/2})$, and $O(\lambda^{1/2})$, respectively.

By the Schwarz inequality, it suffices to prove that for $j\in\{d,ex\}$ and $\#\in\{<,>\}$ the integrals $\int_{\BR^4} V(r) \vert \Phi_\lambda^{j,\#} (r_1,z_2,z_1)\vert^2 \dd r \dd z$  are finite for all $\lambda_0\geq \lambda>0$ and that as $\lambda\to0$ we have $\int_{\BR^4} V(r) \vert \Phi_\lambda^{j,>}(r_1,z_2,z_1)\vert^2 \dd r \dd z=O(\lambda^2)$ for $j\in \{d,ex\}$, $\int_{\BR^4} V(r) \vert \Phi_\lambda^{d,<}(r_1,z_2,z_1)\vert^2 \dd r \dd z=O(\lambda^{-1})$ and $\int_{\BR^4} V(r) \vert \Phi_\lambda^{ex,<}(r_1,z_2,z_1)\vert^2 \dd r \dd z=O(\lambda)$.

Using the definitions of the different $\Phi_\lambda^{j,\#}$ (see \eqref{phi_lambda_j.1} and \eqref{phi_lambda_j.2}) one can rewrite for $\#\in \{<,>\}$
\begin{multline}\label{L3_12.6}
\int_{\BR^4}V(r)|\Phi_\lambda^{d,\#}(r_1,z_2,z_1) |^2 \dd r \dd z=
2\lambda^2 \int_{\BR^4} \widehat{V}(p_1-p_1',0) B_{T_c^1}((p_1',p_2),(q_1,\eta))\overline{\widehat{ V^{1/2}\Psi_\lambda}( p_1',p_2, q_1) } \\
\times B_{T_c^1}(p,(q_1,\eta))\widehat{ V^{1/2}\Psi_\lambda}(p,q_1) \chi_{p^2+q_1^2\#2\mu}\chi_{p_1'^2+p_2^2+q_1^2\#2\mu} \dd p_1 \dd p_1' \dd p_2 \dd q_1
\end{multline}
and
\begin{multline}\label{L3_12.8}
\int_{\BR^4}V(r)|\Phi_\lambda^{ex,\#}(r_1,z_2,z_1) |^2 \dd r \dd z
=2 \lambda^2 \int_{\BR^4} \widehat{V}(p_1-p_1',0) B_{T_c^1}((p_1',p_2),(q_1,\eta))\overline{\widehat{ V^{1/2}\Psi_\lambda}(q_1,p_2,p_1')}  \\
\times B_{T_c^1}(p,(q_1,\eta))\widehat{ V^{1/2}\Psi_\lambda}(q_1,p_2, p_1) \chi_{p^2+q_1^2\#2\mu}\chi_{p_1'^2+p_2^2+q_1^2\#2\mu} \dd p_1 \dd p_1' \dd p_2 \dd q_1.
\end{multline}
For $\Phi_\lambda^{d,>}$, with the aid of the bound on $B_T$ in \eqref{BT_bound} and the estimate for mixed Lebesgue norms in Lemma~\ref{lea:lp_prop} the expression is bounded by
\begin{multline}
C \lambda^2 \lVert V \rVert_1 \int_{\BR^4} \frac{1}{1+p_1'^2+p_2^2}\frac{1}{1+p_1^2+p_2^2} \lVert \widehat{V^{1/2}\Psi_\lambda}(\cdot, q_1)\rVert_\infty^2  \dd q_1 \dd p_1' \dd p_1 \dd p_2 \\
\leq \tilde C \lambda^2 \lVert V \rVert_1^2 \lVert  \Psi_\lambda\rVert_2^2  <\infty
\end{multline}
where the constants $C,\tilde C$ depend only on $\mu$.
For $\Phi_\lambda^{ex,>}$  we use the bound on $B_T$ in \eqref{BT_bound} and the Schwarz inequality in $p_1$ and $p_1'$ to bound \eqref{L3_12.8} by
\begin{equation}
C \lambda^2 \lVert V\rVert_1\int_{\BR^2} \left \lVert \frac{1}{1+|\cdot |^2+p_2^2+q_1^2}\right\rVert_{L^2(\BR)}^2 \dd p_2 \dd q_1 \lVert \widehat{V^{1/2}\Psi_\lambda}\rVert_{L_p^\infty L_q^2}^2 \leq \tilde C  \lambda^2 \lVert V\rVert_1^2\lVert\Psi_\lambda \rVert_2^2
\end{equation}
where we used the estimate for mixed Lebesgue norms from Lemma~\ref{lea:lp_prop} in the second step.
Again, the constants $C,\tilde C$ depend only on $\mu$.

For $\Phi_\lambda^{d,<}$ we bound \eqref{L3_12.6} above by
\begin{multline}\label{phid<_asy}
\frac{\lVert V\rVert_1}{\pi} \lambda^2 \int_{\BR^4} B_{T_c^1}(p,(q_1,\eta))B_{T_c^1}((p_1',p_2),(q_1,\eta))\lVert \widehat{V^{1/2}\Psi_\lambda} (\cdot,q_1)\rVert_\infty^2  \chi_{p^2+q_1^2<2\mu}\chi_{p_1'^2+p_2^2+q_1^2<2\mu}\dd p \dd p_1' \dd q_1\\
\leq \frac{\lVert V\rVert_1^2}{4\pi^3} \lambda^2 \sup_{q_1\in \BR} \int_{\BR^3} B_{T_c^1}(p,(q_1,\eta))B_{T_c^1}((p_1',p_2),(q_1,\eta)) \chi_{p^2+q_1^2<2\mu}\chi_{p_1'^2+p_2^2+q_1^2<2\mu}\dd p \dd p_1'
\end{multline}
where we used the bound on mixed Lebesgue norms from Lemma~\ref{lea:lp_prop} and $\lVert \Psi_\lambda\rVert_2=1$ in the second step.
For fixed $\lambda$ this is finite because $B_{T_c^1}$ is a bounded function.
For $\lambda\to 0$ the first part of the following Lemma together with the weak coupling asymptotics of $T_c^1$ stated in Remark~\ref{rem:Tetato0} imply that this is of order $O(\lambda^{-1})$.
\begin{lemma}\label{lea:BB_asy}
Let $\mu,C>0$. For $T\to 0$ we have
\begin{equation}
 \sup_{q,q'\in \BR^2}\int_{\BR^3} B_{T}(p,q)
B_{T}((p'_1,p_2),q')\dd p_1 \dd p_1' \dd p_2
=O(\ln \mu/T)^3.
\end{equation}
Furthermore, for every $0<\delta_1<\mu$ there is a $\delta_2>0$ such that for $T\to 0$
\begin{multline}
 \sup_{|q|,|q'|<\delta_2}\int_{\BR^3}(1-\chi_{\mu-\delta_1<p_2^2<\mu+\delta_1}\chi_{p_1^2<4\delta_1}\chi_{p_1'^2<4\delta_1}) B_{T}(p,q)
B_{T}((p'_1,p_2),q') \dd p_1 \dd p_1' \dd p_2\\
=O(\ln \mu/T)^{5/2}.
\end{multline}
\end{lemma}
The second part of this Lemma will be used in the proof of Lemma~\ref{lea:lambdato0} to compute the asymptotics of $L_2$.
The proof of Lemma~\ref{lea:BB_asy} can be found in Section~\ref{pf_lea_BB_asy}.

For $\Phi_\lambda^{ex,<}$ we bound \eqref{L3_12.8} above using the bound on mixed Lebesgue norms in Lemma~\ref{lea:lp_prop} and $\lVert \Psi_\lambda\rVert_2=1$, which gives
\begin{equation}\label{L3_12.11}
\frac{\lambda^2}{2\pi^2} \lVert V\rVert_1^2 \lVert B_{T_c^1}^{{ ex}, 2}(\eta)\rVert
\end{equation}
where $B_{T}^{{ ex}, 2}(\xi)$ is the operator acting on $L^2(-\sqrt{2\mu},\sqrt{2\mu})$ with integral kernel
\begin{equation}
 B_{T}^{{ ex}, 2}(\xi)(p_1',p_1)=\int_{\BR^2} B_T((p_1',p_2),(q_1,\xi)) B_T(p,(q_1,\xi)) \chi_{q_1^2+p_2^2<2\mu}\dd q_1 \dd p_2.
\end{equation}
The superscript 2 indicates that there are two factors of $B_T$, as opposed to $B_T^{ex}$ which is defined later in \eqref{btex}.
The following Lemma together with the asymptotics of $T_c^1(\lambda)$ from Remark~\ref{rem:Tetato0} and the fact that $\eta(\lambda)=O(T_c^1(\lambda))$ (see Lemma~\ref{etato0}\eqref{i-eta}) implies that \eqref{L3_12.11} is bounded for fixed $\lambda$ and of order $O(\lambda)$ for $\lambda\to 0$.
\begin{lemma}\label{lea:phiex_norm}
Let $c,\mu>0$.
Then $\sup_{|\xi|<c T}\lVert  B_{T}^{{ ex}, 2}(\xi)\rVert$
is finite for all $T>0$ and of order $O(\ln \mu /T)$ as $T\to 0$.
\end{lemma}
The proof of Lemma~\ref{lea:phiex_norm} is given in Section~\ref{sec:pf_phiex_norm}.
\end{proof}

\subsection{Proof of Lemma~\ref{lea:cont_bd}}\label{sec:pf_cont_bd}
\begin{proof}[Proof of Lemma~\ref{lea:cont_bd}]
Recall the functions $g_0$, $g_+$ and $g_{-}$ on $\BR^2$ defined as
\begin{equation}\label{eq:defg0.2}
g_0(p_2,q_2):= \int_{\BR^2}\overline{\widehat{\Phi_\lambda}(p,q_1)}
\rlap{$\phantom{V} \chi_{\tilde\Omega_1}$}\widehat {V\phantom{\chi_{\Omega}}\Phi_\lambda}
(p_1,q_2,q_1)  \dd p_1 \dd q_1
\end{equation}
and
\begin{equation}
g_\pm (p_2,q_2):=\int_{\BR^2}\overline{\widehat \Phi_\lambda(p,q_1) } \Big[B_{T_c^1}^{-1}(p,q)
-B_{T_c^1}^{-1}(p,(q_1,\eta)) \Big]
\widehat \Phi_\lambda((p_1,\pm q_2),q_1) \dd p_1 \dd q_1.
\end{equation}
We aim to prove that the functions $g_0$ and $g_\pm$ are continuous and bounded and $g_\pm (p_2,\eta)=0$ for all $p_2 \in \BR$.

For functions $\psi$ on $\BR^3$ let $S\psi(p_1,p_2,q_1)=\psi(p,q_1)+\psi(-p_1,p_2,-q_1)\mp \psi(q_1,p_2,p_1)\mp \psi(-q_1,p_2,-p_1)$.
For $p,q\in\BR^2$ let
\begin{align}
L^0(p,q)&:=\lambda B_{T_c^1}(p,(q_1,\eta)) ,\\
L^{\pm}(p,q)&:=\lambda^2 B_{T_c^1}(p,(q_1,\eta))\Big[B_{T_c^1}^{-1}(p,q)-B_{T_c^1}^{-1}(p,(q_1,\eta)) \Big] B_{T_c^1}((p_1,\pm q_2),(q_1,\eta)
\end{align}
Using the expression for $\widehat{\Phi_\lambda}$ in  \eqref{eval_eq} obtained from the eigenvalue equation we have
\begin{equation}\label{cont_bd.0}
g_0(p_2,q_2)=\int_{\BR^2}\overline{S
\rlap{$\phantom{V} \chi_{\tilde\Omega_1}$}\widehat {V\phantom{\chi_{\Omega}}\Phi_\lambda}
(p,q_1)}L^0(p,q)
\rlap{$\phantom{V} \chi_{\tilde\Omega_1}$}\widehat {V\phantom{\chi_{\Omega}}\Phi_\lambda}
(p_1,q_2,q_1) \dd p_1 \dd q_1
\end{equation}
and
\begin{equation}\label{cont_bd.00}
g_{\pm}(p_2,q_2)=\int_{\BR^2}\overline{S
\rlap{$\phantom{V} \chi_{\tilde\Omega_1}$}\widehat {V\phantom{\chi_{\Omega}}\Phi_\lambda}
(p,q_1)}L^{\pm}(p,q)S
\rlap{$\phantom{V} \chi_{\tilde\Omega_1}$}\widehat {V\phantom{\chi_{\Omega}}\Phi_\lambda}
(p_1,\pm q_2 ,q_1) \dd p_1 \dd q_1 .
\end{equation}
Note that $g_\pm (p_2,\eta)=0$ since $L^\pm(p,(q_1,\eta))=0$.
For measurable functions $\psi_1,\psi_2$ on $\BR^3$ and $p_2,q_2\in \BR$ we obtain using the Schwarz inequality in $q_1$
\begin{multline}\label{cont_bd.1}
\int_{\BR^2} |\psi_1(p_1,p_2,q_1)|\frac{1}{1+p_1^2} |\psi_2(p_1,q_2,q_1)| \dd p_1 \dd q_1\\
\leq \int_{\BR}\frac{1}{1+p_1^2}  \dd p_1 \sup_{p\in \BR^2}\lVert \psi_1(p,\cdot)\rVert_{ L^2(\BR)} \sup_{p\in \BR^2}\lVert \psi_2(p,\cdot)\rVert_{ L^2(\BR)}
\end{multline}
and using the Schwarz inequality in $q_1,p_1$
\begin{multline}\label{cont_bd.2}
\int_{\BR^2} |\psi_1(p_1,p_2,q_1)|\frac{1}{1+p_1^2+q_1^2} |\psi_2(q_1,q_2,p_1)| \dd p_1 \dd q_1 \\
\leq   \int_{\BR}\frac{1}{1+p_1^2}  \dd p_1 \sup_{p\in \BR^2}\lVert \psi_1(p,\cdot)\rVert_{ L^2(\BR)} \sup_{p\in \BR^2}\lVert \psi_2(p,\cdot)\rVert_{ L^2(\BR)}.
\end{multline}
There is a constant $C$ independent of $p,q$ (but dependent on $\lambda$) such that $L^0(p,q)\leq \frac{C}{1+p_1^2+q_1^2}$ by \eqref{BT_bound}.
Similarly, the bounds on $B_T$ in \eqref{BT_bound} and Lemma~\ref{KT-Laplace} imply that there is a constant $C$ independent of $p,q$ but dependent on $\lambda$  such that
\begin{equation}\label{cont_bd.3}
L^\pm(p,q)\leq \frac{C(1+p^2+q^2)}{(1+p^2+q_1^2)(1+p_1^2+q^2)}\leq \frac{2C}{1+p_1^2+q_1^2}
\end{equation}
It follows from \eqref{cont_bd.1} and \eqref{cont_bd.2} that there is a constant $C$ such that for all measurable functions $\psi_1,\psi_2$ on $\BR^3$ and $p_2,p_2',q_2,q_2'\in \BR$
\begin{equation}
\Big \vert \int_{\BR^2} \overline{S\psi_1(p,q_1)} L^0(p_1,p_2', q_1, q_2') \psi_2(p_1,q_2,q_1) \dd p_1 \dd q_1 \Big \vert\leq  C\sup_{p\in \BR^2}\lVert \psi_1(p,\cdot)\rVert_{ L^2(\BR)} \sup_{p\in \BR^2}\lVert \psi_2(p,\cdot)\rVert_{ L^2(\BR)},
\end{equation}
and similarly
\begin{equation}\label{cont_bd.4}
\Big \vert \int_{\BR^2} \overline{S\psi_1(p,q_1)} L^{\pm}(p_1,p_2', q_1, q_2') S\psi_2(p_1,\pm q_2,q_1) \dd p_1 \dd q_1 \Big \vert\leq  C\sup_{p\in \BR^2}\lVert \psi_1(p,\cdot)\rVert_{ L^2(\BR)} \sup_{p\in \BR^2}\lVert \psi_2(p,\cdot)\rVert_{ L^2(\BR)}.
\end{equation}
In particular it follows from \eqref{cont_bd.0} and \eqref{cont_bd.00} with the mixed Lebesgue norm bounds in Lemma~\ref{lea:lp_prop} and the normalization $\lVert V^{1/2}\chi_{\tilde\Omega_1}\Phi_\lambda\rVert_2=1$ that $g_0$ and $g_\pm$ are bounded.

To prove continuity, first note that
\begin{equation}
\rlap{$\phantom{V} \chi_{\tilde\Omega_1}$}\widehat {V\phantom{\chi_{\Omega}}\Phi_\lambda}
(p_1,p_2+\eps,q_1)-
\rlap{$\phantom{V} \chi_{\tilde\Omega_1}$}\widehat {V\phantom{\chi_{\Omega}}\Phi_\lambda}
(p, q_1)=
\rlap{$\phantom{W_\epsilon} \chi_{\tilde\Omega_1}$}\widehat {W_\epsilon\phantom{\chi_{\Omega}}\Phi_\lambda}
(p,q_1)
\end{equation}
where $W_\eps (r)=V(r)(e^{-i \eps r_2}-1)$.
We only spell out the proof for $g_\pm$, the argument for $g_0$ is analogous.
For all $p_2,q_2\in \BR$ we have
\begin{multline}\label{cont_bd.5}
g_{\pm}(p_2+\eps,q_2+\eps')-g_\pm(p_2,q_2)\\
= \int_{\BR^2} \overline{S
\rlap{$\phantom{V} \chi_{\tilde\Omega_1}$}\widehat {V\phantom{\chi_{\Omega}}\Phi_\lambda}
(p_1,p_2+\eps,q_1)} L^\pm(p_1,p_2+\eps, q_1, q_2+\eps') S
\rlap{$\phantom{W_\epsilon} \chi_{\tilde\Omega_1}$}\widehat {W_\epsilon\phantom{\chi_{\Omega}}\Phi_\lambda}
(p_1,\pm q_2,q_1) \dd p_1 \dd q_1\\
+ \int_{\BR^2} \overline{S
\rlap{$\phantom{W_\epsilon} \chi_{\tilde\Omega_1}$}\widehat {W_\epsilon\phantom{\chi_{\Omega}}\Phi_\lambda}
(p,q_1)} L^\pm(p_1,p_2+\eps, q_1, q_2+\eps') S
\rlap{$\phantom{V} \chi_{\tilde\Omega_1}$}\widehat {V\phantom{\chi_{\Omega}}\Phi_\lambda}
(p_1,\pm q_2,q_1) \dd p_1 \dd q_1\\
+ \int_{\BR^2} \overline{S
\rlap{$\phantom{V} \chi_{\tilde\Omega_1}$}\widehat {V\phantom{\chi_{\Omega}}\Phi_\lambda}
(p,q_1)} (L^\pm(p_1,p_2+\eps, q_1, q_2+\eps')-L^\pm(p, q)) S
\rlap{$\phantom{V} \chi_{\tilde\Omega_1}$}\widehat {V\phantom{\chi_{\Omega}}\Phi_\lambda}
(p_1,\pm q_2,q_1) \dd p_1 \dd q_1
\end{multline}
Using \eqref{cont_bd.3} it follows by dominated convergence that the last line vanishes as $\eps,\eps'\to0$.
Furthermore, note that by the mixed Lebesgue norm estimates in Lemma~\ref{lea:lp_prop}
\begin{equation}
\lVert
\rlap{$\phantom{W_\epsilon} \chi_{\tilde\Omega_1}$}\widehat {W_\epsilon\phantom{\chi_{\Omega}}\Phi_\lambda}
\rVert_{L_p^\infty L_{q_1}^2} \leq \frac{\lVert W_\eps\lVert_1^{1/2}}{2\pi} \lVert W_\eps^{1/2}  \chi_{\tilde\Omega_1} \Phi_\lambda\lVert_2\leq \frac{\lVert W_\eps\lVert_1}{2\pi} \lVert \Phi_\lambda\lVert_{L_r^\infty L_{z_1}^2}
\end{equation}
where $ \lVert \Phi_\lambda\lVert_{L_r^\infty L_{z_1}^2}<\infty$ was shown in Lemma~\ref{lea:phi_infty}.
Since $\lVert W_\eps\Vert_1\leq |\eps| \lVert |\cdot|V \rVert_1$ it follows from \eqref{cont_bd.4} that the first two lines in \eqref{cont_bd.5} vanish as $\eps,\eps'\to0$.
In particular, $g_\pm$ are continuous.
\end{proof}

\section{ Proof of Lemma~\ref{lea:eps_to_0}}\label{sec:epsto0}
This section contains the proof of Lemma~\ref{lea:eps_to_0}, where we compute $\lim_{\eps \to 0} \langle \psi_\lambda^\eps, UH^{\Omega_2}_{T_c^1(\lambda)}U^\dagger  \psi_\lambda^\eps \rangle$.
Recall from \eqref{t_halfspace} that 
\begin{equation}
t (p_1,q_1,r_1,z_1)=\frac{1}{2}\left(e^{-i(p_1 r_1+q_1 z_1)}+e^{i(p_1 r_1+q_1 z_1)}\mp e^{-i(p_1 z_1+q_1 r_1)} \mp e^{i(p_1 z_1+q_1 r_1)}\right).
\end{equation}
Let $\tilde \Omega_2=\{(r,z)\in \BR^{2}\times \BR^{2} \vert \vert r_1\vert< z_1,\vert r_2\vert< z_2\}$.
Analogously to the expression for $UK_T^{\Omega_1}U^\dagger$ in \eqref{UK1U} we have
\begin{multline}\label{UH2U}
 \langle \psi_\lambda^\eps, UH^{\Omega_2}_{T}U^\dagger  \psi_\lambda^\eps \rangle =
\int_{\BR^{4}}B_{T}(p,q)^{-1} \left|\int_{\tilde\Omega_2}\frac{1}{(2\pi)^2} t (p_1,q_1,r_1,z_1) t (p_2,q_2,r_2,z_2) \psi_\lambda^\eps(r,z) \dd r \dd z\right|^2\dd p\dd q\\
- \lambda \int_{\tilde \Omega_2} V(r)| \psi_\lambda^\eps(r,z)|^2 \dd r \dd z.
\end{multline}
Since the function $\psi_\lambda^\eps$ defined in \eqref{trial_state} is symmetric under $(r_2,z_2)\to-(r_2,z_2)$ and (anti)symmetric under $(r_2,z_2)\to(z_2,r_2)$, we have
\begin{equation}
\int_{|r_2|<z_2}  t (p_2,q_2,r_2,z_2) \psi_\lambda^\eps(r,z) \dd r_2 \dd z_2 = \frac{1}{2} \int_{\BR^2} e^{-i p_2 r_2-i q_2 z_2} \psi_\lambda^\eps(r,z) \dd r_2 \dd z_2
\end{equation}
and
\begin{equation}
\int_{|r_2|<z_2} V(r) |\psi_\lambda^\eps(r,z)|^2 \dd r_2 \dd z_2 = \frac{1}{4} \int_{\BR^2} (V(r)\chi_{ \vert r_2\vert< \vert z_2\vert }+ V(r_1,z_2)\chi_{\vert  z_2\vert < \vert r_2\vert}) |\psi_\lambda^\eps(r,z)|^2 \dd r_2 \dd z_2 .
\end{equation}
Comparing with the expression for $UK_T^{\Omega_1}U^\dagger$ in \eqref{UK1U} we obtain
\[
\langle \psi_\lambda^\eps, UH^{\Omega_2}_{T_c^1(\lambda)}U^\dagger  \psi_\lambda^\eps \rangle=\frac{1}{4} \langle \psi_\lambda^\eps, H^{2}_{T_c^1(\lambda)} \psi_\lambda^\eps \rangle,
\]
 where the operator $H_T^2$ is given by
\begin{equation}
H^2_T=UK_T^{\Omega_1}U^\dagger-\lambda V(r)\chi_{ \vert r_2\vert< \vert z_2\vert }-\lambda V(r_1,z_2)\chi_{\vert  z_2\vert < \vert r_2\vert}
\end{equation}
acting on $L^2(\tilde \Omega_1\times \BR)$ functions symmetric in $r$ and antisymmetric/symmetric under swapping $r_2\leftrightarrow z_2$ for Dirichlet/Neumann boundary conditions, respectively.
Let us define $K_T^2:=UK_T^{\Omega_1}U^\dagger$.

The trial state $\psi_\lambda^\eps$ has four summands, which we number from one to four in the order they appear in \eqref{trial_state} and refer to as $|j\rangle$ for $j\in\{1,2,3,4\}$.
By symmetry under $(z_2,r_2)\to -(z_2,r_2)$ and $(r_2,z_2)\to (z_2,r_2)$ we have
\begin{equation}
\langle \psi_\lambda^\eps, H^2_{T_c^1} \psi_\lambda^\eps \rangle = 4 \sum_{j=1}^4 \langle 1,   H^2_{T_c^1}  j\rangle
\end{equation}
For each $j\in\{1,2,3,4\}$ we write
\begin{equation}\label{eq:1Hj}
 \langle 1,   H^2_{T_c^1}  j\rangle =  \langle 1,   (K_{T_c^1}^{2}-\lambda V(r)  )j\rangle+\langle 1,  ( \lambda V(r) \chi_{|z_2|<|r_2|} +\lambda V(r_1,z_2)\chi_{|r_2|<|z_2|}) j\rangle-\langle 1,  \lambda V(r_1,z_2)\ j\rangle
\end{equation}
We shall prove that
\begin{equation}\label{l0_pf}
\lim_{\eps\to 0}\sum_{j=1}^4 \langle 1,   (K_{T_c^1}^{2}-\lambda V(r) ) j\rangle=0,
\end{equation}
\begin{equation}\label{l1_pf}
L_1=\lim_{\eps\to 0}\sum_{j=1}^4 \langle 1,   (V(r) \chi_{|z_2|<|r_2|} + V(r_1,z_2)\chi_{|r_2|<|z_2|} )j\rangle,
\end{equation}
and
\begin{equation}\label{l2_pf}
L_2=-\lim_{\eps\to 0}\sum_{j=1}^4 \langle 1,   V(r_1,z_2)\ j\rangle
\end{equation}
where $L_1$ and $L_2$ are the expressions in \eqref{eq:eps_to_0.2} and  \eqref{eq:eps_to_0.3}.
In particular, it follows that
\[\lim_{\eps\to 0}\langle \psi_\lambda^\eps, U H^{\Omega_2}_{T_c^1} U^\dagger \psi_\lambda^\eps \rangle =\lambda(L_1+L_2).
\]

\subsection{Proof of \eqref{l0_pf}:}
We argue that all summands vanish as $\eps\to0$.

{\bf j=1:}
We first show that
\begin{equation}\label{l0_pf.1}
 \langle 1,   (K_{T_c^1}^{2}-\lambda V(r) ) 1\rangle = \frac{1}{2\pi}\int_{\BR^4} \Bigg[ B_{T_c^1}^{-1}(p,(q_1,q_2+\eta)) -B_{T_c^1}^{-1}(p,(q_1,\eta))  \Bigg] \frac{\eps^2}{(\eps^2+q_2^2)^2}|\widehat{\Phi_\lambda}(p,q_1)|^2\dd p \dd q
\end{equation}
Using eigenvalue equation $K_{T_c^1}^{1}(\eta)\chi_{\tilde\Omega_1}\Phi_\lambda=\lambda V \chi_{\tilde\Omega_1}\Phi_\lambda$ together with the expressions  \eqref{UK1U} and  \eqref{KT1q}  for $K_T^{\Omega_1}$ and $K_T^1(q_2)$, respectively, we observe that
\begin{multline}
 \langle 1,   (K_{T_c^1}^{2}-\lambda V(r) ) 1\rangle\\
=\frac{1}{(2\pi)^4}\int_{( \tilde \Omega_1\times \BR)^2\times \BR^3} \overline{\Phi_\lambda}(r,z_1)\overline{t(p_1,q_1,r_1,z_1)}e^{i p_2 r_2}\Bigg[\int_\BR B_{T_c^1}^{-1}(p,q) e^{i (\eta-q_2) (z_2'-z_2)-\eps (|z_2|+|z_2'|)}\dd q_2\\
-B_{T_c^1}^{-1}(p,(q_1,\eta))e^{-2\eps |z_2|} 2\pi \delta(z_2-z_2')  \Bigg]t(p_1,q_1,r'_1,z_1')e^{-i p_2 r'_2}\Phi_\lambda(r',z_1')\dd r \dd z \dd r' \dd z' \dd p \dd q_1
\end{multline}
We shall carry out the $r,r',z,z'$ integrations.
Integration of $\frac{1}{(2\pi)^{3/2}} t\cdot  e^{-i p_2 r_2} \Phi_\lambda$ over $r,z_1$ gives $\frac{1}{2} \widehat{\Phi_\lambda}$ (c.f. \eqref{pf_eps0_1}) and for the integration over $z_2, z_2'$ we observe
\begin{align*}
\int_\BR e^{i (\eta-q_2) z_2-\eps |z_2|} \dd z_2&=\frac{2\eps}{\eps^2+(\eta-q_2)^2}, \\
2\pi \int_\BR e^{-2\eps |z_2|}=2\pi \eps^{-1}&= \int_\BR \frac{4\eps^2}{(\eps^2+(\eta-q_2)^2)^2} \dd q_2.
\end{align*}
In total, we obtain
\begin{equation}
\langle 1,(  K_{T_c^1}^{2}-\lambda V(r) ) 1\rangle
= \frac{1}{2\pi}\int_{\BR^4}\Bigg[ B_{T_c^1}^{-1}(p,q) -B_{T_c^1}^{-1}(p,(q_1,\eta))  \Bigg] \frac{\eps^2}{(\eps^2+(\eta-q_2)^2)^2} |\widehat{\Phi_\lambda}(p,q_1)|^2 \dd p \dd q
\end{equation}
and substituting $q_2\to q_2+\eta$ we arrive at \eqref{l0_pf.1}.

For $|q_2|>1$, using $B_T^{-1}(p,q)\leq \tilde C (1+p^2+q^2)$ (see Lemma~\ref{KT-Laplace}) we bound the integrand in \eqref{l0_pf.1} above by $\frac{C \eps^2 (1+p^2+q_1^2)}{q_2^2}|\widehat{\Phi_\lambda}(p,q_1)|^2$ .
Since $\Phi_\lambda \in H^1(\BR^3)$, the integral vanishes as $\eps \to 0$.
For $|q_2|<1$ substitute $q_2\to \eps q_2$ and use that
\[
q_2^{-1}( B_{T_c^1}^{-1}(p,(q_1,q_2+\eta) -B_{T_c^1}^{-1}(p,(q_1,\eta))) =-f(p,(q_1,\eta),q_2) B_{T_c^1}^{-1}(p,(q_1,q_2+\eta) B_{T_c^1}^{-1}(p,(q_1,\eta))
\]
 where $f$ is defined as in Lemma~\ref{bdiff}.
The integral then equals
\begin{equation}
-\frac{1}{2\pi}\int_{\BR^4}\chi_{|q_2|<\eps^{-1}} f(p,(q_1,\eta),\eps q_2) B_{T_c^1}^{-1}(p,(q_1,\eps q_2+\eta) B_{T_c^1}^{-1}(p,(q_1,\eta))
  \frac{q_2}{(1+q_2^2)^2}|\widehat{\Phi_\lambda}(p,q_1)|^2\dd p \dd q.
\end{equation}
By Lemma~\ref{bdiff} and Lemma~\ref{KT-Laplace} the integrand is bounded above by the integrable function
\begin{equation}
C (1+p^2+q_1^2) \frac{|q_2|}{(1+q_2^2)^2}|\widehat{\Phi_\lambda}(p,q_1)|^2.
\end{equation}
Thus by dominated convergence, continuity of $f$ and $B_T$ and since $\int_\BR \frac{q_2}{(1+q_2^2)^2}  \dd q_2=0$ we have $\lim_{\eps \to 0} \langle 1,   K_{T_c^1}^{2}-\lambda V(r)  1\rangle=0$.

{\bf j=2:}
We distinguish the cases $\eta(\lambda)=0$ and $\eta(\lambda)\neq 0$.
If $\eta(\lambda)=0$, $\Phi_\lambda(r,z_1)$ is either even or odd in $r_2$. The term for $j=2$ hence agrees with the term for $j=1$ or its negative and hence vanishes in the limit.
For $\eta(\lambda)\neq 0$, the intuition is that integration over $z_2,z_2'$ approximately gives a product of delta functions $\delta(q_2-\eta)\delta(q_2+\eta)=0$.
Using that the integral of $\frac{1}{(2\pi)^{3/2}} t \cdot e^{-i p_2 r_2} \Phi_\lambda$ over $r,z_1$ gives $\frac{1}{2} \widehat{\Phi_\lambda}$ (see \eqref{pf_eps0_1}) and $e^{-i p_2 r_2}=e^{-i (-p_2)(-r_2)}$ 
we have
\begin{multline}
 \langle 1,  ( K_{T_c^1}^{2}-\lambda V(r)  )2\rangle\\
=\frac{1}{8\pi}\int_{\BR^6}\overline{\widehat \Phi_\lambda(p,q_1) } B_{T_c^1}^{-1}(p,q) e^{-i (\eta-q_2) z_2-i (\eta+q_2) z_2'-\eps (|z_2|+|z_2'|)}\widehat \Phi_\lambda((p_1,-p_2),q_1)\dd z_2 \dd z_2' \dd p \dd q\\
-\int_{\tilde \Omega_1\times \BR} \overline{\Phi_\lambda(r,z_1)} \lambda V(r) \Phi_\lambda(r_1,-r_2,z_1) e^{-2i \eta z_2-2\eps |z_2|} \dd r \dd z
\end{multline}
Carrying out the $z_2$ and $z_2'$ integrations gives
\begin{multline}\label{eq:pf4.8.2}
 \langle 1,   (K_{T_c^1}^{2}-\lambda V(r) ) 2\rangle\\
=\frac{1}{2\pi}\int_{\BR^4}\overline{\widehat \Phi_\lambda(p,q_1) } B_{T_c^1}^{-1}(p,q) \frac{\eps^2}{(\eps^2+(\eta-q_2)^2)(\eps^2+(\eta+q_2)^2)}\widehat \Phi_\lambda((p_1,-p_2),q_1)\dd p \dd q\\
-\int_{\tilde \Omega_1} \overline{\Phi_\lambda(r,z_1)} \lambda V(r) \Phi_\lambda(r_1,-r_2,z_1) \frac{\eps}{\eps^2+\eta^2} \dd r \dd z_1
\end{multline}
Using the Schwarz inequality in the $r_2$ variable, we bound the absolute value of the second term by $\frac{ \eps\lambda}{\eta^2}\int_{\tilde \Omega_1}V(r)\vert \Phi_\lambda(r,z_1)|^2 \dd r \dd z_1 \leq \frac{ \eps\lambda}{\eta^2}\lVert V\rVert_1 \lVert \Phi_\lambda\rVert_{L_1^\infty L_{2}^2}^2$. 
It was shown in Lemma~\ref{lea:phi_infty} that $ \lVert \Phi_\lambda\rVert_{L_1^\infty L_{2}^2}<\infty$ and hence the term vanishes for $\eps \to 0$ .
To bound the absolute value of the first term in \eqref{eq:pf4.8.2}, we first use that $B_T^{-1}(p,q)\leq C(1+p^2+q^2)$ by Lemma~\ref{KT-Laplace} and the Schwarz inequality in the $p_2$ variable, and then use symmetry to restrict to $q_2>0$ and distinguish the cases $|q_2-\eta|\lessgtr \eps$:
\begin{multline}
C \int_{\BR^4} \frac{\eps^2(1+p^2+q^2)}{(\eps^2+(\eta-q_2)^2)(\eps^2+(\eta+q_2)^2)}|\widehat \Phi_\lambda(p,q_1)|^2\dd p \dd q\\
\leq 2C \int_{\BR^3} \Bigg(\int_0^\infty \Bigg[\frac{\chi_{|q_2-\eta|<\eps}(1+p^2+q^2)}{(\eta-q_2)^2+(\eta+q_2)^2}+ \frac{\chi_{|q_2-\eta|>\eps}\eps^2(1+p^2+q^2)}{(\eta-q_2)^2(\eta+q_2)^2}\Bigg] \dd q_2 \Bigg)|\widehat \Phi_\lambda(p,q_1)|^2\dd p \dd q_1.
\end{multline}
There is a constant $C(\eta)$ such that the first term in the square brackets is bounded above by $C(\eta)\chi_{|q_2-\eta|<\eps}(1+p^2+q_1^2)$, and the second term is bounded by $C(\eta)\frac{\chi_{|q_2-\eta|>\eps}\eps^2(1+p^2+q_1^2)}{(\eta-q_2)^2}$.
This gives the upper bound
\begin{equation}
\tilde C  \Bigg(\int_0^\infty \Bigg[\chi_{|q_2-\eta|<\eps}+ \frac{\chi_{|q_2-\eta|>\eps}\eps^2}{(\eta-q_2)^2}\Bigg] \dd q_2 \Bigg) \lVert \Phi_\lambda\rVert_{H^1(\BR^3)}^2
\end{equation}
The remaining integral is of order $O(\eps)$ as $\eps\to0$, and thus the term vanishes in the limit $\eps\to0$.

{\bf j=3,4:}
Using the eigenvalue equation $K_{T_c^1(\lambda)}^{1}(\eta)\chi_{\tilde\Omega_1} \Phi_\lambda=\lambda V\chi_{\tilde\Omega_1} \Phi_\lambda$ and that the integral of $\frac{1}{(2\pi)^{3/2}}t\cdot \Phi_\lambda$ over the spatial variables gives $\frac{1}{2}\widehat{\Phi_\lambda}$, see \eqref{pf_eps0_1}, we have
\begin{multline}
| \langle 1,   (K_{T_c^1}^{2}-\lambda V(r))  j\rangle|\\
=\Big \vert \frac{1}{8\pi}\int_{\BR^6}\overline{\widehat \Phi_\lambda(p,q_1) } \left(B_{T_c^1}^{-1}(p,q)-B_{T_c^1}^{-1}(p,(q_1,\eta))\right) e^{-i (\eta-q_2) z_2-i (\mp \eta+p_2) r_2'-\eps (|z_2|+|r_2'|)}\\
\times \widehat \Phi_\lambda((p_1,\pm q_2),q_1)\dd z_2 \dd r_2' \dd p \dd q\Big \vert
\end{multline}
where the upper signs correspond to $j=3$ and the lower ones to $j=4$, respectively.
Carrying out the integration over $r_2'$ and $z_2$ and substituting $q_2\to \eps q_2+\eta, p_2\to \eps p_2\pm \eta$ we obtain
\begin{multline}
| \langle 1,   (K_{T_c^1}^{2}-\lambda V(r))  j\rangle|\\
=\Big \vert\frac{1}{2\pi}\int_{\BR^4}\overline{\widehat \Phi_\lambda((p_1,\eps p_2\pm \eta),q_1) } \frac{1}{1+ p_2^2}\frac{1}{1+q_2^2}\Big[B_{T_c^1}^{-1}((p_1,\eps p_2\pm \eta),(q_1,\eps q_2+\eta))\\
-B_{T_c^1}^{-1}((p_1,\eps p_2\pm \eta),(q_1,\eta)) \Big]
\widehat \Phi_\lambda((p_1,\pm(\eps q_2+ \eta)),q_1) \dd p \dd q\Big \vert
\end{multline}
With the definition of $g_\pm$ as in Lemma~\ref{lea:cont_bd}, the latter equals
\begin{equation}
\Big \vert\frac{1}{2\pi}\int_{\BR^2} \frac{g_{\pm}(\eps p_2\pm \eta ,\eps q_2+\eta)}{(1+ p_2^2)(1+q_2^2)}\dd p_2 \dd q_2\Big \vert
\end{equation}
With Lemma~\ref{lea:cont_bd} it follows by dominated convergence that $\lim_{\eps\to0} \langle 1, (  K_{T_c^1}^{2}-\lambda V(r)  )j\rangle=0$.

\subsection{Proof of \eqref{l1_pf}:}
We have
\begin{multline}
\sum_{j=1}^4 \langle 1,   (V(r) \chi_{|z_2|<|r_2|} + V(r_1,z_2)\chi_{|r_2|<|z_2|} )j\rangle
=\int_{\tilde \Omega_1\times \BR}(V(r)\chi_{|z_2|<|r_2|}+V(r_1,z_2)\chi_{|r_2|<|z_2|} ) \overline{\Phi_\lambda(r,z_1)}\\
\times \Big(\Phi_\lambda(r,z_1)e^{-2\eps|z_2|}+ \Phi_\lambda(r_1,-r_2,z_1)e^{-2\eps|z_2|-2i\eta z_2}
\mp \Phi_\lambda(r_1,z_2,z_1)e^{-\eps(|r_2|+|z_2|)-i\eta(z_2- r_2)} \\
\mp \Phi_\lambda(r_1,- z_2,z_1)e^{-\eps(|r_2|+|z_2|)-i\eta(z_2+ r_2)}\Big)\dd r \dd z
\end{multline}
The claim follows from dominated convergence provided that
\begin{multline}
\int_{\BR^4}(V(r)\chi_{|z_2|<|r_2|}+V(r_1,z_2)\chi_{|r_2|<|z_2|} )|\Phi_\lambda(r,z_1)|\Big(|\Phi_\lambda(r,z_1)|+ |\Phi_\lambda(r_1,-r_2,z_1)|\\
+| \Phi_\lambda(r_1,z_2,z_1)|+| \Phi_\lambda(r_1,- z_2,z_1)|\Big)\dd r \dd z
\end{multline}
is finite.
Using the Schwarz inequality in $z_1$ and carrying out the integration over $z_2$, this is bounded above by
\begin{equation}
4 \int_{\BR^3}(V(r)\chi_{|z_2|<|r_2|}+V(r_1,z_2)\chi_{|r_2|<|z_2|} ) \lVert \Phi_\lambda \rVert_{L_1^\infty L_{2}^2}\dd r \dd z_2
\leq 16 \int_{\BR^2} V(r)|r_2|\dd r  \lVert \Phi_\lambda \rVert_{L_1^\infty L_{2}^2}
\end{equation}
This is finite since $\lVert \Phi_\lambda \rVert_{L_1^\infty L_{2}^2}<\infty$ was shown in Lemma~\ref{lea:phi_infty} and $|\cdot |V \in L^1$ by assumption.

\subsection{Proof of \eqref{l2_pf}:}
{\bf j=1,2:}
We have
\begin{equation}
\langle 1,   V(r_1,z_2) 1\rangle=\int_{\tilde \Omega_1\times \BR} V(r_1,z_2)|\Phi_\lambda(r,z_1)|^2e^{-2\eps|z_2|}\dd r \dd z
\end{equation}
and
\begin{equation}
\langle 1,   V(r_1,z_2) 2\rangle=\int_{\tilde \Omega_1\times \BR} V(r_1,z_2) \overline{\Phi_\lambda(r,z_1)}\Phi_\lambda(r_1,-r_2,z_1)e^{-2\eps|z_2|-2i\eta z_2}\dd r \dd z
\end{equation}
In both cases we can apply dominated convergence since $V(r_1,z_2)|\Phi_\lambda(r,z_1)|^2\in L^1(\BR^4)$ by Lemma~\ref{lea:Vphi_reg.1} (and using additionally the Schwarz inequality in the second case) and obtain the first two terms in $L_2$.

{\bf j=3,4:}
We start with the case of Neumann boundary conditions.
Rewriting the expression in momentum space we have
\begin{multline}
\langle 1,   V(r_1,z_2) j\rangle=\int_{\BR^4} V(r_1,z_2)  \chi_{\tilde \Omega_1} \overline{\Phi_\lambda(r,z_1)}\Phi_\lambda(r_1,\pm z_2,z_1)e^{-\eps|z_2|-i\eta z_2}e^{-\eps|r_2|\pm i\eta r_2}\dd r \dd z\\
=\frac{2}{\pi} \int_{\BR^4}\overline{\widehat{\Phi_\lambda}(p,q_1)}\rlap{$\phantom{V} \chi_{\tilde\Omega_1}$}\widehat {V\phantom{\chi_{\Omega}}\Phi_\lambda}
(p_1,p_2',q_1) \frac{\eps^2}{(\eps^2 +(p_2\mp \eta)^2)(\eps^2 +(p_2'\mp \eta)^2)} \dd p_1 \dd p_2 \dd p_2' \dd q_1\\
=\frac{2}{\pi}  \int_{\BR^2}g_0(\eps p_2\pm \eta,\eps p_2'\pm \eta )\frac{1}{(1 +p_2^2)(1 +p_2'^2)} \dd p_2 \dd p_2'
\end{multline}
where the upper/lower signs correspond to $j=3$ and $j=4$, respectively, and $g_0$ is defined as in Lemma~\ref{lea:cont_bd}.
It follows from Lemma~\ref{lea:cont_bd}, dominated convergence and $\int_\BR \frac{1}{1+x^2} \dd x =  \pi$ that
\begin{equation}
\lim_{\eps\to0} \langle 1,   V(r_1,z_2) j\rangle=2\pi  g_0(\pm \eta,\pm \eta )
\end{equation}
For Dirichlet boundary conditions this comes with a minus sign.

\section{Weak coupling asymptotics}\label{sec:lambdato0}
In this section we shall prove Lemma~\ref{lea:lambdato0}.
We prove the desired asymptotic bounds $L_1=O(1)$ and $L_2\leq -C/\lambda$ as $\lambda\to 0$ in Sections \ref{sec:l1_bd} and \ref{sec:l2_bd}, respectively.

\subsection{Asymptotics of $L_1$}\label{sec:l1_bd}
We recall the definition of $L_1$
\begin{multline}\label{eq:eps_to_0.2.b}
L_1=  \int_{\tilde \Omega_1\times \BR}  \chi_{\vert z_2\vert <\vert r_2\vert}V(r)\Bigg(\vert \Phi_\lambda(r_1,r_2,z_1)\vert^2+\vert \Phi_\lambda(r_1,z_2,z_1)\vert^2 \\
+\overline{\Phi_\lambda(r_1,r_2,z_1)}\Phi_\lambda(r_1,-r_2,z_1)e^{-2i\eta(\lambda) z_2}
+\overline{\Phi_\lambda(r_1,z_2,z_1)}\Phi_\lambda(r_1,-z_2,z_1)e^{-2i\eta(\lambda) r_2}\\
\mp \overline{\Phi_\lambda(r_1,r_2,z_1)}\Phi_\lambda(r_1,z_2,z_1)e^{i\eta(\lambda)(r_2- z_2)}
\mp \overline{\Phi_\lambda(r_1,z_2,z_1)}\Phi_\lambda(r_1,r_2,z_1)e^{-i\eta(\lambda) (r_2-z_2)}\\
\mp \overline{\Phi_\lambda(r_1,r_2,z_1)}\Phi_\lambda(r_1,-z_2,z_1)e^{-i\eta(\lambda)(r_2+ z_2)}
\mp \overline{\Phi_\lambda(r_1,z_2,z_1)}\Phi_\lambda(r_1,-r_2,z_1)e^{i\eta(\lambda) (-r_2+z_2)}
\Bigg)\dd r \dd z
\end{multline}
The goal is to show that $ L_1$ is of order $O(1)$ as $\lambda\to0$.
By the Schwarz inequality, it suffices to prove that $ \int_{\tilde \Omega_1\times \BR} \chi_{|z_2|<|r_2|} V(r) (|\Phi_\lambda(r_1,r_2,z_1)|^2 +|\Phi_\lambda(r_1,z_2,z_1)|^2)\dd r \dd z=O(1)$.
Furthermore, since $\Phi_\lambda=\Phi_\lambda^d\mp \Phi_\lambda^{ex,<}\mp \Phi_\lambda^{ex,>}$ (see \eqref{phi_lambda_j.1} and \eqref{phi_lambda_j.2} for the definitions), again by the Schwarz inequality it suffices to prove
\begin{equation}\label{l2.1}
 \int_{ \tilde \Omega_1\times \BR} \chi_{|z_2|<|r_2|} V(r) |\Phi^j_\lambda(r_1,r_2,z_1)|^2 \dd r \dd z =O(1)
\end{equation}
and
\begin{equation}\label{l2.2}
 \int_{ \tilde \Omega_1\times \BR} \chi_{|z_2|<|r_2|} V(r) |\Phi^j_\lambda(r_1,z_2,z_1)|^2 \dd r \dd z =O(1)
\end{equation}
for $j\in\{d,(ex,<),(ex,>)\}$.

{\bf Case $j\in \{d,(ex,>)\}$:}
In Lemma~\ref{lea:phi_infty} we show that $\sup_{r\in \BR^2}\int_\BR |\Phi^j_\lambda(r,z_1)|^2\dd z_1 = O(1)$.
Both \eqref{l2.1} and \eqref{l2.2} follow since $|\cdot | V \in L^1$.

{\bf Case $j=(ex,<)$:}
Let $W_1(r):=2|r_2| V(r)$ and $W_2(r):=\int_{\BR} V(r_1,z_2)\chi_{|r_2|<|z_2|} \dd z_2$.
We have $W_1,W_2\in L^1(\BR^2)$.
Note that
\begin{equation}
\int_{ \tilde \Omega_1\times \BR} \chi_{|z_2|<|r_2|} V(r) |\Phi^{ex,<}_\lambda(r_1,r_2,z_1)|^2 \dd r \dd z = \int_{\tilde \Omega_1} W_1(r) |\Phi^{ex,<}_\lambda(r_1,r_2,z_1)|^2 \dd r \dd z_1
\end{equation}
and
\begin{equation}
 \int_{\tilde \Omega_1\times \BR} \chi_{|z_2|<|r_2|} V(r) |\Phi^{ex,<}_\lambda(r_1,z_2,z_1)|^2 \dd r \dd z =\int_{\tilde \Omega_1} W_2(r) |\Phi^{ex,<}_\lambda(r_1,r_2,z_1)|^2 \dd r \dd z_1,
\end{equation}
where we renamed $z_2\leftrightarrow r_2$.
For any $L^1$-function $W\geq 0$ we have
\begin{multline}\label{l2.3}
\left(\int_{\tilde \Omega_1} W(r) |\Phi^{ex,<}_\lambda(r_1,r_2,z_1)|^2 \dd r \dd z_1\right)^{1/2}=\lVert W^{1/2} \Phi^{ex,<}_\lambda \rVert_2=\sup_{\psi\in L^2(\tilde \Omega_1),\lVert \psi\rVert_2=1}| \langle \psi, W^{1/2} \Phi^{ex,<}_\lambda\rangle|\\
\leq \sqrt{2}\lambda \sup_{\psi_1,\psi_2\in L^2(\BR^3),\lVert \psi_1\rVert=\lVert \psi_2\rVert=1}\int_{\BR^3}\big| \widehat{W^{1/2} \psi_1}(p,q_1) B_{T_c^1}(p,(q_1,\eta)) \chi_{p_2^2<2\mu}\\
 \times\widehat{V^{1/2} \psi_2}(( q_1,p_2), p_1)\big| \dd p \dd q_1
\end{multline}
where we used the definition of $\Phi^{ex,<}_\lambda$, see \eqref{phi_lambda_j.2}, and the normalization $\lVert \Psi_\lambda  \rVert=1$ in the last step.
We bound $|\widehat{W^{1/2} \psi_1}(p,q_1)|\leq \lVert W \rVert_1^{1/2} \lVert F_2\psi_1(\cdot,q_1)\rVert_2$, and similarly for $|\widehat{V^{1/2} \psi_2}(p,q_1)|$.
Thus \eqref{l2.3} is bounded above by
\begin{equation}\label{l2.4}
\sqrt{2}\lambda \lVert W \rVert_1^{1/2}\lVert V \rVert_1^{1/2}\lVert B_T^{ex}(\eta) \rVert
\end{equation}
where $B_T^{ ex}(q_2)$ is the operator on $L^2(\BR)$ with integral kernel
\begin{equation}\label{btex}
B_T^{ ex}(q_2)(p_1,q_1)=\int_\BR B_{T}(p,q)\chi_{p_2^2<2\mu} \dd p_2.
\end{equation}
It was shown in \cite[Proof of Lemma 6.1]{roos_bcs_2023} (see Eq. (6.16) and rest of argument), that
\begin{equation}
\sup_T \sup_{q_2} \lVert B_T^{ ex}(q_2)\rVert<\infty.
\end{equation}
In particular, we conclude that $\int_{\tilde \Omega_1} W_k(r) |\Phi^{ex,<}_\lambda(r_1,r_2,z_1)|^2 \dd r \dd z_1 =O(\lambda^2)$ for $k\in\{1,2\}$.

\subsection{Asymptotics of $L_2$}\label{sec:l2_bd}
Recall that
\begin{multline}\label{eq:eps_to_0.3.b}
L_2= - \int_{\tilde \Omega_1\times \BR}  V(r)\Bigg(\vert \Phi_\lambda(r_1,z_2,z_1)\vert^2
+\overline{\Phi_\lambda(r_1,z_2,z_1)}\Phi_\lambda(r_1,-z_2,z_1)e^{-2i\eta(\lambda) r_2}\Bigg) \dd r \dd z\\
\mp 2\pi \int_{\BR^2}\Bigg(\overline{\widehat{\Phi_\lambda}(p_1, \eta (\lambda),q_1)}
\rlap{$\phantom{V} \chi_{\tilde\Omega_1}$}\widehat {V\phantom{\chi_{\Omega}}\Phi_\lambda}
(p_1,\eta(\lambda),q_1)
+\overline{\widehat{\Phi_\lambda}(p_1, -\eta(\lambda) ,q_1)}
\rlap{$\phantom{V} \chi_{\tilde\Omega_1}$}\widehat {V\phantom{\chi_{\Omega}}\Phi_\lambda}
(p_1,-\eta(\lambda),q_1)  \Bigg)\dd p_1 \dd q_1.
\end{multline}
The goal is to prove that $L_2$ diverges like $-\lambda^{-1}$ to negative infinity as $\lambda\to0$.
We shall prove that the second line in \eqref{eq:eps_to_0.3.b} is of order $O(1)$ as $\lambda\to0$.
For the first line in \eqref{eq:eps_to_0.3.b} we shall prove that it is bounded above by $-c\lambda^{-1}$ for some $c>0$ as $\lambda\to0$.

{\bf Second line of \eqref{eq:eps_to_0.3.b}:}
Let $\xi\in\{\eta,-\eta\}$.
Consider the expression
\[
\Big \vert \int_{\BR^2} \overline{\widehat{\Phi_\lambda}(p_1, \xi ,q_1)}
\rlap{$\phantom{V} \chi_{\tilde\Omega_1}$}\widehat {V\phantom{\chi_{\Omega}}\Phi_\lambda}
(p_1,\xi,q_1) \dd p_1 \dd q_1 \Big \vert
\]
which agrees with $|g_0(\xi, \xi)|$ in \eqref{eq:defg0}.
Recalling the expression for $g_0$ in \eqref{cont_bd.0} involving $L^0$ and $S$ defined at the beginning of Section~\ref{sec:pf_cont_bd} we have
\begin{multline}
\Big \vert \int_{\BR^2} \overline{\widehat{\Phi_\lambda}(p_1, \xi ,q_1)}
\rlap{$\phantom{V} \chi_{\tilde\Omega_1}$}\widehat {V\phantom{\chi_{\Omega}}\Phi_\lambda}
(p_1,\xi,q_1) \dd p_1 \dd q_1 \Big \vert \leq  \lambda \int_{\BR^2} (|
\rlap{$\phantom{V} \chi_{\tilde\Omega_1}$}\widehat {V\phantom{\chi_{\Omega}}\Phi_\lambda}
(p_1,\xi,q_1)\vert +|
\rlap{$\phantom{V} \chi_{\tilde\Omega_1}$}\widehat {V\phantom{\chi_{\Omega}}\Phi_\lambda}
(-p_1,\xi,-q_1)\vert\\
+|
\rlap{$\phantom{V} \chi_{\tilde\Omega_1}$}\widehat {V\phantom{\chi_{\Omega}}\Phi_\lambda}
(q_1,\xi,p_1)\vert  )+|
\rlap{$\phantom{V} \chi_{\tilde\Omega_1}$}\widehat {V\phantom{\chi_{\Omega}}\Phi_\lambda}
(-q_1,\xi,-p_1)\vert )
B_{T_c^1}((p_1,\xi),(q_1,\eta))|
\rlap{$\phantom{V} \chi_{\tilde\Omega_1}$}\widehat {V\phantom{\chi_{\Omega}}\Phi_\lambda}
(p_1,\xi,q_1)\vert \dd p_1 \dd q_1
\end{multline}
Using the Schwarz inequality and $|
\rlap{$\phantom{V} \chi_{\tilde\Omega_1}$}\widehat {V\phantom{\chi_{\Omega}}\Phi_\lambda}
(p_1,\xi,q_1)\vert \leq \lVert
\rlap{$\phantom{V} \chi_{\tilde\Omega_1}$}\widehat {V\phantom{\chi_{\Omega}}\Phi_\lambda}
(\cdot ,q_1)\rVert_\infty$ this is bounded above by
\begin{multline}\label{L3.2.1}
 4 \lambda \int_{\BR^2} B_{T_c^1}((p_1,\xi),(q_1,\eta)) \lVert
 \rlap{$\phantom{V} \chi_{\tilde\Omega_1}$}\widehat {V\phantom{\chi_{\Omega}}\Phi_\lambda}
(\cdot ,q_1)\rVert_\infty^2\dd p_1\dd q_1\\
\leq 4 \lambda \sup_{q_1\in \BR}\int_{\BR} B_{T_c^1}((p_1,\xi),(q_1,\eta))\dd p_1 \lVert
\rlap{$\phantom{V} \chi_{\tilde\Omega_1}$}\widehat {V\phantom{\chi_{\Omega}}\Phi_\lambda}
\rVert_{L_{2}^2(\BR)L_1^\infty(\BR^2)}^2,
\end{multline}
where in the second step we used that $\int_\BR B_{T_c^1}((p_1,\xi),(q_1,\eta)) \dd p_1$ acts as multiplication operator on $\lVert \rlap{$\phantom{V} \chi_{\tilde\Omega_1}$}\widehat {V\phantom{\chi_{\Omega}}\Phi_\lambda}
(\cdot ,q_1)\rVert_\infty$.
Using the bound on the mixed Lebesgue norm in Lemma~\ref{lea:lp_prop} and since $\lVert V^{1/2} \chi_{\tilde\Omega_1} \Phi_\lambda \rVert_2=1$ we have $\lVert
\rlap{$\phantom{V} \chi_{\tilde\Omega_1}$}\widehat {V\phantom{\chi_{\Omega}}\Phi_\lambda}
\rVert_{L_{2}^2(\BR)L_1^\infty(\BR^2)}^2\leq \lVert V\rVert_1$.
The following Lemma together with the weak coupling asymptotics of $T_c^1(\lambda)$ and $\eta(T)$ in Remark~\ref{rem:Tetato0} and Lemma~\ref{etato0}\eqref{i-eta} imply that \eqref{L3.2.1} is of order $O(1)$.
\begin{lemma}\label{lea:l3_34}
Let $\xi(T),\xi'(T)$ be functions of $T$ with $\lim_{T\to 0} \xi(T)=\lim_{T\to 0} \xi'(T)=0$.
Then as $T\to 0$,
\begin{equation}
 \sup_{q_1}\int_{\BR} B_{T}((p_1,\xi(T)),(q_1,\xi'(T)))\dd p_1 =O(\ln \mu/T).
\end{equation}
\end{lemma}
The proof can be found in Section~\ref{pf:lea_l3_34}.

{\bf First line of \eqref{eq:eps_to_0.3.b}:}
Recall from Section~\ref{sec:reg_conv} that $\Phi_\lambda=\Phi_\lambda^{>}+\Phi_\lambda^{d,<}\mp \Phi_\lambda^{ex,<}$.
We show in Lemma~\ref{lea:Vphi_reg.1} that the $L^2$-norms of $V^{1/2}(r)\Phi_\lambda^{>}(r_1,z_2,z_1)$, $V^{1/2}(r)\Phi_\lambda^{d,<}(r_1,z_2,z_1)$, and $V^{1/2}(r)\Phi_\lambda^{ex,<}(r_1,z_2,z_1)$ are of order $O(\lambda)$, $O(\lambda^{-1/2})$, and $O(\lambda^{1/2})$, respectively.
It follows with the Schwarz inequality that the first line of $L_2$ in \eqref{eq:eps_to_0.3.b} equals
\begin{equation}\label{L3_12.4}
- \int_{\tilde \Omega_1\times \BR}  V(r)\Bigg(\vert \Phi_\lambda^{d,<}(r_1,z_2,z_1)\vert^2
+\overline{\Phi_\lambda^{d,<}(r_1,z_2,z_1)}\Phi_\lambda^{d,<}(r_1,-z_2,z_1)e^{-2i\eta(\lambda) r_2}\Bigg) \dd r \dd z +O(1)
\end{equation}
Note that $\Phi_\lambda^{d,<}(r_1,z_2,z_1)=\Phi_\lambda^{d,<}(-r_1,z_2,-z_1)$.
We rewrite the expression in \eqref{L3_12.4} as
\begin{multline}\label{L3_12.1}
-\frac{1}{2}\int_{\BR^4} V(r)\overline{\Phi_\lambda^{d,<}(r_1,z_2,z_1)}\Bigg( \Phi_\lambda^{{d,<}}(r_1,z_2,z_1)+\Phi_\lambda^{{d,<}}(r_1,-z_2,z_1)e^{-2i\eta(\lambda) r_2}\Bigg) \chi_{|r_1|<|z_1|} \dd r \dd z\\
=-\frac{1}{2}\int_{\BR^4}V(r)\overline{\Phi_\lambda^{d,<}(r_1,z_2,z_1)}\Bigg( \Phi_\lambda^{{d,<}}(r_1,z_2,z_1)+\Phi_\lambda^{{d,<}}(r_1,-z_2,z_1)e^{-2i\eta(\lambda) r_2}\Bigg) \dd r \dd z\\
+\frac{1}{2}\int_{\BR^4} V(r)\overline{\Phi_\lambda^{d,<}(r_1,z_2,z_1)}\Bigg( \Phi_\lambda^{{d,<}}(r_1,z_2,z_1)+\Phi_\lambda^{{d,<}}(r_1,-z_2,z_1)e^{-2i\eta(\lambda) r_2}\Bigg) \chi_{|z_1|<|r_1|} \dd r \dd z
\end{multline}
We first consider the last line in \eqref{L3_12.1} with the restriction to $|z_1|<|r_1|$.
We prove that this term is of order $O(1)$ as $\lambda\to0$.
Second, we will prove that the expression on the second line in \eqref{L3_12.1} is bounded above by $-c \lambda^{-1}$ for some constant $c>0$ as $\lambda\to0$.

{\bf Asymptotics of third line in \eqref{L3_12.1}:}
Define $W\in L^1(\BR^3)$ by $W(r,z_1):= V(r) \chi_{|z_1|<|r_1|}$.
By the Schwarz inequality it suffices to prove that $\int_{\BR^4}W(r,z_1)|\Phi_\lambda^{d,<}(r_1,z_2,z_1)|^2 \dd r \dd z=O(1)$ for $\lambda\to0$.
Using the definition of $\Phi_\lambda^{d,<}$ we have
\begin{multline}
\int_{\BR^4}W(r,z_1)|\Phi_\lambda^{d,<}(r_1,z_2,z_1)|^2 \dd r \dd z=\frac{2\lambda^2}{(2\pi)^{1/2}} \int_{\BR^5} \widehat{W}((p_1-p_1',0),q_1-q_1')B_{T_c^1}(p,(q_1,\eta))\\
\times \overline{\widehat{V^{1/2}\Psi_\lambda}(p,q_1)}B_{T_c^1}((p'_1,p_2),(q'_1,\eta)) \widehat{V^{1/2}\Psi_\lambda}(p'_1,p_2,q'_1)\chi_{p^2+q_1^2< 2\mu}\chi_{p_1'^2+p_2^2+q_1'^2< 2\mu}\dd p \dd p_1'  \dd q_1 \dd q_1'
\end{multline}
Using $|\widehat{W}(p,q_1)|\leq \frac{\lVert W \rVert_1}{(2\pi)^{3/2}}$ and $\lVert \widehat{V^{1/2}\Psi_\lambda}(\cdot ,q_1)\rVert_\infty\leq \lVert V\rVert_1^{1/2} \lVert F_2 \Psi_\lambda(\cdot,q_1)\rVert_2$ we bound this from above by
\begin{multline}\label{L3_12.22}
\frac{\lambda^2}{2\pi^2}\lVert W \rVert_1 \lVert V\rVert_1\int_{\BR^5} B_{T_c^1}(p,(q_1,\eta)) B_{T_c^1}((p'_1,p_2),(q'_1,\eta))\chi_{p^2+q_1^2< 2\mu}\chi_{p_1'^2+p_2^2+q_1'^2< 2\mu}\\
\times  \lVert F_2\Psi_\lambda(\cdot ,q_1)\rVert_2 \lVert F_2\Psi_\lambda(\cdot ,q'_1)\rVert_2 \dd p \dd p_1'  \dd q_1 \dd q_1'\\
\leq \frac{\lambda^2}{2\pi^2}\lVert W \rVert_1\lVert V\rVert_1\Bigg[\sup_{q_1,q_1'\in \BR}\int_{\BR^3} B_{T_c^1}(p,(q_1,\eta))B_{T_c^1}((p'_1,p_2),(q'_1,\eta)) \chi_{p_1'^2+q_1'^2+p_2^2<2\mu}\chi_{p^2+q_1^2<2\mu} \dd p \dd p_1'\Bigg] \\
\times \Big( \int_\BR \lVert F_2\Psi_\lambda(\cdot ,q_1)\rVert_2  \chi_{q_1^2<2\mu}  \dd q_1 \Big)^2
\end{multline}
The integral over the product of the two $B_{T_c^1}$ terms is of order $O((\ln \mu /T_c^1(\lambda))^3)$ by Lemma~\ref{lea:BB_asy}.
Together with the asymptotics of $T_c^1(\lambda)$ in Remark~\ref{rem:Tetato0}, the term in the square bracket in \eqref{L3_12.22} is thus of order $O(\lambda^{-3})$.
Splitting the domain of integration into $|q_1|/\sqrt{\mu}\gtrless (T_c^1/\mu)^\beta$ for some $0< \beta<1$ and using the Schwarz inequality we observe that
\begin{equation}
\int_\BR  \lVert F_2\Psi_\lambda(\cdot ,q_1)\rVert_2  \chi_{q_1^2< 2\mu}  \dd q_1 \leq (2\sqrt{\mu}(T_c^1/\mu)^\beta)^{1/2} \lVert \Psi_\lambda \rVert_2 +(2\sqrt{2\mu})^{1/2}  \lVert F_2\Psi_\lambda \chi_{|q_1|/\sqrt{\mu}>(T_c^1/\mu)^\beta} \rVert_2
\end{equation}
It was shown in Lemma~\ref{etato0}\eqref{i-q-range} that $\lVert F_2\Psi_\lambda \chi_{|q_1|/\sqrt{\mu}>(T_c^1/\mu)^\beta}  \rVert_2=O(\lambda^{1/2})$.
With the asymptotics of $T_c^1(\lambda)$ in Remark~\ref{rem:Tetato0} we have $(T_c^1/\mu)^{\beta/2}\leq O((\ln \mu/T_c^1)^{-1})=O(\lambda)$.
Thus,
\[
\Big( \int_\BR \lVert F_2\Psi_\lambda(\cdot ,q_1)\rVert_2  \chi_{q_1^2<2\mu}  \dd q_1 \Big)^2=O(\lambda)
\]
and \eqref{L3_12.22} is of order $O(1)$.

{\bf Asymptotics of second line in \eqref{L3_12.1}:}
Writing out the definition of $\Phi_\lambda^{d,<}$, we have
\begin{multline}\label{L3_12.7}
\int_{\BR^4}V(r)\overline{\Phi_\lambda^{d,<}(r_1,z_2,z_1)}\Phi_\lambda^{d,<}(r_1,-z_2,z_1)e^{-2i\eta(\lambda) r_2} \dd r \dd z
=2\lambda^2\int_{\BR^4}\widehat V(p_1+p_1',2\eta) B_{T_c^1}(p,(q_1,\eta)) \\
\times \overline{ \widehat{ V^{1/2}\Psi_\lambda}(p,q_1)}B_{T_c^1}((p_1',p_2),(q_1,\eta))\widehat{ V^{1/2}\Psi_\lambda}(p_1',p_2,q_1) \chi_{p^2+q_1^2<2\mu}\chi_{p_1'^2+p_2^2+q_1^2<2\mu}\dd p \dd p_1' \dd q_1
\end{multline}
We can thus write
\begin{equation}
\frac{1}{2}\int_{\BR^4}V(r)\overline{\Phi_\lambda^{<,d}(r_1,z_2,z_1)}\Big(\Phi_\lambda^{<,d}(r_1,z_2,z_1) +\Phi_\lambda^{<,d}(r_1,-z_2,z_1)e^{-2i\eta(\lambda) r_2}\Big) \dd r \dd z
=\langle F_2 \Psi_\lambda, M_{\lambda} F_2 \Psi_\lambda \rangle ,
\end{equation}
where $M_\lambda$ is the operator acting on $L^2(\BR^3)$ given by
\begin{multline}
\langle \psi, M_\lambda \psi \rangle=
\lambda^2\int_{\BR^4}(\widehat V(p_1-p_1',0)+\widehat V(p_1+p_1',2\eta)) B_{T_c^1}(p,(q_1,\eta))\overline{ F_1 V^{1/2}\psi (p,q_1)} \chi_{p^2+q_1^2<2\mu}\\
\times B_{T_c^1}((p_1',p_2),(q_1,\eta))\chi_{p_1'^2+p_2^2+q_1^2<2\mu} F_1 V^{1/2}\psi (p_1',p_2,q_1)\dd p \dd p_1' \dd q_1
\end{multline}
By the same argument as in the proof of $\int_{\BR^4} V(r) \vert \Phi_\lambda^{d,<}(r_1,z_2,z_1)\vert^2 \dd r \dd z=O(\lambda^{-1})$ in Lemma~\ref{lea:Vphi_reg.1} (see \eqref{phid<_asy}) we have $\lVert M_\lambda \rVert = O(\lambda^{-1})$.
Recall the projections $\BP$ and $\BQ_\beta$ from Section~\ref{sec:reg_conv}.
Let $\BT$ be the projection $\BT=\BP \BQ_\beta$ for some $0<\beta<1$ and $\BT^\perp=1-\BT$.
We have
\begin{equation}\label{L3_12.9}
\langle F_2 \Psi_\lambda, M_\lambda F_2 \Psi_\lambda \rangle = \langle \BT F_2 \Psi_\lambda, M_\lambda \BT F_2 \Psi_\lambda \rangle + \langle \BT F_2 \Psi_\lambda, M_\lambda \BT^\perp F_2 \Psi_\lambda \rangle + \langle \BT^\perp F_2 \Psi_\lambda, M_\lambda F_2 \Psi_\lambda \rangle
\end{equation}
Since $\BP $ and $\BQ_\beta$ commute, we have $\lVert \BT^\perp F_2 \Psi_\lambda\rVert=\lVert \BQ_\beta^\perp F_2\Psi_\lambda+\BQ_\beta \BP^\perp F_2 \Psi_\lambda \rVert =O(\lambda^{1/2})$ according to the asymptotics for $\lVert\BQ^\perp F_2 \Psi_\lambda\rVert$ and $\lVert\BP^\perp F_2 \Psi_\lambda\rVert$ proved in Lemma~\ref{etato0}\eqref{i-Pphi} and \eqref{i-q-range}.
In particular, the last two terms in \eqref{L3_12.9} are of order $O(\lambda^{-1/2})$.
The remaining term in \eqref{L3_12.9} is bounded below by
\begin{multline}\label{jMj.1}
 \langle \BT F_2 \Psi_\lambda, M_\lambda \BT F_2 \Psi_\lambda \rangle \\
\geq \inf_{|q_1|/\sqrt{\mu}<(T_c^1/\mu)^\beta} \lambda^2\int_{\BR^3}(\widehat V(p_1-p_1',0)+\widehat V(p_1+p_1',2\eta)) B_{T_c^1}((p_1,p_2),(q_1,\eta)) \widehat{ V j_2} (p) \chi_{p^2+q_1^2<2\mu}\\
\times B_{T_c^1}((p_1',p_2),(q_1,\eta))\chi_{p_1'^2+p_2^2+q_1^2<2\mu}\widehat{Vj_2} (p_1',p_2)\dd p \dd p_1'  \lVert \BT F_2 \Psi_\lambda \rVert_2^2\lVert V^{1/2}j_2\rVert_2^{-2}
\end{multline}
The remainder of the proof follows the same ideas as the proof of \cite[Lemma 4.11]{roos_bcs_2023}.
Since $V\geq 0$ we have $\widehat{V}(0)>0$.
Furthermore, the eigenvalue equation $e_\mu V^{1/2}j_2=O_\mu V^{1/2}j_2=\widehat{Vj_2}(|p|=\sqrt{\mu}) V^{1/2}j_2 $ implies that $\widehat{Vj_2}(|p|=\sqrt{\mu})=e_\mu>0$.
By continuity of $\widehat{V}$ and $\widehat{Vj_2}$ and since $\eta(\lambda)\to 0$ for $\lambda\to 0$ (see  Lemma~\ref{etato0}\eqref{i-eta}), there exist $\tilde \lambda>0$, $0<\delta<\mu$ and $c_1>0$ such that for all $\sqrt{\mu-\delta}<p_2<\sqrt{\mu+\delta}, p_1^2<4\delta,p_1'^2<4 \delta$ and $\lambda<\tilde \lambda$ we have
\begin{equation}
(\widehat V(p_1-p_1',0)+\widehat V(p_1+p_1',2\eta))  \widehat{ Vj_2} (p) \widehat{Vj_2} (p_1',p_2)\chi_{p^2+q_1^2<2\mu}\chi_{p_1'^2+p_2^2+q_1^2<2\mu}\lVert V^{1/2}j_2\rVert_2^{-2}>c_1.
\end{equation}

Using the second part of Lemma~\ref{lea:BB_asy} and the boundedness of $\widehat{V}, \widehat{Vj_2}$, it follows that up to an error of order $O(\lambda^2 (\ln \mu/T_c^1)^{5/2})=O(\lambda^{-1/2})$ we may restrict the domain of integration in \eqref{jMj.1} to $\sqrt{\mu-\delta}<p_2<\sqrt{\mu+\delta}, p_1^2<4\delta,p_1'^2<4 \delta$ .
Since $\lVert \BT F_2 \Psi_\lambda \rVert_2^2=1-O(\lambda) \geq \frac{1}{2}$ for small $\lambda$, we obtain
\begin{multline}
 \langle \BT F_2 \Psi_\lambda, M_\lambda \BT F_2 \Psi_\lambda \rangle
\geq \frac{c_1}{2}\inf_{|q_1|/\sqrt{\mu}<(T_c^1/\mu)^\beta} \lambda^2\int_{\BR^3} B_{T_c^1}((p_1,p_2),(q_1,\eta))\\
\times B_{T_c^1}((p_1',p_2),(q_1,\eta))\chi_{\mu-\delta<p_2^2<\mu+\delta}\chi_{p_1^2<4 \delta} \chi_{p_1'^2<4 \delta}\dd p \dd p_1'  +O(\lambda^{-1/2})
\end{multline}
Using Lemma~\ref{lea:BB_asy} once more, we may leave away the characteristic functions at the expense of an error of order $O(\lambda^{-1/2})$.
Since $\eta(\lambda)=O(T_c^1(\lambda))$, there is a $c_2>0$ such that $\eta^2+(\sqrt{\mu}(T_c^1/\mu)^\beta)^2\leq c_2^2 \mu (T_c^1/\mu)^{2\beta}$ for $T_c^1<\mu$.
The following Lemma, whose proof is given in Section~\ref{pf:lea_5.2},
 thus concludes the proof of Lemma~\ref{lea:lambdato0}.
\begin{lemma}\label{bb_asy_lowerbd}
Let $\mu,c_2>0$, $0<\beta<1$ and $\eps:=c_2 \sqrt{\mu} (T/\mu)^{\beta}$ for $T>0$.
Then there are constants $T_0,C>0$ such that
\begin{equation}\label{jMj.2}
\inf_{|q|<\eps} \int_\BR \left(\int_{\BR} B_{T}(p,q) \dd p_1\right)^2\dd p_2\geq C(\ln \mu/T)^3
\end{equation}
for all $0<T<T_0$.
\end{lemma}

\section{Proof of Theorem~\ref{thm2}}\label{sec:pfthm2}
This Section is dedicated to the proof of Theorem~\ref{thm2}, which states that the relative difference of $T_c^2$ and $T_c^0$ vanishes in the weak coupling limit.
It has been shown in \cite[Theorem 1.7]{roos_bcs_2023} that the relative difference of $T_c^1$ and $T_c^0$ vanishes in the weak coupling limit and we follow the same proof strategy here.
We first switch to the Birman-Schwinger picture.
Recall the Birman-Schwinger operator $A_T^0$ corresponding to $H_T^{\Omega_0}$ defined in \eqref{eq:defAT0}.
Furthermore, recall the notation $ t$, $\tilde \Omega_2$ and the representation of $UH_T^{\Omega_2}U^\dagger$ in \eqref{UH2U} from Section~\ref{sec:epsto0}.
The corresponding Birman-Schwinger operator $A_T^2:L^2_{\textrm{s}}(\tilde \Omega_2)\to L^2_{\textrm{s}}(\tilde \Omega_2)$ is given by
\begin{equation}\label{a2}
\langle \psi, A_T^2 \psi \rangle =
\int_{\BR^{4}}B_{T}(p,q) \left|\int_{\tilde\Omega_2}\frac{1}{(2\pi)^2} t (p_1,q_1,r_1,z_1) t (p_2,q_2,r_2,z_2) V^{1/2}(r)\psi(r,z) \dd r \dd z\right|^2\dd p\dd q
\end{equation}
and it follows from the Birman-Schwinger principle that $\sgn \inf \sigma(H_T^{\Omega_2}) = \sgn (1/\lambda-\sup \sigma(A_T^{2}))$.
Let $a_T^j=\sup \sigma(A_T^j)$.
For $\lambda\to 0$ asymptotically $a_T^0=e_\mu \ln(\mu/T)+O(1)$, see e.g.~\cite[Section 6]{roos_bcs_2023}.
It is a straightforward generalization of \cite[Lemma 4.1]{hainzl_boundary_2023} that the claim \eqref{eq:treldiff} is equivalent to
\begin{equation}\label{eq:at0-at2}
\lim_{T\to0} (a_T^0-a_T^2)=0
\end{equation}
and we refer to \cite{hainzl_boundary_2023} for the proof.

To verify \eqref{eq:at0-at2}, the first step is to argue that $a_T^2\geq a_T^0$ for all $T>0$.
Lemma~\ref{t2geqt1} together with \cite[Lemma 2.3]{roos_bcs_2023} imply that $\inf \sigma(H_T^{\Omega_2})\leq \inf \sigma(H_T^{\Omega_0})$ for all $\lambda,T>0$.
Using the Birman-Schwinger principle, it follows that $a_T^2\geq a_T^0$ for all $T>0$.
For details we refer to the proof of \cite[Theorem 1.7]{roos_bcs_2023}.

It remains to show that $\lim_{T\to0} (a_T^0-a_T^2)\geq 0$.
We decompose $A_T^2$ in the same spirit as we decomposed $A_T^1(q_2)$ in \eqref{AT1_decomposition}.
For $A_T^1$, the decomposition consisted of the ``unperturbed" term $A_T^0$ and the ``perturbation term" $G_T$, where the first components of the momentum variables were swapped.
For $A_T^2$ we additionally get the terms arising from swapping the variables in the second component, which leads to four terms in total.
Let $\tilde\iota: L^2(\tilde \Omega_2)\to  L^2(\BR^{4})$ be the isometry
\begin{multline}\label{tildeiota}
\tilde \iota \psi(r,z)
= \frac{1}{2} \Big(\psi(r,z) \chi_{\tilde \Omega_2}(r,z) +\psi(-r_1, r_2,-z_1,z_2) \chi_{\tilde \Omega_2}(-r_1,r_2,-z_1,z_2)\\
+\psi(r_1, -r_2,z_1,-z_2) \chi_{\tilde \Omega_2}(r_1,-r_2,z_1,-z_2)+\psi(-r,-z) \chi_{\tilde \Omega_2}(-r,-z) \Big).
\end{multline}
Using the definition of $ t$ and evenness of $V$ in $r_1$ and $r_2$ we rewrite \eqref{a2} as
\begin{multline}\label{at2r}
\langle \psi, A_{T}^2\psi \rangle=\int_{\BR^4} B_T(p,q)\Big \vert \frac{1}{2}(\rlap{$\phantom{V^{1/2}} \tilde\iota$}\widehat {V^{1/2}\phantom{\iota}\psi} (p,q)\mp  \rlap{$\phantom{V^{1/2}} \tilde\iota$}\widehat {V^{1/2}\phantom{\iota}\psi} ((q_1,p_2),(p_1,q_2))\\
\mp  \rlap{$\phantom{V^{1/2}} \tilde\iota$}\widehat {V^{1/2}\phantom{\iota}\psi} ((p_1,q_2),(q_1,p_2)) +\rlap{$\phantom{V^{1/2}} \tilde\iota$}\widehat {V^{1/2}\phantom{\iota}\psi} (q,p))\Big \vert^2 \dd p \dd q
\end{multline}
Define the self-adjoint operators $G_T^1, G_T^2$, and $N_T$ on $L^2(\BR^{4})$ through
\begin{align}
\langle \psi, G_{T}^1\psi \rangle&=\int_{\BR^{4}} \overline{F_1 V^{1/2}\psi((q_1,p_2),(p_1,q_2))} B_{T}(p,q) F_1 V^{1/2}\psi(p,q) \dd p \dd q, \label{gt1}\\
\langle \psi, G_{T}^2\psi \rangle&=\int_{\BR^{4}} \overline{F_1 V^{1/2}\psi((p_1,q_2),(q_1,p_2))} B_{T}(p,q) F_1 V^{1/2}\psi(p,q) \dd p \dd q, \ \mathrm{and} \label{gt2}\\
\langle \psi, N_{T}\psi \rangle&=\int_{\BR^{4}} \overline{F_1 V^{1/2}\psi(q,p)} B_{T}(p,q) F_1 V^{1/2}\psi(p,q) \dd p \dd q. \label{nt}
\end{align}
We slightly abuse notation and write $F_2$ for the Fourier transform in the second variable also when the second variable has two components, i.e.~$F_2 \psi (r, q)=\frac{1}{2\pi} \int_{\BR^2} e^{-iq \cdot z} \psi(r,z) \dd z$.
It follows from \eqref{at2r} and $B_T(p,q)=B_T((q_1,p_2),(p_1,q_2))=B_T(q,p)$ that
\begin{equation}\label{AT2_decomposition}
A_T^2= \tilde \iota^\dagger (A_T^0- F_2^\dagger   R_T F_2) \tilde \iota,
\end{equation}
where $R_T=\pm G_T^1 \pm G_T^2 - N_T$.
Let $B_T(\cdot,q):L^2(\BR^2)\to L^2(\BR^2)$ denote multiplication by $B_T(p,q)$ in momentum space and define the function $E_T(q)$ on $\BR^2$ through
\begin{equation}
E_T(q):=a_T^0-\lVert V^{1/2} B_T(\cdot,q) V^{1/2}\rVert_{\textrm{s}},
\end{equation}
where $\lVert \cdot \rVert_{\textrm{s}}$ denotes the operator norm of the operator restricted to even functions.
Note that $a_T^0=\sup_{q\in \BR^2}\lVert V^{1/2} B_T(\cdot,q) V^{1/2}\rVert_{\textrm{s}}$ and therefore $E_T(q)\geq 0$.
For $\psi\in L^2(\BR^4)$ let $E_T \psi (r,q)=E_T(q) \psi(r,q)$.
We get the operator inequality $a_T^0 \BI -A_T^0 \geq F_2^\dagger  E_T F_2$, where $\BI$ denotes the identity operator on  $L_{\textrm{s}}^2(\BR^4)$.
Using \eqref{AT2_decomposition}, the above inequality and that $\lVert F_2 \tilde \iota \psi \rVert_2=\lVert \psi \rVert_2$ we obtain
\begin{equation}
a_T^0-a_T^2 \geq \inf_{\psi \in L^2_s(\tilde \Omega_2), \lVert \psi\rVert_2=1} \langle F_2\tilde \iota \psi,( E_T +R_T) F_2 \tilde \iota \psi \rangle
\geq  \inf_{\psi \in L^2(\BR^4), \lVert \psi\rVert_2=1} \langle \psi,(  E_T +R_T) \psi \rangle .
\end{equation}
Therefore, it suffices to show that $\lim_{T\to 0} \inf \sigma ( E_{T}+R_T)\geq 0$.
The proof relies on the following three Lemmas.
\begin{lemma}\label{corner_w_lea3}
Let $\mu>0$ and let $V$ satisfy Assumption~\ref{asptn1}. Then $\sup_{T>0} \lVert R_{T} \rVert<\infty$.
\end{lemma}
\begin{lemma}\label{corner_w_lea4}
Let $\mu>0$ and let $V$ satisfy Assumption~\ref{asptn1}.
Let $\BI_{\leq \epsilon}$ act on $L^2(\BR^{4})$ as $\BI_{\leq \epsilon} \psi(r,q)=\psi(r,q)\chi_{\vert q\vert\leq \epsilon}$.
Then $\lim_{\epsilon \to 0} \sup_{T>0} \lVert \BI_{\leq \epsilon} R_{T} \BI_{\leq \epsilon} \lVert =0$.
\end{lemma}
\begin{lemma}\label{corner_w_lea2}
Let $\mu>0$ and let $V$ satisfy Assumption~\ref{asptn1}.
Let $0<\epsilon<\sqrt{\mu}$.
There are constants $c_1,c_2,T_0>0$ such that for $0<T<T_0$ and $|q|>\eps$ we have $ E_{T}(q)>c_1 \vert \ln(c_2/T)\vert$.
\end{lemma}
The first two Lemmas are extensions of \cite[Lemma 6.1 and Lemma 6.2]{roos_bcs_2023} and proved in Sections~\ref{sec:pf_tcrel_1} and \ref{sec:pf_tcrel_2}, respectively.
The third Lemma is contained in \cite[Lemma 6.3]{roos_bcs_2023}.

With these Lemmas, the claim follows completely analogously to the proof of \cite[Theorem 1.2 (ii)]{hainzl_boundary_2023} and we provide a sketch for completeness.
Using that $E_{T}(q)\geq 0$, we write
\begin{equation}\label{E+G+N+d}
E_{T}+ R_{T}+\delta=\sqrt{E_{T}+\delta}\left(\BI + \frac{1}{\sqrt{E_{T}+\delta}}R_{T}\frac{1}{\sqrt{E_{T}+\delta}}\right)\sqrt{E_{T}+\delta}
\end{equation}
for any $\delta>0$.
It suffices to prove that for all $\delta>0$ the norm of the second term in the bracket vanishes in the limit $T\to0$.
With the notation from Lemma~\ref{corner_w_lea4} we estimate for all $0<\epsilon<\sqrt{\mu}$
\begin{multline}
\left\lVert \frac{1}{\sqrt{E_{T}+\delta}}R_{T}\frac{1}{\sqrt{E_{T}+\delta}}\right\rVert \leq
\left\lVert \BI_{\leq\epsilon} \frac{1}{\sqrt{E_{T}+\delta}}R_{T}\frac{1}{\sqrt{E_{T}+\delta}}\BI_{\leq\epsilon}\right\rVert\\
+\left\lVert \BI_{\leq\epsilon} \frac{1}{\sqrt{E_{T}+\delta}}R_{T}\frac{1}{\sqrt{E_{T}+\delta}}\BI_{>\epsilon}\right\rVert
+\left\lVert \BI_{>\epsilon} \frac{1}{\sqrt{E_{T}+\delta}}R_{T}\frac{1}{\sqrt{E_{T}+\delta}}\right\rVert\,.
\end{multline}
Lemma~\ref{corner_w_lea2} and $E_{T}\geq 0$ imply
\begin{equation}
\lim_{T\to 0}\left\lVert \frac{1}{\sqrt{E_{T}+\delta}}R_{T}\frac{1}{\sqrt{E_{T}+\delta}}\right\rVert \leq \sup_{T>0} \frac{1}{\delta}\left\lVert \BI_{\leq\epsilon} R_{T} \BI_{\leq\epsilon}\right\rVert+\lim_{T\to 0} \frac{2}{(\delta  c_1 \vert \ln(c_2/T)\vert)^{1/2}} \lVert R_{T} \rVert.
\end{equation}
The first term can be made arbitrarily small by Lemma~\ref{corner_w_lea4} and the second term vanishes by Lemma~\ref{corner_w_lea3}.
Hence, Theorem~\ref{thm2} follows.

\section{Proofs of Auxiliary Lemmas}\label{sec:aux}
\subsection{Proof of Lemma~\ref{bdiff}}\label{sec:pfbdiff}
\begin{proof}[Proof of Lemma~\ref{bdiff}]
Using the Mittag-Leffler series (as in \cite[(2.1)]{hainzl_boundary_2023}) one can write
\begin{multline}
f(p,q,x)=2T\sum_{n\in \BZ} \Xi_n^{-1}\Big[(2q_2+ x ) (2\mu  -2q^2-2p^2-x^2+2(p_2-q_2)x)\\
+2p_2(4 p \cdot q- 2 i w_n+2(p_2-q_2)x -x^2)\Big]
\end{multline}
where
\begin{multline}
\Xi_n=\left(\left(p+q+(0,x)\right)^2-\mu-iw_n\right)\left(\left(p-q-(0,x)\right)^2-\mu+iw_n\right)\\
\times \left(\left(p+q\right)^2-\mu-iw_n\right)\left(\left(p-q\right)^2-\mu+iw_n\right)
\end{multline}
and $w_n=(2n+1)\pi T$.
Continuity of $f$ follows from dominated convergence.
For $x>\sqrt{\mu}/4$ the bound on $f$ follows from the bound on $B_T$ in \eqref{BT_bound}.
Let $Q_2=Q_1+\sqrt{\mu}/4$.
For $x<\sqrt{\mu}/4$ we have
\begin{equation}
|f(p,q,x)|\leq \sup_{|q_2|\leq Q_2} |\frac{\p}{\p q_2} B_T(p,q)|=\sup_{|q_2|\leq Q_2}| f(p,q,0)|.
\end{equation}
To bound $| f(p,q,0)|$, first note that for $x=0$ with the notation $y=(p+q)^2-\mu$, $z=(p-q)^2-\mu$ and $v=\max\{\left(|p_1|+|q_1|\right)^2+(|p_2|-|q_2|)^2-\mu,0\} $,
\begin{equation}
|\Xi_n| = \left(y^2+w_n^2\right)\left(z^2+w_n^2\right) \geq  \left(v^2+w_n^2\right)\left(\max\{(|p_2|-|q_2|)^2-\mu,0\})^2+w_n^2\right).
\end{equation}
Furthermore,
\begin{multline}
\sup_{(p,q)\in \BR^{4}, \vert q_2 \vert\leq Q_2}\left \vert \frac{4 i w_n p_2}{\max\{(|p_2|-|q_2|)^2-\mu,0\})^2+w_n^2}\right \vert\\
\leq \sup_{(p,q)\in \BR^{4}, \vert q_2 \vert<Q_2} \frac{ 4\vert p_2\vert }{\sqrt{\max\{(|p_2|-|q_2|)^2-\mu,0\})^2+w_0^2}}=:c_1 <\infty
\end{multline}
There is a constant $c_2>\mu$ depending only on $\mu$ and $Q_2$ such that $|p_2|^2\leq 4(\min\{y,z\}+c_2)$ for $|q_2|\leq Q_2$ and all $p_1,q_1\in \BR$.
One obtains that for $|q_2|\leq Q_2$
\begin{equation}\label{pf_bdiff.1}
\vert f(p,q,0) \vert \leq 2T\sum_{n\in \BZ} \frac{2Q_2|y+z|+4 \sqrt{\min\{y,z\}+c_2} |y-z|}{(y^2+w_n^2)(z^2+w_n^2)}\\
+ 2T\sum_{n\in \BZ}\frac{c_1}{v^2+w_n^2}
\end{equation}
Since the summands are decreasing in $n$, we can estimate the sums by integrals.
The second term is bounded by
\begin{multline}
4T c_1\left[ \frac{1}{v^2 +w_0^2} +\int_{1/2}^\infty  \frac{1}{v^2+4\pi^2T^2 x^2} \dd x \right]
=4T c_1 \left[ \frac{1}{v^2+w_0^2} +\frac{\arctan \left(\frac{v}{\pi T}\right)}{2\pi T v}\right]\\
 <\frac{C}{1+p_1^2+q_1^2+p_2^2}
\end{multline}
for some constant $C$ independent of $p$ and $q_1$, since $\sup_{(p,q)\in \BR^{4}, \vert q_2 \vert\leq Q_2} \frac{1+p_1^2+q_1^2+p_2^2}{1+v}<\infty$.
The first term in \eqref{pf_bdiff.1} is bounded by
\begin{multline}
16T (Q_2+2 \sqrt{\min\{|y|,|z|\}+c_2} )\max\{|y|,|z|\}\Bigg[ \frac{1}{(y^2+w_0^2)(z^2+w_0^2)} \\
+\int_{1/2}^\infty  \frac{1}{(y^2+4\pi^2T^2 x^2)(z^2+4\pi^2T^2 x^2)}\dd x \Bigg]
\end{multline}
Note that $y+z+2\mu+1=1+2p^2+2q^2$.
The claim thus follows if we prove that for $c_3>0$
\begin{equation}
\sup_{y>z>0}(1+y+z)(1+\sqrt{z+1})y\Bigg[ \frac{1}{(y^2+1)(z^2+1)} +\int_{c_3}^\infty  \frac{1}{(y^2+ x^2)(z^2+x^2)}\dd x \Bigg]<\infty
\end{equation}
The supremum over the first summand is obviously finite.
The supremum over the second summand is bounded by
\begin{equation}
\sup_{y>z>0}\frac{(1+2y)y}{y^2+ c_3^2} \frac{1+\sqrt{z+1}}{(z^2+c_3^2)^{1/4}}\int_{c_3}^\infty  \frac{1}{x^{3/2}}\dd x <\infty.
\end{equation}
\end{proof}

\subsection{Proof of Lemma~\ref{lea:BB_asy}}\label{pf_lea_BB_asy}
\begin{proof}[Proof of Lemma~\ref{lea:BB_asy}]
Using the inequality \eqref{BT-ineq} and substituting $p_1 \pm q_1\to p_1, p_1'\pm q_1' \to p_1'$ , we have
\begin{multline}
 \int_{\BR^3} B_{T}(p,q)
B_{T}((p'_1,p_2),q') \dd p_1 \dd p_1' \dd p_2
\leq \frac{1}{4}\int_{\BR^3}( B_{T}((p_1,p_2+q_2),0) +B_{T}((p_1,p_2-q_2),0) ) \\
\times(B_{T}((p'_1,p_2+q_2'),0) +B_{T}((p'_1,p_2-q_2'),0) ) \dd p_1 \dd p_1' \dd p_2
\end{multline}
One can bound this from above by
\begin{multline}
\sup_{q_2,q_2'\in \BR}\int_{\BR}\left(\int_\BR B_{T}((p_1,p_2+q_2),0)\dd p_1 \right) \left(\int_\BR B_{T}((p'_1,p_2+q_2'),0)  \dd p_1' \right)\dd p_2\\
\leq \sup_{q_2\in \BR}\int_{\BR^3}B_{T}((p_1,p_2+q_2),0) B_{T}((p'_1,p_2+q_2),0) \dd p_1 \dd p_1' \dd p_2\\
=\int_{\BR^3}B_{T}((p_1,p_2),0) B_{T}((p'_1,p_2),0) \dd p_1 \dd p_1' \dd p_2
\end{multline}
where in the second step we used the Schwarz inequality in $p_2$.
The latter expression is of order $O(\ln(\mu/T)^3)$ for $T\to 0$, as was shown in the proof of \cite[Lemma 4.10]{roos_bcs_2023}.

To prove the second statement, we shall use that for fixed $0<\delta<\mu$
\begin{equation}\label{bb_asy_halfspace}
\int_{\BR^3}(1-\chi_{\mu-\delta<p_2^2<\mu}\chi_{p_1^2<2\delta}\chi_{p_1'^2<2\delta}) B_{T}(p,0)
B_{T}((p'_1,p_2),0) \dd p_1 \dd p_1' \dd p_2 = O((\ln \mu/T)^2)
\end{equation}
for $T\to 0$ as was shown in the proof of \cite[Lemma 4.10]{roos_bcs_2023}.
We choose $\delta_2$ and $\delta$ small enough, such that for all $q^2<\delta_2$, if $ p_1^2>4\delta_1$ we have $(p_1+q_1)^2>2\delta$ and if $p_2^2<\mu-\delta_1$ or $p_2^2>\mu+\delta_1$ we have $(p_2+q_2)^2<\mu-\delta$ or $(p_2+q_2)^2>\mu$, respectively.
Using the same inequality \eqref{BT-ineq} as above, we have
\begin{multline}\label{pflea4.6.1}
 \sup_{q^2,q'^2<\delta_2}\int_{\BR^3}(1-\chi_{\mu-\delta_1<p_2^2<\mu+\delta_1}\chi_{p_1^2<4\delta_1}\chi_{p_1'^2<4\delta_1}) B_{T}(p,q)
B_{T}((p'_1,p_2),q') \dd p_1 \dd p_1' \dd p_2\\
\leq  \sup_{q^2,q'^2<\delta_2}\int_{\BR^3}(1-\chi_{\mu-\delta_1<p_2^2<\mu+\delta_1}\chi_{p_1^2<4\delta_1}\chi_{p_1'^2<4\delta_1}) B_{T}(p+q,0)
B_{T}((p'_1,p_2)+q',0) \dd p_1 \dd p_1' \dd p_2
\end{multline}
Note that $1-\chi_{\mu-\delta_1<p_2^2<\mu+\delta_1}\chi_{p_1^2<4\delta_1}\chi_{p_1'^2<4\delta_1} \leq \chi_{\mu-\delta_1>p_2^2}+ \chi_{\mu+\delta_1<p_2^2}+ \chi_{p_1^2>4\delta_1}+ \chi_{p_1'^2>4\delta_1}$.
Using the Schwarz inequality in $p_2$ we bound \eqref{pflea4.6.1} above by
\begin{multline}
  \sup_{q^2<\delta_2}\int_{\BR^3}( \chi_{\mu-\delta_1>p_2^2}+ \chi_{\mu+\delta_1<p_2^2}) B_{T}(( p_1+q_1, p_2+q_2),0) B_{T}(( p'_1+q_1, p_2+q_2),0) \dd  p_1 \dd  p_1' \dd  p_2\\
+ 2 \sup_{q^2,q'^2<\delta_2} \Big(\int_{\BR^3} B_T((p_1+q_1,p_2+q_2),0) B_T((p_1'+q_1,p_2+q_2),0)\chi_{p_1^2>4\delta_1}\chi_{p_1'^2>4\delta_1}\dd  p_1 \dd  p_1' \dd  p_2\Big)^{1/2} \\
\times\Big(\int_{\BR^3} B_T((p_1,p_2+q_2),0) B_T((p_1',p_2+q_2),0)\dd  p_1 \dd  p_1' \dd  p_2\Big)^{1/2}
\end{multline}
Substituting $p_j+q_j\to p_j$ and by choice of $\delta_2$ and $\delta$, this is bounded above by
\begin{multline}
 \int_{\BR^3}( \chi_{\mu-\delta>p_2^2}+ \chi_{\mu<p_2^2}) B_{T}(p,0) B_{T}(( p'_1, p_2),0) \dd  p_1 \dd  p_1' \dd  p_2\\
+ 2 \Big(\int_{\BR^3} B_T(p,0) B_T((p_1',p_2),0)\chi_{p_1^2>2\delta}\chi_{p_1'^2>2\delta}\dd  p_1 \dd  p_1' \dd  p_2\Big)^{1/2} \\
\times\Big(\int_{\BR^3} B_T(p,0) B_T((p_1',p_2),0)\dd  p_1 \dd  p_1' \dd  p_2\Big)^{1/2}
\end{multline}
By \eqref{bb_asy_halfspace} and the first part of this Lemma, this is of order $O((\ln \mu/T)^2)+O((\ln \mu/T)(\ln \mu/T)^{3/2})=O((\ln \mu/T)^{5/2})$.
\end{proof}

\subsection{Proof of Lemma~\ref{lea:phiex_norm}}\label{sec:pf_phiex_norm}
\begin{proof}[Proof of Lemma~\ref{lea:phiex_norm}]
For $p_2,q_2\in \BR$ let $B_{T}((\cdot,p_2),(\cdot,q_2)) $ denote the self-adjoint operator on $L^2((-\sqrt{2\mu},\sqrt{2\mu}))$ acting as $\langle \psi,B_{T}((\cdot,p_2),(\cdot,q_2))  \psi\rangle = \int_{-\sqrt{2\mu}}^{\sqrt{2\mu}}\int_{-\sqrt{2\mu}}^{\sqrt{2\mu}}\overline{\psi(p_1)}B_T(p,q) \psi(q_1) \dd p_1\dd q_1$.
Enlarging the domain of integration for $(q_1,p_2)$ from a disk to square we have
\begin{multline}
\lVert B_{T}^{{ ex}, 2}(\xi) \rVert \leq \sup_{\lVert \psi \rVert_2=1} \int_{(-\sqrt{2\mu},\sqrt{2\mu})^4} \overline{\psi(p_1')}B_{T}((p_1',p_2),(q_1,\xi)) B_{T}(p,(q_1,\xi))\psi(p_1) \dd p_1 \dd p_1' \dd q_1 \dd p_2\\
= \sup_{\lVert \psi \rVert_2=1} \int_{-\sqrt{2\mu}}^{\sqrt{2\mu}}\langle \psi, B_{T}((\cdot,p_2),(\cdot,\xi)) ^2\psi \rangle \dd p_2.
\end{multline}
By the triangle inequality,
\begin{equation}
\lVert B_{T}^{{ ex}, 2}(\xi) \rVert\leq  \int_{-\sqrt{2\mu}}^{\sqrt{2\mu}}\lVert B_{T}((\cdot,\xi),(\cdot,p_2)) \rVert^2 \dd p_2.
\end{equation}
For fixed $p_2,q_2$ we derive two bounds on $\lVert B_{T}((\cdot,p_2),(\cdot,q_2)) \rVert^2$.
For the first bound we estimate the Hilbert-Schmidt norm using the bounds on $B_T$ \eqref{BT_bound}:
\begin{multline}
\lVert B_{T}((\cdot,p_2),(\cdot,q_2)) \rVert^2\leq \lVert B_{T,\mu}((\cdot,p_2),(\cdot,q_2)) \rVert_{\rm HS}^2\\
\leq \int_{-\sqrt{2\mu}}^{\sqrt{2\mu}}\int_{-\sqrt{2\mu}}^{\sqrt{2\mu}}\frac{1}{\max\{|p_1^2+q_1^2+p_2^2+q_2^2-\mu|^2,T^2\}} \dd p_1 \dd q_1 \\
\leq 2\pi \int_0^{2\sqrt{\mu}} \frac{r}{\max\{|r^2+p_2^2+q_2^2-\mu|^2,T^2\}} \dd r
\leq \pi \int_\BR \frac{1}{\max\{x^2,T^2\}} \dd x=\frac{4\pi}{T}
\end{multline}
where we first switched to angular coordinates and then substituted $x=r^2+p_2^2+q_2^2-\mu$.

For the second bound the idea is to apply \cite[Lemma 6.5]{roos_bcs_2023}.
For $\mu_1,\mu_2\in \BR$ let $D_{\mu_1,\mu_2}$ be the operator on $L^2(\BR)$ with integral kernel
\begin{equation}\label{Dmu1mu2}
D_{\mu_1,\mu_2}(p_1,q_1)=\frac{2}{|(p_1+q_1)^2-\mu_1|+|(p_1-q_1)^2-\mu_2|}.
\end{equation}
It was shown in \cite[Lemma 4.6]{hainzl_boundary_2023} that
\begin{equation}\label{BT_bound_2}
B_T(p,q)\leq \frac{2}{|(p+q)^2-\mu|+|(p-q)^2-\mu|}.
\end{equation}
In particular, we have $\lVert B_{T}((\cdot,p_2),(\cdot,q_2)) \rVert\leq \lVert D_{\mu-(p_2+q_2)^2,\mu-(p_2-q_2)^2}\rVert$ and
\begin{equation}
\lVert B_T^{{ ex},2}(\xi) \rVert \leq \int_{-\sqrt{2\mu}}^{\sqrt{2\mu}}\min\Big\{ \frac{4\pi}{T},\lVert D_{\mu-(\xi+q_2)^2,\mu-(\xi-q_2)^2}\rVert ^2\Big\} \dd q_2
\end{equation}
According to \cite[Lemma 6.5]{roos_bcs_2023}, for $\mu_1,\mu_2\leq \mu$ there is a constant $C>0$ such that
\begin{equation}\label{Dmu1mu2_bound}
\lVert D_{\mu_1,\mu_2} \rVert \leq C+\frac{C \mu^{1/2}}{|\min\{\mu_1,\mu_2\}|^{1/2}}\Bigg[1+\chi_{\min\{\mu_1,\mu_2\}<0<\max\{\mu_1,\mu_2\} }\ln \left(1+\frac{\max\{\mu_1,\mu_2\}}{\vert \min\{\mu_1,\mu_2\}\vert}\right)\Bigg]
\end{equation}
The condition $\mu-(|q_2|+|\xi|)^2<0<\mu-(|q_2|-|\xi|)^2 $ can only be satisfied for $\sqrt{\mu}-|\xi| \leq |q_2|\leq \sqrt{\mu}+|\xi|$.
We get the bound
\begin{multline}
\sup_{|\xi|<cT} \lVert B_T^{{ ex},2}(\xi)  \rVert \leq C \Bigg(\int_{||q_2|-\sqrt{\mu}|<2cT}\frac{1}{T} \dd q_2\\
+ \sup_{|\xi|<cT}\int_{-\sqrt{2\mu}}^{\sqrt{2\mu}} \chi_{||q_2|-\sqrt{\mu}|>2cT}\left[1+\frac{1}{|\mu-(|q_2|+|\xi|)^2|^{1/2}}\right]^2\dd q_2\Bigg)
\leq \tilde C(1+\ln \mu/T)
\end{multline}
\end{proof}

\subsection{Proof of Lemma~\ref{lea:l3_34}}\label{pf:lea_l3_34}
\begin{proof}[Proof of Lemma~\ref{lea:l3_34}]
Applying the inequality \eqref{BT-ineq}, we have
\begin{multline}
 \sup_{q_1}\int_{\BR} B_{T,\mu}((p_1,\xi(T)),(q_1,\xi'(T)))\dd p_1 \\
\leq \frac{1}{2} \left[\int_{\BR} B_{T,\mu}((p_1,\xi(T)+\xi'(T)),0)\dd p_1 +\int_{\BR} B_{T,\mu}((p_1,\xi(T)-\xi'(T)),0)\dd p_1\right]
\end{multline}
The first integral equals
\begin{equation}
\int_{\BR} B_{T,\mu-(\xi(T)+ \xi'(T))^2}(p_1,0)\dd p_1,
\end{equation}
where here $B_{T,\mu}$ is understood as the function defined through the same expression as $B_T$ in \eqref{BT} on $\BR\times \BR$ instead of $\BR^2\times \BR^2$.
For the second integral replace $ \xi'(T)$ by $- \xi'(T)$.
The claim follows from the asymptotics
\begin{equation}
\int_{\BR} B_{T,\mu}(p_1,0)\dd p_1= \frac{2}{\sqrt{\mu}}(\ln(\mu/T)+O(1))
\end{equation}
for $T/\mu\to0$, see e.g.~\cite[Lemma 3.5]{hainzl_boundary_2023}.
\end{proof}

\subsection{Proof of Lemma~\ref{bb_asy_lowerbd}}\label{pf:lea_5.2}
\begin{proof}[Proof of Lemma~\ref{bb_asy_lowerbd}]
Let $\gamma=\mu(T/\mu)^{\beta/2}$.
By invariance of $B_T(p,q)$ under $(p_j,q_j)\to -(p_j,q_j) $ for $j\in \{1,2\}$, we may assume without loss of generality that $q\in [0,\infty)^2$.
For a lower bound, we restrict the integration to $p_1,p_2>0$, $p_2^2<\mu-\eps^2-\gamma$ and $p_1^2>(\sqrt{\mu}+\eps)^2+T-p_2^2$.
For $p,q\in [0,\infty)^2$ with $|q|<\eps$ and $p^2>(\sqrt{\mu}+\eps)^2+T$, we have $(p-q)^2-\mu\geq ||p|-|q||^2-\mu \geq 0$ and $(p+q)^2-\mu\geq p^2+q^2-\mu \geq T$.
Therefore, in this regime
\begin{equation}
B_T(p,q)\geq \frac{1}{2}\frac{\tanh(1/2)}{p^2+q^2-\mu}.
\end{equation}
This is minimal if $|q|=\eps$.
Since  for $a>b>0$
\[
\int_a^\infty \frac{1}{p_1^2-b^2} \dd p_1=\frac{\artanh(b/a)}{b}=\frac{1}{b}\artanh\left(\sqrt{1-(a^2-b^2)/a^2}\right),
\]
the left hand side of \eqref{jMj.2} is bounded below by
\begin{equation}
\frac{\tanh(1/2)^2}{4}\int_{\sqrt{\mu-\delta}}^{\sqrt{\mu-\eps^2-\gamma}} \frac{\artanh\left(\sqrt{1-\frac{2\sqrt{\mu}\eps+2\eps^2+T}{(\sqrt{\mu}+\eps)^2+T-p_2^2}}\right)^2}{\mu-\eps^2-p_2^2}\dd p_2
\end{equation}
By monotonicity of $\artanh$, the $\artanh$ term in the integrand is minimal for $p_2=\sqrt{\mu-\eps^2-\gamma}$.
Since $\int_{\sqrt{\mu-\delta}}^{\sqrt{\mu-\eps^2-\gamma}} \frac{1}{\mu-\eps^2-p_2^2}\dd p_2=\frac{1}{\mu-\eps^2}(\artanh(\sqrt{1-(\eps^2+\gamma)/\mu})-\artanh(\sqrt{1-\delta/\mu}))$, the left hand side of \eqref{jMj.2} is bounded below by
\begin{equation}
\frac{\tanh(1/2)^2}{4(\mu-\eps^2)} \artanh\left(\sqrt{1-\frac{2\sqrt{\mu}\eps+2\eps^2+T}{2\sqrt{\mu}\eps+2\eps^2+T+\gamma}}\right)^2\Bigg[\artanh\left(\sqrt{1-\frac{\eps^2+\gamma}{\mu}}\right)-\artanh\left(\sqrt{1-\frac{\delta}{\mu}}\right)\Bigg]
\end{equation}
With $\artanh(\sqrt{1-x})=\frac{1}{2}\ln (4/x)+o(1)$ as $x\to0$, we have for $T\to0$
\begin{equation}
\artanh\left(\sqrt{1-\frac{2\sqrt{\mu}\eps+2\eps^2+T}{2\sqrt{\mu}\eps+2\eps^2+T+\gamma}}\right)=\frac{\beta}{4}\ln(\mu/T)+O(1)
\end{equation}
and
\begin{equation}
\artanh\left(\sqrt{1-\frac{\eps^2+\gamma}{\mu}}\right)=\frac{\beta}{4}\ln(\mu/T)+O(1).
\end{equation}
Hence, the left hand side of \eqref{jMj.2} is bounded below by $\frac{\tanh(1/2)^2}{4^3}\frac{ \beta^3}{\mu}(\ln \mu/T)^3+O(\ln\mu/T)^2$, and the claim follows.
\end{proof}

\subsection{Proof of Lemma~\ref{corner_w_lea3}}\label{sec:pf_tcrel_1}
\begin{proof}[Proof of Lemma~\ref{corner_w_lea3}]
According to \cite[Lemma 6.1]{roos_bcs_2023}, $\sup_T\lVert G_T^j\rVert<\infty$ for $j\in\{1,2\}$ and it suffices to prove $\sup_T\lVert N_T\rVert<\infty$.
We have $\lVert N_T \rVert \leq \lVert N_T^<\rVert + \lVert N_T^>\rVert$, where
\begin{equation}
\langle \psi, N_{T}^<\psi \rangle=\int_{\BR^{4}} \overline{F_1 V^{1/2}\psi(q,p)} B_{T}(p,q)\chi_{p^2,q^2<2\mu} F_1 V^{1/2}\psi(p,q) \dd p \dd q \label{nt<}
\end{equation}
and for $N_T^>$ replace the characteristic function by $1-\chi_{p^2,q^2<2\mu}$.

To bound $\lVert N_T^> \rVert$, we first use the Schwarz inequality to obtain
\begin{equation}
\lVert N_T^> \rVert \leq \sup_{\psi \in L^2(\BR^4), \lVert \psi\rVert_2=1} \int_{\BR^{4}}B_{T}(p,q)(1-\chi_{p^2,q^2<2\mu}) |F_1 V^{1/2}\psi(p,q)|^2 \dd p \dd q
\end{equation}
By the bound on $B_T$ in \eqref{BT_bound}, there is a constant $C>0$ independent of $T$ such that $\lVert N_T^> \rVert \leq C \lVert M \rVert$, where $M:=V^{1/2} \frac{1}{1-\Delta} V^{1/2}$ on $L^2(\BR^2)$.
The Young and H\"older inequalities imply that $M$ is a bounded operator \cite{lieb_analysis_2001}.

To bound $\lVert N_T^< \rVert$, we use that $\lVert F_1 V^{1/2} \psi( \cdot,q)\rVert_\infty \leq \lVert V\rVert_1^{1/2} \lVert \psi(\cdot,q)\rVert_2$ by the Schwarz inequality and the upper bound for $B_T$ in \eqref{BT_bound_2} to obtain
\begin{equation}
\langle \psi, N_T^< \psi \rangle \leq 2 \lVert V \rVert_1 \int_{\BR^4} \frac{\lVert \psi(\cdot,q)\rVert_2 \lVert \psi(\cdot,p)\rVert_2 }{|(p+q)^2-\mu|+|(p-q)^2-\mu|}\chi_{p^2,q^2<2\mu} \dd p \dd q
\end{equation}
Recalling the definition of the operator $D_{\mu_1,\mu_2}$ from \eqref{Dmu1mu2}, this is further bounded by
\begin{equation}
2 \lVert V \rVert_1 \int_{\BR^2}\lVert \psi(\cdot,(\cdot, q_2))\rVert_2 \lVert D_{\mu-(p_2+q_2)^2,\mu-(p_2-q_2)^2} \rVert \lVert \psi(\cdot,(\cdot,p_2))\rVert_2 \chi_{p_2^2,q_2^2<2\mu} \dd p \dd q
\end{equation}
It follows from the bound on $\lVert D_{\mu_1,\mu_2}\rVert$ in \eqref{Dmu1mu2_bound} that for any $\alpha>0$ there is a constant $C_\alpha$ independent of $p_2,q_2$ such that
\[
 \lVert D_{\mu-(p_2+q_2)^2,\mu-(p_2-q_2)^2} \rVert \leq C_\alpha (1+|\mu-(|p_2|+|q_2|)^2|^{-1/2-\alpha}).
\]
Let $\tilde D_\alpha$ denote the operator on $L^2((-\sqrt{2\mu},\sqrt{2\mu}))$ with integral kernel
\[
\tilde D_\alpha(q_2,p_2)=(1+|\mu-(|p_2|+|q_2|)^2|^{-1/2-\alpha}).
\]
Then we have $\lVert N_T^<\rVert \leq 2 C_\alpha \lVert V\rVert_1 \lVert \tilde D_\alpha\rVert$ and it remains to prove that $\lVert \tilde D_\alpha\rVert<\infty$ for a suitable choice of $\alpha$.
Applying the Schur test with constant test function gives
\begin{equation}
\lVert \tilde D_\alpha\rVert\leq \sup_{|q_2|<\sqrt{2\mu}} \int_{-\sqrt{2\mu}}^{\sqrt{2\mu}}(1+|\mu-(|p_2|+|q_2|)^2|^{-1/2-\alpha}) \dd p_2,
\end{equation}
which is finite for $\alpha<1/2$.
\end{proof}

\subsection{Proof of Lemma~\ref{corner_w_lea4}}\label{sec:pf_tcrel_2}
\begin{proof}[Proof of Lemma~\ref{corner_w_lea4}]
It was shown in \cite[Lemma 6.2]{roos_bcs_2023} that $\lim_{\epsilon \to 0} \sup_{T>0} \lVert \BI_{\leq \epsilon} G_{T}^j \BI_{\leq \epsilon} \lVert =0$ for $j\in\{1,2\}$ and it remains to prove $\lim_{\epsilon \to 0} \sup_{T>0} \lVert \BI_{\leq \epsilon} N_{T} \BI_{\leq \epsilon} \lVert =0$.
We use the Schwarz inequality twice to bound
\begin{multline}
 \lVert \BI_{\leq \epsilon} N_{T} \BI_{\leq \epsilon} \lVert  \leq  \lVert V\rVert_1\sup_{\psi \in L^2(\BR^4), \lVert \psi\rVert_2=1} \int_{\BR^4}  \lVert \psi(\cdot,p)\rVert_2 B_T(p,q) \chi_{\vert p\vert,\vert q \vert\leq \epsilon} \lVert \psi(\cdot,q)\rVert_2 \dd p \dd q\\
\leq \lVert V\rVert_1 \sup_{\psi \in L^2(\BR^4), \lVert \psi\rVert_2=1} \int_{\BR^4}  B_T(p,q) \chi_{\vert p\vert,\vert q \vert\leq \epsilon} \lVert \psi(\cdot,q)\rVert_2^2 \dd p \dd q
\leq  \lVert V\rVert_1 \sup_{|q|\leq \eps} \int_{\vert p\vert\leq \epsilon}  B_T(p,q)\dd p.
\end{multline}
Applying the bound on $B_T$ in \eqref{BT_bound}, for $\eps<\sqrt{\mu/2}$ one can bound the right hand side uniformly in $T$ by
\begin{equation}
 \lVert V\rVert_1 \sup_{|q|\leq \eps}\int_{\vert p\vert\leq \epsilon}  \frac{1}{\mu-p^2-q^2} \dd p,
\end{equation}
which vanishes as $\eps\to0$.
The claim follows.
\end{proof}

\section*{Acknowledgments}
Financial support by the Austrian Science Fund (FWF) through project number I 6427-N (as part of the SFB/TRR 352) is gratefully acknowledged.

\appendix

\bibliographystyle{abbrv}

\end{document}